\documentclass{sigplanconf}

\usepackage[usenames,dvipsnames]{color}

\usepackage{mathpartir}
\usepackage{amsthm}
\usepackage{stmaryrd}
\usepackage{amssymb}
\usepackage{amsmath}

\usepackage{listings}

\usepackage{pstricks, pst-node, pst-plot, pst-tree}

\usepackage{booktabs}

\usepackage[colorinlistoftodos]{todonotes}

\newcommand{\webpageGoogleBenchmarks}{\texttt{http://v8.googlecode.com/svn/data/benchmarks/v7/run.html}}

\newtheorem{definition}{Definition}
\newtheorem{lemma}{Lemma}
\newtheorem{theorem}{Theorem}

\newcommand{\Rule}[1]{\textsc{#1}}

\newcommand{\Label}[1]{\textit{#1}}

\newcommand{\subseteqIH}{\overset{\textsf{IH}}{\subseteq}}
\newcommand{\eqIH}{\overset{\textsf{IH}}{=}}

\lstdefinelanguage{JavaScript}{
   keywords={      attributes, class, classend, do, empty, endif, endwhile, fail,
	  function, functionend, if, implements, in, inherit, inout, not, of,
	  operations, out, return, set, then, types, while, use, else, switch, case,
   break, default, for, var},
   keywordstyle=\color{blue}\bfseries,
   ndkeywords={trace, permit, apply, applyObj, permitArgs, applyArgs},
   ndkeywordstyle=\color{blue}\bfseries,
   identifierstyle=\color{black},
   sensitive=false,
   comment=[l]{//},
   morecomment = [s]{/*}{*/},
   morecomment = [s][\color{green}]{/**}{*/},
   commentstyle=\color{gray}\ttfamily,
   stringstyle=\color{red}\ttfamily
}
\newcommand{\eval}{\Downarrow}

\newcommand{\entails}{\vdash}
\newcommand{\entailsFA}{\vdash_{\textsf{App}}}
\newcommand{\entailsPR}{\vdash_{\textsf{Get}}}
\newcommand{\entailsPA}{\vdash_{\textsf{Put}}}

\newcommand{\dom}{\textit{dom}}

\newcommand{\TpTrie}{PathTrie}

\newcommand{\ApcContract}{{\it Contract\/}}
\newcommand{\ApcLiteral}{{\it Literal}}

\newcommand{\tpPath}{\mathcal{P}}
\newcommand{\tpTrie}{\mathcal{T}}

\newcommand{\tpEmpty}{\epsilon}
\newcommand{\tpBlank}{\iota}
\newcommand{\tpProperty}{p}

\newcommand{\apcContract}{\mathcal{C}}
\newcommand{\apcLiteral}{c}
\renewcommand{\apcLiteral}{\ell}

\newcommand{\apcEmptySet}{\emptyset}
\newcommand{\apcEmpty}{\epsilon}
\renewcommand{\apcEmpty}{\mathcal{E}}
\newcommand{\apcAT}{{\tt @}}
\newcommand{\apcStar}{{\tt *}}
\newcommand{\apcQMark}{{\tt ?}}
\newcommand{\apcOr}{{{\tt +}}}
\newcommand{\apcAnd}{\&}
\newcommand{\apcNeg}{!}
\newcommand{\apcDot}{.}
\newcommand{\apcRegEx}{r} 
\newcommand{\tpDot}{.}

\newcommand{\lang}[1]{\mathcal{L}\llbracket#1\rrbracket}
\newcommand{\leftquotient}[2]{#1^{-1}#2}

\newcommand{\apcAlphabet}{\mathcal{A}}
\newcommand{\apcWords}{\mathcal{A}^*}

\newcommand{\append}[2]{(#1 \oplus #2)}
\newcommand{\union}[2]{(#1 \uplus #2)}

\newcommand{\match}[2]{#1 \succcurlyeq #2}

\newcommand{\regmatch}[2]{#1 \succ #2}

\newcommand{\nullable}{\nu}
\newcommand{\isNullable}[1]{\nullable(#1)}

\newcommand{\isReadable}[2]{#1 \vdash_{\mathcal{R}} #2}
\newcommand{\isNotReadable}[2]{#1 \not\vdash_{\mathcal{R}} #2}
\newcommand{\isWrireable}[2]{#1 \vdash_{\mathcal{W}} #2}
\newcommand{\isNotWrireable}[2]{#1 \not\vdash_{\mathcal{W}} #2}

\newcommand{\deriv}[2]{\partial_{#1}(#2)}

\newcommand{\lderiv}[2]{\nabla_{#1}(#2)}%{(\partial_{#1}~#2)}
\newcommand{\blankReducible}{\textsf{bl}}
\newcommand{\isBlank}[1]{\blankReducible(#1)}

\newcommand{\emptyReducible}{\textsf{emp}}%{\textsf{empty}}
\newcommand{\isEmpty}[1]{\emptyReducible(#1)}

\newcommand{\indifferentReducible}{\textsf{ind}}
\newcommand{\isIndifferent}[1]{\indifferentReducible(#1)}

\newcommand{\universalReducible}{\textsf{unv}}%{\textsf{univ}}
\newcommand{\isUniversal}[1]{\universalReducible(#1)}

\newcommand{\firstP}{\textsf{next}}
\newcommand{\firstC}{\textsf{first}}
\newcommand{\getFirstP}[1]{\firstP(#1)}
\newcommand{\getFirstC}[1]{\firstC(#1)}

\newcommand{\ccContext}{\Gamma}%{\Sigma\Phi}%{\Delta}
\newcommand{\ccExp}{\phi}
\newcommand{\inCcContext}[1]{#1\in\ccContext} %{\ccContext(#1)}
\newcommand{\negInCcContext}[1]{#1\notin\ccContext} %{\ccContext(#1)}
\newcommand{\reduce}[1]{\lfloor#1\rfloor}

\newcommand{\isSuperSetOf}[2]{#1 \sqsupseteq #2}
\newcommand{\isSubSetOf}[2]{#1 \sqsubseteq #2}

\newcommand{\isNotSubSetOf}[2]{#1 \notsubsetof #2}

\newcommand{\notsubsetof}{\not\sqsubseteq}

\newcommand{\norm}[2]{#1 \leadsto #2}

\newcommand{\regexSubsetof}{\sqsubseteq_{\apcRegEx}}
\newcommand{\regexCap}{\sqcap_{\apcRegEx}}

\newcommand{\regexIsSubSetOf}[2]{#1 \regexSubsetof #2}
\newcommand{\regexIsNotSubSetOf}[2]{#1 \not\regexSubsetof #2}
\newcommand{\regexIsCap}[2]{#1 \regexCap #2}

\newcommand{\monitor}{\mathcal{M}}

\newcommand{\addReadPath}[1]{\monitor\vartriangleleft_{\mathcal{R}}#1}
\newcommand{\addWritePath}[1]{\monitor\vartriangleleft_{\mathcal{W}}#1}
\newcommand{\addReadViolation}[2]{\monitor\blacktriangleleft_{\mathcal{R}}(#1,#2)}
\newcommand{\addWriteViolation}[2]{\monitor\blacktriangleleft_{\mathcal{W}}(#1,#2)}

\newcommand{\lcpj}{\lambda_{J}}

\newcommand{\ljNew}{\textbf{new}}

\newcommand{\ljIn}{\textbf{in}}

\newcommand{\ljStr}{\textit{str}}

\newcommand{\ljUndefined}{\textbf{undefined}}
\newcommand{\ljNull}{\textbf{null}}

\newcommand{\ljPermit}{\textbf{permit}}

\newcommand{\LjValue}{Value}
\newcommand{\LjExpression}{Expression}

\newcommand{\LjMonitor}{Monitor}

\newcommand{\LjPrototype}{Prototype}

\newcommand{\LjClosure}{Closure}
\newcommand{\LjLocation}{Location}
\newcommand{\LjObject}{Object}
\newcommand{\LjProxy}{Proxy}
\newcommand{\LjHandler}{Access~Handler}
\newcommand{\LjStorable}{Storable}
\newcommand{\LjEnvironment}{Environment}
\newcommand{\LjHeap}{Heap}

\newcommand{\ljObj}{o}
\newcommand{\ljVal}{v}
\newcommand{\ljConst}{c}
\newcommand{\ljVar}{x}

\newcommand{\ljProxy}{P}
\newcommand{\ljHandler}{H}

\newcommand{\ljFunction}{\lambda\ljVar.\ljExp}
\newcommand{\ljClosure}{f}
\newcommand{\ljStorable}{s}
\newcommand{\ljPrototype}{\pi}         % p conflicts with property names

\newcommand{\ljExp}{e}

\newcommand{\ljLocation}{\xi}

\newcommand{\ljHeap}{\mathcal{H}}
\newcommand{\ljEnv}{\rho}

\newcommand{\RuleCCIdentity}{(C-Identity)}

\newcommand{\RuleCCNullable}{(C-Nullable)}

\newcommand{\RuleCCEmpty}{(C-Proof-Edge)}

\newcommand{\RuleCCBlank}{(C-Disprove-Empty)}

\newcommand{\RuleCCIndifferentTwo}{(C-Disprove-Blank)}

\newcommand{\RuleCCContext}{(C-Delete)}
\newcommand{\RuleCCUnfoldTrue}{(C-Unfold-True)}
\newcommand{\RuleCCUnfoldFalse}{(C-Unfold-False)}
\newcommand{\RuleCCDisprove}{(C-Disprove)}

\newcommand{\RuleLjConstant}{(Const)}
\newcommand{\RuleLjVariable}{(Var)}

\newcommand{\RuleLjObjectCreation}{(New)}
\newcommand{\RuleLjFunctionCreation}{(Abs)}
\newcommand{\RuleLjFunctionApplication}{(App)}
\newcommand{\RuleLjPropertyReference}{(Get)}
\newcommand{\RuleLjPropertyAssignment}{(Put)}

\newcommand{\RuleLjPermit}{(Permit)}

\newcommand{\RuleLjFunctionApplicationOne}{(App-NoProxy)}
\newcommand{\RuleLjFunctionApplicationTwo}{(App-Proxy)}
\newcommand{\RuleLjFunctionApplicationThree}{(App-Membrane)}

\newcommand{\RuleLjPropertyReferenceOne}{(Get-NoProxy)}
\newcommand{\RuleLjPropertyReferenceTwo}{(Get-Proxy)}
\newcommand{\RuleLjPropertyReferenceThree}{(Get-Membrane)}

\newcommand{\RuleLjPropertyAssignmentOne}{(Put-NoProxy)}
\newcommand{\RuleLjPropertyAssignmentTwo}{(Put-Proxy)}

\newcommand{\RuleLjOPropertyReferenceTwo}{(Get-Proxy-Observer)}
\newcommand{\RuleLjOPropertyReferenceThree}{(Get-Membrane-Observer)}
\newcommand{\RuleLjOPropertyAssignmentTwo}{Put-Proxy-Observer)}

\newcommand{\RuleLjPPropertyReferenceTwo}{(Get-Proxy-Protector)}
\newcommand{\RuleLjPPropertyAssignmentTwo}{Put-Proxy-Protector)}

\newcommand{\RuleLjTriePropertyReferenceTwo}{(Get-TrieProxy)}
\newcommand{\RuleLjTriePropertyReferenceThree}{(Get-TrieMembrane-NonExisting)}
\newcommand{\RuleLjTriePropertyReferenceFour}{(Get-TrieMembrane-Existing)}
\newcommand{\RuleLjTriePropertyAssignmentTwo}{Put-TrieProxy)}

\usepackage{amsmath}

\begin{document}

\setlength{\pdfpageheight}{\paperheight}
\setlength{\pdfpagewidth}{\paperwidth}

\conferenceinfo{DLS~'13}{October 28, 2013, Indianapolis, Indiana, USA}
\copyrightyear{2013} 
\copyrightdata{978-1-4503-2433-5/13/10}
\doi{2508168.2508176} 

% Uncomment one of the following two, if you are not going for the 
% traditional copyright transfer agreement.

%\exclusivelicense                % ACM gets exclusive license to publish, 
								  % you retain copyright

%\permissiontopublish             % ACM gets nonexclusive license to publish
								  % (paid open-access papers, 
								  % short abstracts)

\titlebanner{}        % These are ignored unless
\preprintfooter{Efficient Access Analysis Using JavaScript Proxies}   % 'preprint' option specified.

\newcommand{\trnote}{\titlenote{
   This report is a slightly edited versions of the paper appeared in the \emph{Proceedings of the 9th symposium on Dynamic languages}. 
   To avoid confusions we revised the notation of access permission contracts. Further, we split the theorem of \emph{Syntactic derivative of contracts}.}}

\title{Efficient Dynamic Access Analysis Using JavaScript Proxies}
\subtitle{Technical Report\trnote}

\authorinfo{Matthias Keil \and Peter Thiemann}
{Institute for Computer Science\\ University of Freiburg\\ Freiburg, Germany}
{\{keilr,thiemann\}@informatik.uni-freiburg.de}

\maketitle

\begin{abstract}
%          _         _                  _   
%    /\   | |       | |                | |  
%   /  \  | |__  ___| |_ _ __ __ _  ___| |_ 
%  / /\ \ | '_ \/ __| __| '__/ _` |/ __| __|
% / ____ \| |_) \__ \ |_| | | (_| | (__| |_ 
%/_/    \_\_.__/|___/\__|_|  \__,_|\___|\__|

   JSConTest introduced the notions of effect monitoring and dynamic
   effect inference for JavaScript. It enables the description of effects
   with path specifications resembling regular expressions. It is
   implemented by an offline source code transformation.

   To overcome the limitations of the JSConTest implementation, we
   redesigned and reimplemented effect monitoring by taking advantange of
   JavaScript proxies. Our new design avoids all drawbacks of the prior
   implementation. It guarantees full interposition; it is not restricted
   to a subset of JavaScript; it is self-maintaining; and its scalability
   to large programs is significantly better than with JSConTest.

   The improved scalability has two sources. First, the reimplementation is
   significantly faster than the original, transformation-based
   implementation. Second, the reimplementation relies on the fly-weight
   pattern and on trace reduction to conserve memory. Only the
   combination of these techniques enables monitoring and inference for
   large programs.
\end{abstract}

\category{D.3.3}{PROGRAMMING LANGUAGES}{Language Constructs and Features}[Classes and objects]
\category{D.3.3}{SOFTWARE ENGINEERING}{Software/Program Verification}[Programming by contract,Validation]
\category{D.4.6 }{OPERATING SYSTEMS}{Security and Protection}[Access controls]

% general terms are not compulsory anymore, 
% you may leave them out
\terms
Design, Languages, Security, Verification

\keywords
Access Permission Contracts, JavaScript, Proxies

% _____       _                 _            _   _             
%|_   _|     | |               | |          | | (_)            
%  | |  _ __ | |_ _ __ ___   __| |_   _  ___| |_ _  ___  _ __  
%  | | | '_ \| __| '__/ _ \ / _` | | | |/ __| __| |/ _ \| '_ \ 
% _| |_| | | | |_| | | (_) | (_| | |_| | (__| |_| | (_) | | | |
%|_____|_| |_|\__|_|  \___/ \__,_|\__,_|\___|\__|_|\___/|_| |_|

\section{Introduction}
\label{sec:introduction}

JSConTest \cite{Heidegger:2010:CTJ:1894386.1894395} introduced the notions of effect monitoring and dynamic
effect inference \cite{Heidegger:2011:HAC:2025896.2025908} for
JavaScript. It enables the programmer to specify the effect of a function
using access permission contracts. These contracts consist of an anchor
specifying a start object and a regular expression specifying the
admissible access paths that a contract-annotated function may use. Matching
paths can be assigned read or write permission.

The inference component of JSConTest may help a software maintainer who
wants to investigate the effect of an unfamiliar function by monitoring its execution and
then summarizing the observed access traces to access permission contracts.

JSConTest is implemented by an offline source code transformation. This
approach enabled a quick development, but it comes with a number of
drawbacks. First, it requires a lot of effort to construct an offline
transformation that guarantees full
interposition and that covers the full JavaScript language: the
implemented transformation has known omissions (e.g., no support for
\texttt{with} and prototypes) and it
does not apply to code created at run time using \texttt{eval}
or other mechanisms. 
Second, the transformation is subject to bitrotting because it becomes
obsolete as the language evolves.
Third, the implementation represents access paths with strings and checks
them against the specification using the built-in regular expression matching
facilities of JavaScript. This approach quickly fills up memory with many
large strings and processes the matching of regular expressions in a
monolithic way. 

In this work, we present JSConTest2, a redesign and reimplementation of JSConTest using JavaScript proxies
\cite{VanCutsem:2010:PDP:1869631.1869638, van2012design}, a JavaScript extension which is
scheduled for the upcoming ECMAScript 6 standard. This new implementation
addresses all shortcomings of the previous version.
First, the proxy-based
implementation guarantees full interposition for the full language and for all
code regardless of its origin, including dynamically loaded code and
code injected via \texttt{eval}. 
Second, maintenance is alleviated because there is no transformation that needs
to be adapted to changes in the language syntax. Also, future extensions are
catered for as long as the proxy API is supported. By adapting ideas from code
contracts \cite{FaehndrichBarnettLogozzo2010}, we also avoided a custom syntax extension.
Third, our new implementation represents access paths in a space efficient
way. It also incrementalizes the path matching by encoding its state in an
automaton state, which is represented by a regular expression. It applies the
fly-weight pattern to reduce memory consumption of the states.
Last but not least, the new implementation is significantly faster
than the previous one.

JSConTest2 employs Brzozowski's derivatives of regular expressions
\cite{Brzozowski1964} to
perform the path matching incrementally and efficiently. It applies a rewriting system inspired
by Antimirov's techniques \cite{Antimirov95:0} for deciding subset constraints
for regular expressions to simplify regular expressions if more than
one contract is applied to an object at the same time. 

To evaluate the scalability of JSConTest2, we applied path monitoring
to a number of example programs including web page dumps. The main problem we
had to deal with was excessive memory use. We explain several techniques for
reducing memory consumption, including the reduction of regular expression based
effect contracts using an adaptation of Antimirov's techniques.

% The presented framework provides an easy applicable solution which can be
% added to existing software projects.  

%  ___         _       _ _         _   _          
% / __|___ _ _| |_ _ _(_) |__ _  _| |_(_)___ _ _  
%| (__/ _ \ ' \  _| '_| | '_ \ || |  _| / _ \ ' \ 
% \___\___/_||_\__|_| |_|_.__/\_,_|\__|_\___/_||_|

\paragraph{Contributions}
\label{sec:contribution}
\begin{itemize}
   \item Reimplementation of JSConTest using JavaScript proxies
   \item Formalization of violation logging and contract enforcement
   \item Reduced memory use by simplification of regular expressions
   \item Practical evaluation with case studies
\end{itemize}

%  ___                  _            
% / _ \__ _____ _ ___ _(_)_____ __ __
%| (_) \ V / -_) '_\ V / / -_) V  V /
% \___/ \_/\___|_|  \_/|_\___|\_/\_/ 

\paragraph{Overview}
\label{sec:overview}

Section~\ref{sec:effects_for_js} gives some examples and rationales for JSConTest2.
Section~\ref{sec:preliminaries} gives a high-level overview of the
approach taken in this paper.
Section~\ref{sec:proxies-membranes} recalls proxies and membranes from related
work. Section~\ref{sec:access-permission-contracts} defines the syntax of access
paths and access contracts. Section~\ref{sec:formalization} formalizes a core
language and defines the semantics for path logging and contract
enforcement. Section~\ref{sec:reduction} explains the techniques used to reduce
memory consumption. Sections~\ref{sec:implementation}
and~\ref{sec:practical_results} describe the implementation and some experiences
in applying JSConTest2. Section~\ref{sec:related_work} discusses related
work. It is followed by a conclusion. 

Appendix \ref{sec:crashing_rules} and \ref{sec:extended_membrane} shows the formal semantics for violations and merged proxies. Section \ref{sec:auxiliaryfunctions} states some auxiliary functions used. The proofs of semantic containment, syntactic, derivative, syntactic containment, and correctness are shown in the appendix \ref{sec:proof_semantic-containment}, \ref{sec:proof_syntactic-derivative}, \ref{sec:proof_syntactic-containment}, and \ref{sec:proof_correctness}.

%A technical report \cite{proxy2013:TechnicalReport} is available with additional technical details including the proofs.

% ___  __  __        _         __              _               ___         _      _   
%| __|/ _|/ _|___ __| |_ ___  / _|___ _ _   _ | |__ ___ ____ _/ __| __ _ _(_)_ __| |_ 
%| _||  _|  _/ -_) _|  _(_-< |  _/ _ \ '_| | || / _` \ V / _` \__ \/ _| '_| | '_ \  _|
%|___|_| |_| \___\__|\__/__/ |_| \___/_|    \__/\__,_|\_/\__,_|___/\__|_| |_| .__/\__|
%                                                                           |_|       

\section{Effects for JavaScript}
\label{sec:effects_for_js}

JavaScript is the language of the Web. More than 90\% of all Web pages provide
functionality using JavaScript. Most of them rely on third-party libraries for
calenders, social networks, or feature extensions. Some of these libraries are
statically included with the main script, others are loaded dynamically.

Software development and maintenance is tricky in JavaScript because
dynamically loaded libraries have arbitrary access to the application state.
Some libraries override global objects to add features, others manipulate
data stored in the browser's DOM or in cookies, yet others may send data to the
net. In addition, there are security concerns if the application has to
guarantee confidentiality or integrity of data. As all scripts run with
the same authority, the main script has no handle on the use of data by an
included script.

As all resources in a JavaScript program are
accessible via property read and write operations, controlling those
operations is sufficient to control the resources. Thus, effect 
monitoring and inference have a role to play in the context of test-driven
development, in maintenance to analyze a piece of software, or in security
to prevent the software from compromising confidentiality or integrity.

JSConTest2 monitors read and write operations on objects through access
permission contracts that specify allowed effects as outlined in the introduction.
A contract restricts effects by defining a set of permitted access paths starting
from some anchor object.

\subsection{Contracts and the Contract API}
\label{sec:design_principles}

This section introduces the contract syntax and the JSConTest2 API.
In a first example, 
a developer may want to ascertain that only some parts of an object are
accessed. 

\begin{lstlisting}
var protected =
  __APC.permit(|'(a.?+b*)'|, {a:{a:3,b:5},b:{a:7,b:11}});
\end{lstlisting}

Here, \lstinline'__APC' is the object that encapsulates the JSConTest2
implementation. Its \lstinline'permit' method takes a contract and an object as
parameters and returns a ``contracted'' object where only access paths that are explicitly permitted by
the contract are admitted. The contract consists of two alternative parts connected by
\lstinline$|+|$. The first part, \lstinline$|'a.?'|$, gives read/write access to
all properties of the object in the \lstinline$a$ property, but \lstinline$a$
itself is read-only.  The second part, \lstinline$|'b*'|$, allows read and write
access to an arbitrarily long chain of properties named \lstinline$b$.

Here is an example with some uses of the contracted object.

\begin{lstlisting}
var x = __APC.permit(|'a.b'|, {a:{b:3}, b:{b:5}});
y = x.a;
y.b = 3;
\end{lstlisting}

The access permission contract  \lstinline$|'a.b'|$ specifies the singleton set $\{a.b\}$ of
permitted access paths. The contract allows us to read and write property \lstinline$a.b$
and to read the prefix \lstinline$a$. Properties which are not addressed by a contract are
neither readable nor writeable. 
%
% Contracts are applied to objects. In the example above, the object \lstinline${a:{b:3}, b:{b:5}}$ is augmented with contract \lstinline{|'a.b'|}. After reading \lstinline$x.a$ the assigned object in \lstinline$y$ is \lstinline${b:3}$ augmented with the remaining contract \lstinline$|'b'|$--permitting read and write access to property \lstinline$b$.  
%
% The build-in function \lstinline$permit$ attaches the contract, given as first argument, to the committed object. 
%
The read and write operations in lines~2 and~3 abide by the contract, but reading from
\lstinline$x.b$ or writing to \lstinline$x.a$ would not be permitted and would cause a violation.

Only the last property of a path in the set of permitted access paths is writeable and all
prefixes are readable. The special property \lstinline$@$ stands for a ``blank'' property
that matches no other property. Using it at the end of a contract specifies a read-only path
as shown in the following example.

\begin{lstlisting}
var x = __APC.permit(|'a.b.@'|, {a:{b:3}, b:{b:5}});
x.a.b = 3; // violation
\end{lstlisting}

% As \lstinline$|'a.b.@'|$ sets \lstinline$x.a.b$ to read-only the property assignment
% \lstinline$x.a.b = 3$ raises a violation.

One could imagine contracts for defining write-only paths, for instance, in a security
context. This case is not covered by our implementation, but it would
be straightforward to provide 
an interface that separates read and write permissions.

The next example demonstrates how contracts interact with assignments. 

\begin{lstlisting}
var x = __APC.permit(|'((a+a.b)+b.b.@)'|, {a:{b:3}, b:{b:5}});
x.a = x.b;
x.a.b=7; // violation
\end{lstlisting}

The contract \lstinline$|'((a+a.b)+b.b.@)'|$ allows read access to \lstinline$x.b$ and
\lstinline$x.b.b$ as well as read and write access to \lstinline$x.a$ and
\lstinline$x.a.b$. Reading \lstinline$x.b$ yields a contracted object
\lstinline${b:5}$ with contract \lstinline$|'b.@'|$, where \lstinline$b$ is read-only. This
object is assigned to \lstinline$x.a$ so that \lstinline$x.a$ and \lstinline$x.b$ are now
aliases. The strategy of JSConTest2 is to obey the contracts along all access paths.
Reading \lstinline$x.a$ again yields an object with contract
\lstinline$|'(e+b)&b.@'|$, where \lstinline$|e|$ stands for the empty
word and \lstinline$|&|$ is the conjunction operator. Thus, the
resulting contract \lstinline$|'(e+b)&b.@'|$ simplifies to
\lstinline$|'b.@'|$ such that writing to \lstinline$x.a.b$ causes a
violation.  

% Indeed, the contract of \lstinline{x} allows write access to \lstinline$x.a.b$, but the contract
% of the object behind property \lstinline$x.a$ doesn't. Both contracts must be observed. This
% approach is insensitive to aliasing and therefore suited for security settings.

%If an wrapped object is nested in another 
% wrapped object
% augmented with a contract is nested in a another augmented object, both contracts  have to
% receive attention. This approach is insensitive to aliasing and therefore suited for security settings.

% Beyond wrapping objects functions can be wrapped too. If a contract is attached to a function, than each invocation will install the contract to the arguments of the call. As a consequence, side-effects on the arguments get restricted by the given contract.

% \begin{lstlisting}
% var g = permitArgs(|'(arguments.0.a+arguments.1.@)'|,
%    function(x, y) {
%       s = x.a;
%       y.a = 3; // violation
%       y = 11;
%    });
% \end{lstlisting}

% The function object \lstinline$g$ is wrapped in the same way as other objects are. Calling \lstinline$g$ will wrap the arguments before they are forwarded to the function body (\lstinline$arguments.0$ address the first-, \lstinline$arguments.1$ the second argument). Reading \lstinline$x.a$ is permitted by the given contract, writing \lstinline$y.a$ not. Overwriting \lstinline$y$ causes no violation because only the assignment in \lstinline$y$ changes, not the object.

In addition to using full property names in contracts, the syntax admits regular expressions
for property names, too. For example, the contract \lstinline$|'(/^get.+/+next)*.length.@'|$
allows us to read the \lstinline{length} property after reading a chain of properties that either
start with \lstinline{get} or that are equal to \lstinline{next}.

%   _             _ _         _   _            ___                       _     
%  /_\  _ __ _ __| (_)__ __ _| |_(_)___ _ _   / __| __ ___ _ _  __ _ _ _(_)___ 
% / _ \| '_ \ '_ \ | / _/ _` |  _| / _ \ ' \  \__ \/ _/ -_) ' \/ _` | '_| / _ \
%/_/ \_\ .__/ .__/_|_\__\__,_|\__|_\___/_||_| |___/\__\___|_||_\__,_|_| |_\___/
%      |_|  |_|                                                                

\subsection{A Security Example}
\label{sec:security_example}

% A security scenario offers enormous possibilities to apply our analysis.

% a web page running JavaScript from different sources. Dynamically loaded code
% is neither easy to analyze nor is its influence completely known by the
% developer. All objects and functions run in a global scope without any
% namespace management and all objects inherit from a global object. This
% enables code fragments to mess other variables, to redefine functions, or to
% override native methods. All included scripts run with the same
% authority. Handled data can be transmitted to the net and cookie values can be
% stolen. The security concept of JavaScript is outdated. The implemented
% sandbox model restricts only access to data stored on a computer, but not to
% the elements inside. 

% JavaScript itself has no security awareness. Restricting access to objects is
% a powerful mechanism to guarantee safety properties like confidentiality or
% data integrity. The possibility to define read and write effects enables us to
% limit the influence on such objects. 

As an example from a security context, consider the following scenario, which
was used as an exploit to extract the contacts out of a GMail
account.\footnote{This exploit has been fixed in 2006.}

\begin{lstlisting}[language=html]
<script type="text/javascript"
  src="http://docs.google.com/data/contacts?out=js&
  show=ALL&psort=Affinity&callback=google&max=99999">
</script>
\end{lstlisting}

This script element is a JSONP request that loads the Google Mail contacts and sends it to the
\lstinline$google$ function, which is given as callback. The following listing
shows what the data given to the callback function could look like.

\begin{lstlisting}
var contacts = {
  Success: true,
  Errors: [],
  Body: {
    AuthToken: {
    Value: '********'
    },
    Contacts: [
    {
      Name: 'Jimmy Example',
      Email: 'email@example.org',
      Addresses: [],
      Phones: [],
      Ims: []
    },
    // More contacts
    ]
  }
};
\end{lstlisting}

To restrict access to the \lstinline$contacts$ object, the developer
could wrap it into a contract as follows.

\begin{lstlisting}
return __APC.permit(
  |'((Success.@+Errors.?*)+Body.Contacts.?.Name)'|,
  contacts);
\end{lstlisting}

This contract enables read access to \lstinline$Success$
(\lstinline$|'Success.@'|$), read/write access to everything below
\lstinline$Errors$ (\lstinline$|'Errors.?*'|$). Furthermore, only the \lstinline$Name$
propertiy can be accessed on each element of the contacts array
\lstinline$Body.Contacts$.
Access to the properties \lstinline$AuthToken$ as well as to the actual
contact data (e.g., \lstinline$Email$, \lstinline$Phones$,
\lstinline$Ims$) is not permitted, thus substantially diminishing the
value of an exploit. 

In this case, an access permission contract should restrict the
\texttt{google} function from using its 
argument arbitrarily. To be effective, a HTTP proxy would
have to insert the contract in the HTTP request resulting from the script tag.

% In a second view a developer wants to guarantee that a function, which
% critical data is send to, satisfies given access restrictions. 

\begin{lstlisting}
var google = __APC.permitArgs(|'arguments.0.
  ((Success.@+Errors.?*)+Body.Contacts.?.Name)'|,
  function(contacts) {
    // do something
});
\end{lstlisting}

The \lstinline$permitArgs$ method takes an access permission contract and a
function and returns a wrapped function, such that 
each call to the wrapped function enforces the contract. Arguments are addressed by position
so that \lstinline$|'arguments.0'|$ addresses the first argument. The remaining
contract specification is as before.

Because of the transparent implementation
of contract enforcement, the function that is wrapped is arbitrary: it may be
defined in the same source, it may be loaded dynamically, or it may be the
result of \texttt{eval}. Contract enforcement works in all cases.

% _____          _ _           _                  _           
%|  __ \        | (_)         (_)                (_)          
%| |__) | __ ___| |_ _ __ ___  _ _ __   __ _ _ __ _  ___  ___ 
%|  ___/ '__/ _ \ | | '_ ` _ \| | '_ \ / _` | '__| |/ _ \/ __|
%| |   | | |  __/ | | | | | | | | | | | (_| | |  | |  __/\__ \
%|_|   |_|  \___|_|_|_| |_| |_|_|_| |_|\__,_|_|  |_|\___||___/

\section{The JSConTest2 Approach}
\label{sec:preliminaries}

JSConTest2 implements a contracted target object by wrapping it in a proxy
object that intercepts all operations on the target and either
forwards them to the target or signals a contract violation. The monitoring
requires storing a set of access paths and the contract along with
the proxy. When reading a property of a contracted object that
contains another object, then the read operation must return a contracted object
that carries the remaining contract after the read operation (cf.\ the
examples in Section~\ref{sec:design_principles}). This ``contract
inheritance'' is an instance of the membrane pattern that is often
used in connection with proxies. Section~\ref{sec:proxies-membranes}
gives an introduction to proxies and membranes.

As contracts are closely related to regular expressions, the remaining
contract after a read operation can be nicely characterized using
Brzozowski-derivatives of regular
expressions. Section~\ref{sec:access-permission-contracts} formally
defines contracts and their semantics in terms of access paths, it
defines the derivative operation on contracts, establishes its basic
properties, and finishes by defining readable and writeable paths.
This section form the basis for 
Section~\ref{sec:formalization}, which formalizes the semantics of the
\lstinline$permit$ operations.

Section~\ref{sec:reduction} addresses some practical problems that
arise from the implementation. Under certain circumstances, the same
object may be subject to multiple contracts. A naive implementation
would create an inefficient chain of proxy objects, which can be
avoided by merging the path set and contract information. However,
these merge operations themselves lead to memory bloat, which can be
addressed by using suitable data structures and aggressive contract
simplification. 

\section{Proxies and Membranes}
\label{sec:proxies-membranes}

% ___             _        
%| _ \_ _ _____ _(_)___ ___
%|  _/ '_/ _ \ \ / / -_|_-<
%|_| |_| \___/_\_\_\___/__/

\subsection{Proxies}
\label{sec:pre-proxies}

\begin{figure}
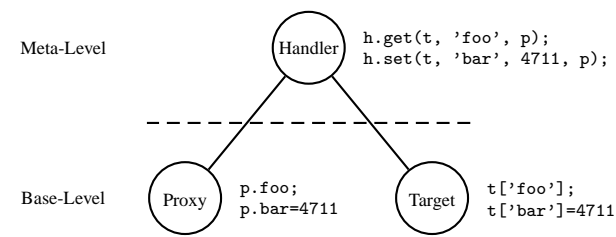

   \scriptsize
   \def\dedge{\ncline[linestyle=dashed]} \center
   \psset{radius=14pt, dotsize=2pt}
   \pstree[thislevelsep=2, edge=none, levelsep=0, treesep=0, xbbr=0.5]{\TR{Meta-Level}}{{\TR{Base-Level}}}
   \psset{edge=\ncline}
   \pstree[thislevelsep=2, levelsep=1, treesep=1]{
	  \TCircle[xbbr=0, xbbl=0, xbbd=0.5, name=handler]{Handler}~[tnpos=r]{
		 \parbox{3.4cm}{\texttt{h.get(t, 'foo', p);\\ h.set(t, 'bar', 4711, p);}}
	  }
   }{
	  \TCircle[xbbr=0, xbbl=0, xbbd=0, name=proxy]{Proxy} ~[tnpos=r]{
		 \parbox{1cm}{\texttt{p.foo;\\ p.bar=4711}}
	  }
	  \TCircle[xbbr=0, xbbl=0, xbbd=0,name=target]{Target} ~[tnpos=r]{
		 \parbox[c]{1.5cm}{\texttt{t['foo'];\\ t['bar']=4711}}
	  }
   }
   \ncline[linestyle=dashed,nodesep=-1,offset=-1]{target}{proxy}
   \caption{Example of proxy operation.}
   \label{fig:prxie_pattern}
\end{figure}

% The proxy pattern, its behavior, and the use of proxies in programming
% languages is well known from the literature
% \cite{VanCutsem:2010:PDP:1869631.1869638,van2012design,RobustComposition}. In
% our formalization we take use of a simplified proxy model.  

%The full implementation is shown in section \ref{sec:implementation}.

% A proxy is a wrapper for an object. Ideally, a proxy is not distinguishable from
% other objects. Calling an operation on a proxy should be no different from
% calling the operation on any other object.

A JavaScript proxy \cite{VanCutsem:2010:PDP:1869631.1869638} is an
object whose behavior is controlled by a handler object. A typical use
case is to have the handler mediate access to an arbitrary
target object, which may be a native or proxy object. The proxy is
then intended to be used in place of the target and is
not distinguishable from other objects. However, the proxy may modify the
original behavior of the target object in many respects. 

The handler object defines trap functions that implement the
operations on the proxy. Operations like property lookup or property update
are forwarded to the corresponding trap. The handler may implement the operation
arbitrarily; in the simplest case, it forwards the operation to the target
object. The handler may also be a proxy.

Figure~\ref{fig:prxie_pattern} contains a simple example, where the
handler \texttt{h} causes the proxy \texttt{p} to behave as a
wrapper for a target object \texttt{t}. Performing the property access
\texttt{p.foo} on the proxy object results in a meta-level call to the
corresponding trap on the handler object \texttt{h}. Here, the handler forwards
the property access to the target object. The property write is handled
similarly. 

% Following from the definition, the target and handler object can be a proxy
% object indeed. The usage of a proxy as handler object is described as
% meta-level funneling. In case that the target object is a proxy there will be
% a concatenation of handler calls. To avoid long chains of identical handler
% calls it is advisable to use maps to store handler and proxy references and to
% merge its attributes (e.g. the attached contracts) instead of re-wrapping
% proxies. 

% __  __           _                      
%|  \/  |___ _ __ | |__ _ _ __ _ _ _  ___ 
%| |\/| / -_) '  \| '_ \ '_/ _` | ' \/ -_)
%|_|  |_\___|_|_|_|_.__/_| \__,_|_||_\___|

\subsection{Membranes}
\label{sec:pre-membrane}

\begin{figure}
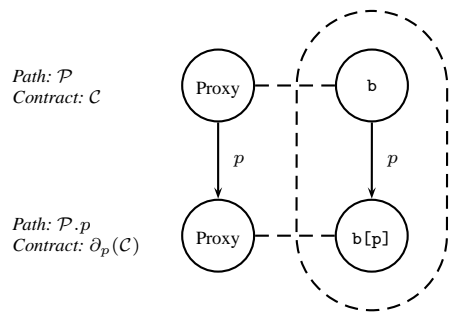

   \scriptsize
   \def\dedge{\ncline[linestyle=dashed]} \center
   \psset{radius=14pt, dotsize=2pt}
   \pstree[thislevelsep=2, levelsep=3, treesep=1, arrows=->]{
	  \TCircle[xbbr=0.5, xbbl=0.5, name=p1]{Proxy}~[tnpos=l]{
		 \parbox{2cm}{
			\Label{Path:} $\tpPath$\\
			\Label{Contract:} $\apcContract$
		 }
	  }
   }{
	  \TCircle[xbbr=0.5, xbbl=0, xbbd=0.5, name=p2]{Proxy} ~[tnpos=l]{
		 \parbox{2cm}{
			\Label{Path:} $\tpPath\tpDot\tpProperty$\\
			\Label{Contract:} $\deriv{\tpProperty}\apcContract$
		 }
	  }\naput{$\tpProperty$}
   }
   \psset{edge=\ncline}
   \pstree[thislevelsep=2, levelsep=3, treesep=1, arrows=->]{
	  \TCircle[xbbr=0.5, xbbl=0.5, name=b]{\texttt{b}}
   }{
	  \TCircle[xbbr=0.5, xbbl=0.5, xbbd=0, name=t]{\texttt{b[p]}}
	  \naput{$\tpProperty$}
   }\ncbox[linearc=1,boxsize=1,linestyle=dashed,nodesep=0.5]{b}{t}
   \ncline[linestyle=dashed]{p1}{b}
   \ncline[linestyle=dashed]{p2}{t}
   \caption{Example of property access through membrane.}
   \label{fig:membrane_pattern}
\end{figure}

Our technique to implement objects under a contract is inspired by
\emph{Revocable Membranes}
\cite{VanCutsem:2010:PDP:1869631.1869638,van2012design,RobustComposition}. A
membrane serves as a regulated communication channel between an object and the
rest of the program. It ensures that all parts of the objects behind a membrane
also remain behind. For example, each property access on a wrapped object
(e.g. $\texttt{obj.p}$) returns another wrapped object.  Therefore, after
wrapping, no new direct references to target objects behind the membrane become
available.
One use of this mechanism is to revoke all references to an object network or to
enforce write protection 
\cite{VanCutsem:2010:PDP:1869631.1869638,van2012design,RobustComposition}.

In our use of membranes (cf.\ Figure~\ref{fig:membrane_pattern}), each handler contains a path $\tpPath$, and a contract
$\apcContract$ describing the allowed field accesses. Each property access
\texttt{obj.$\tpProperty$} on a wrapped object returns a wrapped object whose
path is $\tpPath\tpDot\tpProperty$.  In addition, the handler traps enforce the
contract $\apcContract$. If the access on property $\tpProperty$ is allowed by
contract $\apcContract$ the handler forwards the request to the target object
and wraps the returned object with the new contract
$\deriv{\tpProperty}\apcContract$, which is the derivative of $\apcContract$ with respect to
$\tpProperty$ (explained in Section~\ref{sec:derivation}). If this access is not allowed, then the handler prevents it in
a configurable way.

The Figure~\ref{fig:membrane_pattern} shows a membrane arising
from an allowed property access. The information on the left is contained in the
handler objects and the objects inside the membrane on the right are the target
objects of the proxies. Thus, our implementation logs all
access paths to wrapped objects in their handlers.

%                                 _____                    _         _             
%    /\                          |  __ \                  (_)       (_)            
%   /  \   ___ ___ ___  ___ ___  | |__) |__ _ __ _ __ ___  _ ___ ___ _  ___  _ __  
%  / /\ \ / __/ __/ _ \/ __/ __| |  ___/ _ \ '__| '_ ` _ \| / __/ __| |/ _ \| '_ \ 
% / ____ \ (_| (_|  __/\__ \__ \ | |  |  __/ |  | | | | | | \__ \__ \ | (_) | | | |
%/_/    \_\___\___\___||___/___/ |_|   \___|_|  |_| |_| |_|_|___/___/_|\___/|_| |_|                                                                                 
%  _____            _                  _       
% / ____|          | |                | |      
%| |     ___  _ __ | |_ _ __ __ _  ___| |_ ___ 
%| |    / _ \| '_ \| __| '__/ _` |/ __| __/ __|
%| |___| (_) | | | | |_| | | (_| | (__| |_\__ \
% \_____\___/|_| |_|\__|_|  \__,_|\___|\__|___/

\section{Access Permission Contracts}
\label{sec:access-permission-contracts}

This section defines the syntax and semantics of access permission contracts and access paths.

% ___      _   _    
%| _ \__ _| |_| |_  
%|  _/ _` |  _| ' \ 
%|_| \__,_|\__|_||_|

\subsection{Access Paths}
\label{sec:paths}

Let $\apcAlphabet$ be a set of property names and $\tpBlank\notin
\apcAlphabet$ be a special blank property that does not occur in any
JavaScript object. Its sole purpose is to indicate read-only accesses.
Let $\tpProperty\in
\apcAlphabet \cup \{\tpBlank\}$ range over all properties.
An access path $\tpPath \in (\apcAlphabet \cup \{\tpBlank\})^*$ is a
sequence of properties.
% Every property access extends path $\tpPath$ with the accessed property name
% $\tpProperty$. The extended path $\tpPath\tpDot\tpProperty$ is given to the subsequent
% proxy or used in a violation message. 
%
% \begin{figure}
%    \centering
%    \begin{displaymath}
%       \begin{array}{llrlll}
%          \Label{\TpProperty} &\ni~ \tpProperty &::=& \ljVar &\Label{(name)}\\
%          &&|& \tpBlank &\Label{(empty property)}\\
%          \\
%          \Label{\TpPath} &\ni~ \tpPath &::=& \tpEmpty &\Label{(empty path)}\\   
%          &&|& \tpProperty &\Label{(property)}\\
%          &&|& \tpPath\tpDot\tpPath &\Label{(sequence)}\\
%       \end{array}
%    \end{displaymath}
%    \caption{Syntax of access paths.}
%    \label{fig:syntax_paths}
% \end{figure}
%
% The syntax is shown in figure \ref{fig:syntax_paths}.
% A property $\tpProperty$ is either a property name or the empty
% property $\tpBlank$. 
We write $\tpEmpty$ for the empty path and $\tpPath\tpDot\tpPath$ for
the concatenation of two paths (considered as sequences).

\subsection{Contracts}
\label{sec:contracts}

% Access permission contracts grant the permissions  \emph{readable} or \emph{writeable} to
% properties. Each contract defines a set of permitted access paths of the form
% $\tpPath$ and assigns \emph{readable} for all prefixes of $\tpPath$ and \emph{writeable}
% only to $\tpPath$.

\begin{figure}
   \centering
   \begin{displaymath}
	  \begin{array}{llrlll}

		 \Label{\ApcLiteral} &\ni~ \apcLiteral &::=& \apcAT &\Label{(empty literal)}\\ 
		 &&|& \apcQMark &\Label{(universe)}\\
		 &&|& \apcRegEx &\Label{(regular expression)}\\
		 &&|& \apcNeg\apcRegEx &\Label{(negation)}\\

		 \\
		 \Label{\ApcContract} &\ni~ \apcContract &::=& \apcEmptySet &\Label{(empty set)}\\
		 &&|& \apcEmpty &\Label{(empty contract)}\\
		 &&|& \apcLiteral &\Label{(literal)}\\
%space                &&|& \apcContract\apcQMark &\Label{(option)}\\
		 &&|& \apcContract\apcStar &\Label{(Kleene star)}\\
		 &&|& \apcContract\apcOr\apcContract &\Label{(logical or)}\\
		 &&|& \apcContract\apcAnd\apcContract &\Label{(logical and)}\\
		 &&|& \apcContract\apcDot\apcContract &\Label{(concatenation)}\\

	  \end{array}
   \end{displaymath}
   \caption{Syntax of access permission contracts.}
   \label{fig:syntax_contracts}
\end{figure}

Figure~\ref{fig:syntax_contracts} defines the syntax of contracts.
Contract literals $\apcLiteral$ are the primitive building blocks of
contracts. Each literal defines a property access. A literal
$\apcLiteral$ is either the empty
literal $\apcAT$, the universe literal $\apcQMark$, a regular
expression $\apcRegEx$, or a  negated regular expression
$\apcNeg\apcRegEx$. The empty literal $\apcAT$ stands for the blank 
property $\tpBlank$. It should not be confused with the empty set
contract $\apcEmptySet$. The universe literal $\apcQMark$ represents
the set of all JavaScript property names. A regular expression
$\apcRegEx$ describes a set of matching property names. We assume that
these expressions are JavaScript regular expressions, which we treat
as abstract in this work.

Contracts are regular expressions extended with intersection.
A contract $\apcContract$ is either an empty set $\apcEmptySet$, an empty contract
$\apcEmpty$, a single literal $\apcLiteral$, a Kleene star $\apcContract\apcStar$, a
disjunction $\apcContract\apcOr\apcContract$, a conjunction
$\apcContract\apcAnd\apcContract$, or a concatenation
$\apcContract\apcDot\apcContract$.
% We consider the option contract $\apcContract\apcQMark$ as an
% abbreviation for $\apcEmpty\apcOr\apcContract$. 
Beware that a literal may contain a
regular expression at the character level.

Each contract defines a set of access paths as defined in
Figure~\ref{fig:language_literals}. This definition follows the usual semantics of regular
expressions with a few specialities. The empty literal yields the empty
property. $\apcAlphabet$ is the set of all property
names. $\regmatch{\apcRegEx}{\tpProperty}$ is a predicate that indicates whether property
$\tpProperty$ matches regular expression $\apcRegEx$ (as a standard regular expression on characters). 

% The optional contract $\apcContract\apcQMark$ is equivalent to $(\apcContract\apcOr\apcEmpty)$. In a similar way, the kleene star contract $\apcContract\apcStar$ applies $\apcContract$ infinity often, but can also be reduced to the empty contract. The disjunctions of contracts $\apcContract\apcOr\apcContract$ applies one of them. A conjunctions of contracts $\apcContract\apcAnd\apcContract$ have to apply both. The negation of contracts $\apcNeg\apcContract$ reduces the set of possible access paths by the set defined in $\apcContract$. The contract concatenations $\apcContract\apcDot\apcContract$ combines contract elements in a consecutive order to define access paths.

% Let the alphabet $\regexAlphabet$ be the finite set of symbols acceptable in property names. $\regexWords$ denotes all finite strings over $\regexAlphabet$. $\apcAlphabet=\regexWords$ is the set of all possible property names. $\apcWords$ denotes all finite chains of property names over $\apcAlphabet$. $\apcProperties=\apcWords\backslash\apcEmpty$ is the set of all chains acceptable as path. The language $\lang{\apcRegEx}$ described by the regular expression $\apcRegEx$ is subset of $\apcAlphabet$. 

We say that the contract literal $\apcLiteral$ matches property
$\tpProperty$,  written as $\match{\apcLiteral}{\tpProperty}$,  iff $\tpProperty\in\lang{\apcLiteral}$.
% \begin{definition}[Property Language]\label{def:language_property}
%    The property language $\lang{\apcLiteral}$ described by a literal $\apcLiteral$ builds the set of matching property names $\tpProperty$.
%    Literal $\apcLiteral$ matches property $\tpProperty$ iff $\tpProperty$ is subset of the language defined by $\lang{\apcLiteral}$.
%    \begin{gather}
%       \match{\apcLiteral}{\tpProperty} ~\Leftrightarrow~ \tpProperty\in\lang{\apcLiteral}
%    \end{gather}
%    We write $\regmatch{\apcRegEx}{\tpProperty}$ iff $\apcRegEx$ matches property $\tpProperty$.
% \end{definition}
We further say that a contract $\apcContract$ matches path $\tpPath$, written
$\match{\apcContract}{\tpPath}$, iff $\tpPath\in\lang{\apcContract}$.
% \begin{definition}[Path Language]\label{def:language_path}
%    The path language $\lang{\apcContract}$ of a contract $\apcContract$ is the set of all deducible paths $\tpPath$.
%    Contract $\apcContract$ matches path $\tpPath$ iff $\tpPath$ is subset of the language defined by $\lang{\apcContract}$.
%    \begin{gather}
%       \match{\apcContract}{\tpPath} ~\Leftrightarrow~ \tpPath\in\lang{\apcContract}
%    \end{gather}
% \end{definition}

\begin{figure}
   \centering
   \begin{displaymath}
	  \begin{array}{lll}

		 \lang{\apcAT} &=& \{\tpBlank\}\\
		 \lang{\apcQMark} &=& \apcAlphabet\\
		 \lang{\apcRegEx} &=& \{\tpProperty ~|~ \regmatch{\apcRegEx}{\tpProperty}\}\\
		 \lang{\apcNeg\apcRegEx} &=& \apcAlphabet \backslash \lang{\apcRegEx}\\
		 \lang{\apcEmptySet} &=& \{\}\\
		 \lang{\apcEmpty} &=& \{\tpEmpty\}\\
%space                \lang{\apcContract\apcQMark} &=& \{\apcEmpty\} \cup \lang{\apcContract} \\
		 \lang{\apcContract\apcStar} &=& \{\tpEmpty\} \cup \lang{\apcContract\apcDot\apcContract\apcStar}\\
		 \lang{\apcContract\apcOr\apcContract'} &=& \lang{\apcContract} \cup \lang{\apcContract'}\\
		 \lang{\apcContract\apcAnd\apcContract'} &=& \lang{\apcContract} \cap \lang{\apcContract'}\\
		 \lang{\apcContract\apcDot\apcContract'} &=& \{\tpPath.\tpPath' ~|~ \tpPath\in\lang{\apcContract},\tpPath'\in\lang{\apcContract'}\}\\

	  \end{array}
   \end{displaymath}
   \caption{Language of contracts.}
   \label{fig:language_literals}
\end{figure}

% The defined language is according to the definition in figure \ref{fig:language_literals}.

The last property $\tpProperty$ of an access path $\tpPath\tpDot\tpProperty$ is readable and
writeable. All properties along the prefix $\tpPath$ are readable. A contract ending with
the empty literal $\apcContract\apcDot\apcAT$ is a read-only contract. It matches access
paths of the form $\tpPath\tpDot\tpBlank$ that end with the blank property $\tpBlank$, which
never occurs in a program.

% \begin{definition}[Accessibility of Paths]\label{def:accessibility_path}
%    Each deducible path $\tpPath\tpDot\tpProperty$ of $\lang{\apcContract}$ is readable on his prefix $\tpPath$ and readable or writeable on his suffix $\tpProperty$.
% \end{definition}

% Further we denote the literal width $\width{\apcContract}$ of contract $\apcContract$ as the amount of literals $\apcLiteral$ in $\apcContract$.

% Two contracts $\apcContract$ and $\apcContract'$ are similar $\isSimilar{\apcContract}{\apcContract'}$ iff the described language of both contracts is equivalent $\lang{\apcContract}=\lang{\apcContract'}$. The normalization of contracts $\norm{\apcContract}{\apcContract'}$ is defined as reduction to a minimal contract that grantees the similarity.

%\begin{definition}[Similar]\label{def:similar}
%   Two contracts $\apcContract$ and $\apcContract'$ are similar $\isSimilar{\apcContract}{\apcContract'}$ iff the described language of both contracts is equivalent $\lang{\apcContract}=\lang{\apcContract'}$.
%   \begin{gather}
%      \isSimilar{\apcContract}{\apcContract'} ~\Leftrightarrow~ \lang{\apcContract}=\lang{\apcContract'}
%   \end{gather}
%\end{definition}

%\begin{definition}[Normalization]\label{def:normalization_contracts}
%   The normalization of contracts $\norm{\apcContract}{\apcContract'}$ is defined as reduction to a minimal contract that grantees the similarity.
%   \begin{gather}
%      \norm{\apcContract}{\apcContract'} ~\Rightarrow~ \isSimilar{\apcContract}{\apcContract'}
%      ~\wedge~ \width{\apcContract}\leq\width{\apcContract'}
%   \end{gather}
%\end{definition}

\begin{figure}
   \centering
   \begin{minipage}{0.4\linewidth}
	  \begin{displaymath}
		 \begin{array}{rll}
								%\norm{\apcEmpty\apcDot\apcContract&}{&\apcContract}\\
%space                                \norm{\apcEmpty\apcQMark&}{&\apcEmpty}\\
			\norm{\apcEmpty\apcStar&}{&\apcEmpty}\\

			\\
			\norm{\apcEmpty\apcDot\apcContract&}{&\apcContract}\\
			\norm{\apcAT\apcDot\apcContract&}{&\apcAT}\\
			\norm{\apcEmptySet\apcDot\apcContract&}{&\apcEmptySet}\\
		 \end{array}
	  \end{displaymath}
   \end{minipage}
   \begin{minipage}{0.4\linewidth}
	  \begin{displaymath}
		 \begin{array}{rll}
			\norm{\apcEmptySet\apcOr\apcContract&}{&\apcContract}\\
			\norm{\apcAT\apcOr\apcContract&}{&\apcContract}\\
			\norm{\apcContract\apcOr\apcContract&}{&\apcContract}\\

			\\
			\norm{\apcEmptySet\apcAnd\apcContract&}{&\apcEmptySet}\\
			\norm{\apcAT\apcAnd\apcContract&}{&\apcAT}\\
			\norm{\apcContract\apcAnd\apcContract&}{&\apcContract}\\
		 \end{array}
	  \end{displaymath}
   \end{minipage}
   \caption{Normalization rules for contracts.}
   \label{fig:normalization-rules_contracts}
\end{figure}

Figure~\ref{fig:normalization-rules_contracts} contains normalization rules for
contracts. 
We say that a contract $\apcContract$ is \emph{normalized} iff it cannot be further reduced
by these rules.
From now on, we regards all contracts as normalized.

% ___          _          _   _                 __    ___         _               _      
%|   \ ___ _ _(_)_ ____ _| |_(_)___ _ _    ___ / _|  / __|___ _ _| |_ _ _ __ _ __| |_ ___
%| |) / -_) '_| \ V / _` |  _| / _ \ ' \  / _ \  _| | (__/ _ \ ' \  _| '_/ _` / _|  _(_-<
%|___/\___|_| |_|\_/\__,_|\__|_\___/_||_| \___/_|    \___\___/_||_\__|_| \__,_\__|\__/__/

\subsection{Derivatives of Contracts}
\label{sec:derivation}

In this section we introduce the notion of a derivative for a contract, which is defined analogously
to the derivative of a regular expression \cite{Brzozowski1964,Owens:2009:RDR:1520284.1520288}.
Derivatives are best explained in terms of a language quotient, which
is the set of suffixes of words in the language after taking away a
prescribed prefix.

\begin{definition}[Left quotient]\label{def:left_quotient}
   Let $L \subseteq \apcAlphabet^*$ be a language.
   The \emph{left quotient} $\leftquotient{\tpPath}{L}$ of the language $L$ with respect to an access path
   $\tpPath$ is defined as:
   \begin{gather}
	  \leftquotient{\tpPath}{L} ~=~ \{\tpPath' ~|~ \tpPath\tpDot\tpPath'\in L\}
   \end{gather}
\end{definition}

Clearly, it holds that $\{\tpPath\tpDot\tpPath' ~|~ \tpPath'\in\leftquotient{\tpPath}{L}\} \subseteq L$.
It is also immediate from the definition that  $\leftquotient{(\tpProperty\tpDot\tpPath)}{L} ~=~ \leftquotient{\tpPath}{(\leftquotient{\tpProperty}{L})}$.

To compute the derivative of a contract $\apcContract$ w.r.t. an access path $\tpPath$ we have to introduce an auxiliary function $\nullable$ to determine if a contract $\apcContract$ matches the empty path $\tpEmpty$.
Figure~\ref{fig:is-nullable} contains its definition.

\begin{figure}
   \begin{displaymath}
	  \begin{array}[t]{lll}
		 \isNullable{\apcAT} &=& \perp\\
		 \isNullable{\apcQMark} &=& \perp\\
		 \isNullable{\apcRegEx} &=& \perp\\
		 \isNullable{\apcNeg\apcRegEx} &=& \perp\\
		 \isNullable{\apcEmptySet} &=& \perp\\
	  \end{array}
	  \qquad
	  \begin{array}[t]{lll}

		 \isNullable{\apcEmpty} &=& \top\\
%space                \isNullable{\apcContract\apcQMark} &=& \top\\
		 \isNullable{\apcContract\apcStar} &=& \top\\
		 \isNullable{\apcContract\apcOr\apcContract'} &=& \isNullable{\apcContract} \vee \isNullable{\apcContract'}\\
		 \isNullable{\apcContract\apcAnd\apcContract'} &=& \isNullable{\apcContract} \wedge \isNullable{\apcContract'}\\
		 \isNullable{\apcContract\apcDot\apcContract'} &=& \isNullable{\apcContract} \wedge \isNullable{\apcContract'}\\

	  \end{array}
   \end{displaymath}
   \caption{The predicate ``is nullable''.}
   \label{fig:is-nullable}
\end{figure}

\begin{definition}[Nullable]\label{def:nullable}
   A contract $\apcContract$ is nullable iff its language $\lang{\apcContract}$ contains the
   empty access path $\apcEmpty$.
\end{definition}

\begin{lemma}[Nullable]\label{thm:nullable}
   \begin{gather}
	  \apcEmpty\in\lang{\apcContract} ~\Leftrightarrow~ \isNullable{\apcContract}=\top
   \end{gather}
\end{lemma}

If we access a target object by reading property $\tpProperty$ on an object with contract
$\apcContract$, then the access language for the target object is
$\leftquotient{\tpProperty}{\lang{\apcContract}}$. As for regular
expressions, we can compute a derivative contract $\deriv{\tpProperty}{\apcContract}$ of
$\apcContract$ with respect to $\tpProperty$ symbolically, such that
$\leftquotient{\tpProperty}{\lang{\apcContract}} = \lang{\deriv{\tpProperty}{\apcContract}}$. 
Figure~\ref{fig:derivation} contains the definition of the derivative
for a single property. We extend this definition to access paths by 
\begin{displaymath}
   \begin{array}{lcl}
	  \deriv{\apcEmpty}{\apcContract}& =& \apcContract \\
	  \deriv{\tpProperty\tpDot\tpPath}{\apcContract}& =& \deriv{\tpPath}{\deriv{\tpProperty}{\apcContract}}
   \end{array}
\end{displaymath}

\begin{lemma}[Derivatives of Contracts]\label{thm:derivation}
   For all paths $\tpPath$ it holds that:
   \begin{enumerate}
	  \item  $\lang{\deriv{\tpPath}{\apcContract}} = \leftquotient{\tpPath}{\lang{\apcContract}}$
	  \item $\lang{\tpPath\tpDot \deriv{\tpPath}{\apcContract}} ~\subseteq~ \lang{\apcContract}$
	  \item $\tpPath\in\lang{\apcContract} ~\Leftrightarrow~ \isNullable{\deriv{\tpPath}{\apcContract}}$
   \end{enumerate}
\end{lemma}

\begin{figure}
   \centering
   \begin{displaymath}
	  \begin{array}{lll}

		 \deriv{\tpProperty}{\apcAT} &=& \apcEmptySet\\
		 \deriv{\tpProperty}{\apcQMark} &=& \apcEmpty\\
		 \deriv{\tpProperty}{\apcRegEx} &=& \begin{cases}
			\apcEmpty, &\regmatch{\apcRegEx}{\tpProperty}\\
			\apcEmptySet, &\text{otherwise}
		 \end{cases}\\  
		 \deriv{\tpProperty}{\apcNeg\apcRegEx} &=& \begin{cases}
			\apcEmptySet, &\regmatch{\apcRegEx}{\tpProperty}\\
			\apcEmpty, & \text{otherwise}
		 \end{cases}\\
		 \deriv{\tpProperty}{\apcEmptySet} &=& \apcEmptySet\\
		 \deriv{\tpProperty}{\apcEmpty} &=& \apcEmptySet\\
%space                 \deriv{\tpProperty}{\apcContract\apcQMark} &=& \deriv{\tpProperty}{\apcContract}\\
		 \deriv{\tpProperty}{\apcContract\apcStar} &=& \deriv{\tpProperty}{\apcContract}\apcDot\apcContract\apcStar\\
		 \deriv{\tpProperty}{\apcContract\apcOr\apcContract'} &=& \deriv{\tpProperty}{\apcContract} \apcOr \deriv{\tpProperty}{\apcContract'}\\
		 \deriv{\tpProperty}{\apcContract\apcAnd\apcContract'} &=& \deriv{\tpProperty}{\apcContract}\apcAnd \deriv{\tpProperty}{\apcContract'}\\
		 \deriv{\tpProperty}{\apcContract\apcDot\apcContract'} &=& \begin{cases}
			\deriv{\tpProperty}{\apcContract}\apcDot\apcContract'\apcOr\deriv{\tpProperty}{\apcContract'}, &\isNullable{\apcContract}\\
			\deriv{\tpProperty}{\apcContract}\apcDot\apcContract', &\text{otherwise}
		 \end{cases}\\

	  \end{array}
   \end{displaymath}
   \caption{Derivative of a contract by a property.}
   \label{fig:derivation}
\end{figure}

% __  __      _      _    _           
%|  \/  |__ _| |_ __| |_ (_)_ _  __ _ 
%| |\/| / _` |  _/ _| ' \| | ' \/ _` |
%|_|  |_\__,_|\__\__|_||_|_|_||_\__, |
%                               |___/ 

\subsection{Matching}
\label{sec:matching}

% Now, we define the accessibility of $\tpProperty$ in $\apcContract$ as derivation of $\apcContract$ with respect to $\tpProperty$.

By Lemma~\ref{thm:derivation}, $\apcContract'=\deriv{\tpPath}{\apcContract}$ defines the
language containing the remaining paths after reading $\tpPath$. If path $\tpPath$ is not a
prefix of a path in $\lang{\apcContract}$, then $\apcContract'$ must be the empty set
$\apcEmptySet$. If $\tpPath$ is an element of $\lang{\apcContract}$, then the language
of $\apcContract'$ contains the empty path $\apcEmpty$. By definition, each path and each prefix of a
path is readable. Thus, readability and writeability can be determined by checking whether the
remaining language is the empty set.

\begin{definition}[Readable] \label{def:readable}
   An access path $\tpPath$ is readable with respect to contract $\apcContract$
   iff the derivative of contract $\apcContract$ with respect to path $\tpPath$ results in contract $\apcContract'$
   with $\lang{\apcContract'}\neq\emptyset$. That is:
   \begin{gather}
	  \isReadable{\apcContract}{\tpPath} ~\Leftrightarrow~ \lang{\deriv{\tpPath}{\apcContract}}\neq\emptyset
   \end{gather}
\end{definition}

Every path $\tpPath$ in $\lang{\apcContract}$ is writeable. By Lemma~\ref{thm:derivation},
we know that a path $\tpPath$ is an element of the language defined by $\apcContract$ iff
$\tpEmpty\in\lang{\deriv{\tpPath}{\apcContract}}$. 

\begin{definition}[Writeable] \label{def:writeable}
   An access path $\tpPath$ is writeable with respect to contract $\apcContract$   iff the
   derivative of $\apcContract$ with respect to path $\tpPath$ is nullable. 
   That is:
   \begin{gather}
	  \isWrireable{\apcContract}{\tpPath} ~\Leftrightarrow~ \isNullable{\deriv{\tpPath}{\apcContract}}
   \end{gather}
\end{definition}

% Now, we apply our contract system to the evaluation of a JavaScript core calculus.

% ______                         _ _          _   _             
%|  ____|                       | (_)        | | (_)            
%| |__ ___  _ __ _ __ ___   __ _| |_ ______ _| |_ _  ___  _ __  
%|  __/ _ \| '__| '_ ` _ \ / _` | | |_  / _` | __| |/ _ \| '_ \ 
%| | | (_) | |  | | | | | | (_| | | |/ / (_| | |_| | (_) | | | |
%|_|  \___/|_|  |_| |_| |_|\__,_|_|_/___\__,_|\__|_|\___/|_| |_|

\section{Formalization}
\label{sec:formalization}

This section presents the formal semantics of path monitoring and contract enforcement in terms of a JavaScript core calculus $\lcpj$ extended with access permission contracts and access paths.

% ___          _            
%/ __|_  _ _ _| |_ __ ___ __
%\__ \ || | ' \  _/ _` \ \ /
%|___/\_, |_||_\__\__,_/_\_\
%     |__/                  

\subsection{Syntax}
\label{sec:syntax}

$\lcpj$ (Figure \ref{fig:semantics_lcpj}) is a call-by-value lambda calculus extended with objects and object-proxies. The syntax is close to JavaScript core calculi from the literature \cite{Guha:2010:EJ:1883978.1883988,ecma:262}.

A $\lcpj$ expression is either a constant $\ljConst$ including $\ljUndefined$ and $\ljNull$, a variable $\ljVar$, a lambda expression, an application, an object creation, a property reference, a property assignment, or a permit expression. The novel $\ljPermit$ expression applies the given contract $\apcContract$ to the object arising from expression $\ljExp$. 

% ___                     _   _      ___                 _         
%/ __| ___ _ __  __ _ _ _| |_(_)__  |   \ ___ _ __  __ _(_)_ _  ___
%\__ \/ -_) '  \/ _` | ' \  _| / _| | |) / _ \ '  \/ _` | | ' \(_-<
%|___/\___|_|_|_\__,_|_||_\__|_\__| |___/\___/_|_|_\__,_|_|_||_/__/

\subsection{Semantic Domains}
\label{sec:semantic_domains}

\begin{figure}
   \begin{displaymath}
	  \begin{array}{llrl}

		 \Label{\LjExpression} &\ni~ \ljExp &::=& \ljConst \mid \ljVar \mid \ljFunction \mid \ljExp(\ljExp) \mid \ljNew~\ljExp\\
		 &&\mid& \ljExp[\ljExp] \mid \ljExp[\ljExp]=\ljExp \mid \ljPermit~ \apcContract ~\ljIn~ \ljExp\\

		 \\
		 \Label{\LjLocation} &\ni~\ljLocation &&\\       
		 \Label{\LjValue} &\ni~ \ljVal &::=& \ljConst ~|~ \ljLocation\\

		 \\
		 \Label{\LjMonitor} &\ni~ \monitor &::=& \emptyset ~|~ \addReadPath{\tpPath} ~|~ \addWritePath{\tpPath}\\
		 &&\mid& \monitor;\monitor'\\

		 \\
		 \Label{\LjHandler}  &\ni~ \ljHandler &::=& \langle \tpPath, \apcContract \rangle\\
		 \Label{\LjProxy}  &\ni~ \ljProxy &::=& \langle \ljLocation, \ljHandler \rangle\\

		 \\
		 \Label{\LjPrototype} &\ni~\ljPrototype &::=& \ljVal\\
		 \Label{\LjClosure} &\ni~\ljClosure &::=& \emptyset ~|~ \langle \ljEnv, \ljFunction \rangle\\
		 \Label{\LjObject} &\ni~ \ljObj &::=& \emptyset ~|~ \ljObj[\ljStr\mapsto\ljVal]\\
		 \Label{\LjStorable} &\ni~ \ljStorable &::=& \langle \ljObj, \ljClosure, \ljPrototype \rangle ~|~ \ljProxy\\
		 \Label{\LjEnvironment} &\ni~ \ljEnv &::=& \emptyset ~|~ \ljEnv[\ljVar\mapsto\ljVal]\\
		 \Label{\LjHeap} &\ni~ \ljHeap &::=& \emptyset ~|~ \ljHeap[\ljLocation\mapsto\ljStorable]
	  \end{array}
   \end{displaymath}
   \caption{Syntax and semantic domains of $\lcpj$.}
   \label{fig:semantics_lcpj}
\end{figure}

Figure \ref{fig:semantics_lcpj} defines the semantic domains of $\lcpj$.

The heap maps a location $\ljLocation$ to a storable $\ljStorable$, which is either a proxy object $\ljProxy$ or a triple consisting of an object $\ljObj$, a function closure $\ljClosure$, and a value $\ljPrototype$ as prototype. A Proxy is a wrapper for a location $\ljLocation$ augmented with an access handler $\ljHandler$, which is a tuple consisting of a path $\tpPath$ and a contract $\apcContract$. An object $\ljObj$ maps a string to a value. A function closure consists of an expression $\ljExp$ and an environment $\ljEnv$, which maps a variable to a value $\ljVal$.

Further, a monitor $\monitor$ is a collection used for effect monitoring. It records all paths that have been accessed during the evaluation. The notation $\addReadPath{\tpPath}$ adds path $\tpPath$ as read effect to the monitor. Synonymously, $\addWritePath{\tpPath}$ adds a write effect. $\monitor;\monitor'$ denotes the union of two collections.

%, and $\addReadViolation{\tpPath}{\apcContract}$ and $\addWriteViolation{\tpPath}{\apcContract}$ adds a violation, consisting of path $\tpPath$ and contract $\apcContract$, to the effect monitor.

%A list $\lst^{\tpPath,\apcContract}$ of pairs consisting of a path and a contract stores a accessed path violating $\apcContract$

\begin{figure}
   \centering
   \begin{displaymath}
	  \begin{array}{l@{~}l@{~}l}

		 \langle \ljObj, \ljClosure, \ljPrototype \rangle(\ljStr) ~&=&~ \begin{cases}
			\ljVal, & \ljObj=\ljObj'[\ljStr\rightarrow\ljVal]\\
			\ljObj'(\ljStr), & \ljObj=\ljObj'[\ljStr'\rightarrow\ljVal] \\ %~|~ \ljStr\neq\ljStr'\\
			\ljHeap(\ljLocation)(\ljStr), & \ljObj=\emptyset ~\wedge~ \ljPrototype=\ljLocation\\
			\ljUndefined, & \ljObj=\emptyset ~\wedge~ \ljPrototype=\ljConst
		 \end{cases}\\

		 \langle \ljObj, \ljClosure, \ljPrototype \rangle[\ljStr\mapsto\ljVal] ~&=&~ \langle \ljObj[\ljStr\mapsto\ljVal], \ljClosure, \ljPrototype \rangle\\

		 \ljHeap[\ljLocation,\ljStr \mapsto \ljVal] ~&=&~
		 \ljHeap[\ljLocation\mapsto\ljHeap(\ljLocation)[\ljStr \mapsto \ljVal]]\\

		 \ljHeap[\ljLocation \mapsto \ljPrototype] ~&=&~ \ljHeap[\ljLocation \mapsto \langle \emptyset,\emptyset,\ljPrototype \rangle]\\

		 \ljHeap[\ljLocation \mapsto \ljClosure] ~&=&~ \ljHeap[\ljLocation \mapsto \langle \emptyset,\ljClosure,\ljNull \rangle]\\

		 \langle \ljLocation, \langle \tpPath, \apcContract \rangle \rangle ~&=&~  \langle \ljLocation, \tpPath, \apcContract \rangle

	  \end{array}
   \end{displaymath}
   \caption{Abbreviations.}
   \label{fig:syntactic_sugar}
\end{figure}

Figure \ref{fig:syntactic_sugar} introduces some abbreviated notations. A property lookup or a property update on a storable $\ljStorable=\langle \ljObj, \ljClosure, \ljPrototype \rangle$ is relayed to the underlying object. The property access $\ljStorable(\ljStr)$ returns $\ljUndefined$ by default if the accessed string is not defined in $\ljObj$ and the prototype of $\ljStorable$ is not a location $\ljLocation$. The notation $\ljHeap[\ljLocation,\ljStr \mapsto \ljVal]$ updates a property of storable $\ljHeap(\ljLocation)$, $\ljHeap[\ljLocation \mapsto \ljPrototype]$ initializes an object, and $\ljHeap[\ljLocation \mapsto \ljClosure]$ defines a function. Further, we write $\langle \ljLocation, \tpPath, \apcContract \rangle$ for a protected location.

% ___          _           _   _          
%| __|_ ____ _| |_  _ __ _| |_(_)___ _ _  
%| _|\ V / _` | | || / _` |  _| / _ \ ' \ 
%|___|\_/\__,_|_|\_,_\__,_|\__|_\___/_||_|

\subsection{Evaluation of $\lcpj$}
\label{sec:formalization_evaluation_lj}

Program execution is modeled by a big-step evaluation judgment of the form $ \ljHeap,\ljEnv
~\entails~\ljExp ~\eval~ \ljHeap' ~|~ \ljVal ~|~ \monitor$. The evaluation of expression $\ljExp$ with
initial heap $\ljHeap$ and environment $\ljEnv$ results in final heap $\ljHeap'$, value
$\ljVal$, and monitor $\monitor$. The Figures~\ref{fig:fnference_lcpj}, \ref{fig:inference_function-application},
\ref{fig:inference_property-reference}, and~\ref{fig:inference_property-assignment}  contain
its inference rules.

\begin{figure}
   \centering
   \begin{mathpar}
	  \inferrule [\RuleLjConstant]
	  {
	  }
	  {
		 \ljHeap,\ljEnv ~\entails~ \ljConst ~\eval~ \ljHeap ~|~ \ljConst ~|~ \emptyset
	  }\and
	  \inferrule [\RuleLjVariable]
	  {
	  }
	  {
		 \ljHeap,\ljEnv ~\entails~ \ljVar ~\eval~ \ljHeap ~|~ \ljEnv(\ljVar) ~|~ \emptyset
	  }\and
	  \inferrule [\RuleLjFunctionCreation]
	  {
		 \ljLocation \notin \dom(\ljHeap)
	  }
	  {
		 \ljHeap,\ljEnv ~\entails~ \ljFunction ~\eval~ \ljHeap[\ljLocation\mapsto\langle \ljEnv, \ljFunction \rangle] ~|~ \ljLocation ~|~ \emptyset
	  }\and
	  \inferrule [\RuleLjFunctionApplication]
	  {
		 \ljHeap,\ljEnv ~\entails~ \ljExp_0 ~\eval~ \ljHeap' ~|~ \ljLocation ~|~ \monitor \\\\
		 \ljHeap',\ljEnv ~\entails~ \ljExp_1 ~\eval~ \ljHeap'' ~|~ \ljVal_1 ~|~ \monitor'\\\\
		 \ljHeap'',\ljLocation ~\entailsFA~ \ljVal_1 ~\eval~ \ljHeap''' ~|~ \ljVal ~|~ \monitor''
	  }
	  {
		 \ljHeap,\ljEnv ~\entails~ \ljExp_0(\ljExp_1) ~\eval~ \ljHeap''' ~|~ \ljVal ~|~ \monitor;\monitor';\monitor''
	  }\and
	  \inferrule [\RuleLjObjectCreation]
	  {
		 \ljHeap,\ljEnv ~\entails~ \ljExp ~\eval~ \ljHeap' ~|~ \ljVal ~|~ \monitor\\
		 \ljLocation \notin \dom(\ljHeap')
	  }
	  {
		 \ljHeap,\ljEnv ~\entails~ \ljNew~\ljExp ~\eval~ \ljHeap'[\ljLocation\mapsto\ljVal] ~|~ \ljLocation ~|~ \monitor
	  }\and
	  \inferrule [\RuleLjPropertyReference]
	  {
		 \ljHeap,\ljEnv ~\entails~ \ljExp_{0} ~\eval~ \ljHeap' ~|~ \ljLocation ~|~ \monitor \\\\
		 \ljHeap',\ljEnv ~\entails~ \ljExp_{1} ~\eval~  \ljHeap'' ~|~ \ljStr ~|~ \monitor' \\\\
		 \ljHeap'', \ljLocation ~\entailsPR~ \ljStr ~\eval~  ~|~ \ljHeap''' ~|~ \ljVal ~|~ \monitor''
	  }
	  {
		 \ljHeap,\ljEnv ~\entails~ \ljExp_{0}[\ljExp_{1}] ~\eval~ \ljHeap''' ~|~ \ljVal ~|~ \monitor;\monitor';\monitor''
	  }\and
	  \inferrule [\RuleLjPropertyAssignment]
	  {
		 \ljHeap,\ljEnv ~\entails~ \ljExp_0 ~\eval~ \ljHeap' ~|~ \ljLocation ~|~ \monitor\\\\
		 \ljHeap',\ljEnv ~\entails~ \ljExp_1 ~\eval~ \ljHeap'' ~|~ \ljStr ~|~ \monitor' \\\\
		 \ljHeap'',\ljEnv ~\entails~ \ljExp_2 ~\eval~ \ljHeap''' ~|~ \ljVal ~|~ \monitor'' \\\\
		 \ljHeap''', \ljLocation ~\entailsPA~ \ljStr, \ljVal ~\eval~ \ljHeap'''' ~|~ \ljVal' ~|~ \monitor'''
	  }
	  {
		 \ljHeap,\ljEnv ~\entails~ \ljExp_0[\ljExp_1] = \ljExp_2 ~\eval~ \ljHeap'''' ~|~ \ljVal' ~|~ \monitor;\monitor';\monitor'';\monitor'''
	  }\and
	  \inferrule [\RuleLjPermit]
	  {
		 \ljHeap,\ljEnv ~\entails~ \ljExp ~\eval~ \ljHeap' ~|~ \ljLocation ~|~ \monitor \\\\
		 \ljLocation' \notin \dom(\ljHeap')\\
		 \ljHeap'' = \ljHeap'[\ljLocation'\mapsto \langle \ljLocation, \tpEmpty, \apcContract \rangle]
	  }
	  {
		 \ljHeap,\ljEnv ~\entails~ \ljPermit~ \apcContract ~\ljIn~ \ljExp ~\eval~ \ljHeap'' ~|~ \ljLocation' ~|~ \monitor
	  }
   \end{mathpar}
   \caption{Inference rules for $\lcpj$.}
   \label{fig:fnference_lcpj}
\end{figure}

The rules \Rule{\RuleLjConstant}, \Rule{\RuleLjVariable}, \Rule{\RuleLjFunctionCreation} are
standard. \Rule{\RuleLjObjectCreation} allocates a new object based on the evaluated
prototype. The rule \Rule{\RuleLjPermit} creates a new proxy object for the location
resulting from the subexpression. The handler of this proxy  contains an empty access path $\tpEmpty$ and the initial contract $\apcContract$.

\begin{figure}
   \centering
   \begin{mathpar}
	  \inferrule [\RuleLjFunctionApplicationOne]
	  {
		 \langle \ljObj, \langle \dot{\ljEnv}, \ljFunction \rangle, \ljPrototype \rangle = \ljHeap(\ljLocation)\\\\
		 \ljHeap,\dot{\ljEnv}[\ljVar \mapsto \ljVal] ~\entails~ \ljExp ~\eval~ \ljHeap' ~|~ \ljVal' ~|~ \monitor
	  }
	  {
		 \ljHeap,\ljLocation ~\entailsFA~ \ljVal ~\eval~ \ljHeap' ~|~ \ljVal ' ~|~\monitor
	  }\and
	  \inferrule [\RuleLjFunctionApplicationTwo]
	  {
		 \langle \ljLocation', \tpPath, \apcContract \rangle = \ljHeap(\ljLocation)\\\\
		 \ljHeap, \ljHeap(\ljLocation') ~\entailsFA~ \ljConst ~\eval~ \ljHeap' ~|~ \ljVal' ~|~ \monitor
	  }
	  {
		 \ljHeap,\ljLocation ~\entailsFA~ \ljConst ~\eval~ \ljHeap' ~|~ \ljVal' ~|~ \monitor
	  }\and
	  \inferrule [\RuleLjFunctionApplicationThree]
	  {
		 \langle \ljLocation'', \tpPath, \apcContract \rangle = \ljHeap(\ljLocation)\\
		 \ljLocation''' \notin \dom(\ljHeap')\\\\
		 \ljHeap[\ljLocation'''\mapsto\langle \ljLocation', \tpEmpty, \apcContract \rangle], \ljHeap(\ljLocation'') ~\entailsFA~ \ljLocation''' ~\eval~ \ljHeap' ~|~ \ljVal' ~|~ \monitor
	  }
	  {
		 \ljHeap,\ljLocation ~\entailsFA~ \ljLocation' ~\eval~ \ljHeap' ~|~ \ljVal' ~|~ \monitor
	  }
   \end{mathpar}
   \caption{Inference rules for function application.}
   \label{fig:inference_function-application}
\end{figure}

\begin{figure}
   \centering
   \begin{mathpar}
	  \inferrule [\RuleLjPropertyReferenceOne]
	  {
		 \langle \ljObj, \ljClosure, \ljPrototype \rangle = \ljHeap(\ljLocation)
	  }
	  {
		 \ljHeap,\ljLocation ~\entailsPR~ \ljStr ~\eval~ \ljHeap ~|~ \ljObj(\ljStr) ~|~ \emptyset
	  }\and
	  \inferrule [\RuleLjPropertyReferenceTwo]
	  {
		 \langle \ljLocation', \tpPath, \apcContract \rangle = \ljHeap(\ljLocation)\\
		 \isReadable{\apcContract}{\ljStr}\\\\
		 \ljHeap, \ljLocation' ~\entailsPR~ \ljStr ~\eval~ \ljHeap' ~|~ \ljConst ~|~ \monitor\\
	  }
	  {
		 \ljHeap,\ljLocation ~\entailsPR~ \ljStr ~\eval~ \ljHeap' ~|~ \ljConst ~|~ \addReadPath{\tpPath\tpDot\ljStr}
	  }\and
	  \inferrule [\RuleLjPropertyReferenceThree]
	  {
		 \langle \ljLocation', \tpPath, \apcContract \rangle = \ljHeap(\ljLocation)\\
		 \isReadable{\apcContract}{\ljStr}\\\\
		 \ljHeap, \ljLocation' ~\entailsPR~ \ljStr ~\eval~ \ljHeap' ~|~ \ljLocation'' ~|~ \monitor\\\\
		 \ljProxy = \langle \ljLocation'', \tpPath\tpDot\ljStr, \deriv{\ljStr}{\apcContract}\rangle\\
		 \ljLocation''' \notin \dom(\ljHeap')
	  }
	  {
		 \ljHeap,\ljLocation ~\entailsPR~ \ljStr ~\eval~ \ljHeap'[\ljLocation'''\mapsto \ljProxy] ~|~ \ljLocation''' ~|~ \addReadPath{\tpPath\tpDot\ljStr}
	  }
   \end{mathpar}
   \caption{Inference rules for property reference.}
   \label{fig:inference_property-reference}
\end{figure}

\begin{figure}
   \centering
   \begin{mathpar}
	  \inferrule [\RuleLjPropertyAssignmentOne]
	  {
		 \langle \ljObj, \ljClosure, \ljPrototype \rangle = \ljHeap(\ljLocation)
	  }
	  {
		 \ljHeap,\ljLocation ~\entailsPA~ \ljStr, \ljVal ~\eval~ \ljHeap[\ljLocation,\ljStr\mapsto\ljVal] ~|~ \ljVal ~|~ \emptyset 
	  }\and
	  \inferrule [\RuleLjPropertyAssignmentTwo]
	  {
		 \langle \ljLocation', \tpPath, \apcContract \rangle = \ljHeap(\ljLocation)\\
		 \isWrireable{\apcContract}{\ljStr}\\\\
		 \ljHeap, \ljLocation' ~\entailsPA~ \ljStr, \ljVal ~\eval~ \ljHeap' ~|~ \ljVal' ~|~ \monitor
	  }
	  {
		 \ljHeap,\ljLocation ~\entailsPA~ \ljStr, \ljVal ~\eval~ \ljHeap' ~|~ \ljVal' ~|~ \addWritePath{\tpPath\tpDot\ljStr}
	  }
   \end{mathpar}
   \caption{Inference rules for property assignment.}
   \label{fig:inference_property-assignment}
\end{figure}

Function application, property lookup and property assignment
distinguish two cases: either the operation applies directly to a
target object or it applies to a proxy. If
the given reference is a %function closure or a
non-proxy object, then the usual rules apply:
\Rule{\RuleLjFunctionApplicationOne} for calling the closure and rules \Rule{\RuleLjPropertyReferenceOne} and
\Rule{\RuleLjPropertyAssignmentOne} for property read and write. Otherwise, if the given
reference is a proxy, then the proxy rules formalize the operation
implemented by the trap.

In case of a function proxy, the contract is applied to the argument
to protect the function's input.  If the argument is a constant
\Rule{\RuleLjFunctionApplicationTwo}, then the constant need not be
wrapped and the function application is forwarded to the target
object. Otherwise, if the argument is a location
\Rule{\RuleLjFunctionApplicationThree}, then the argument is wrapped
in a new proxy with the function's contract and then passed to the
target function. This wrapping may happen multiple times.

To perform a property read on a proxy, the handler first checks whether the access is
allowed by the contract. If so and in case the accessed value is a constant, then the value
gets returned in rule \Rule{\RuleLjPropertyReferenceTwo}. Otherwise, if the value is a
location, the location gets wrapped by the derivative of the original
contract with respect to the accessed
property $\deriv{\ljStr}{\apcContract}$ and an extended path $\tpPath\tpDot\ljStr$
\Rule{\RuleLjPropertyReferenceThree}. The accessed object is now wrapped with the derivative
that describes the remaining permitted paths. Thus, a property access
to a wrapped object only returns constants or wrapped objects so that
the membrane remains intact. In both cases the monitor registers the accessed path.

In a similar way, the property assignment checks if the property is writeable with respect
to the given contract \Rule{\RuleLjPropertyAssignmentTwo}. If so, the value gets
assigned and the monitor extended by a write effect.

Both, the get and put rules, signal a violation by being stuck. The
behavior of the implementation is configurable: it may raise an
exception and stop execution or it may just log the violation and continue. This
behavior could be formalized as well by introducing separate crash judgments. We
omit the definition of these judgments because their inference rules largely
duplicate the rules from Figures~\ref{fig:fnference_lcpj},
\ref{fig:inference_function-application}, 
\ref{fig:inference_property-reference},
and~\ref{fig:inference_property-assignment}.

\section{Reduction}
\label{sec:reduction}

As the target object of a proxy may be a proxy itself, there may be a chain of
proxies to traverse before reaching the actual non-proxy 
target object. Such chains waste memory, they increase the run time of
all operations, and the intermediate proxies 
may contain redundant information. To avoid the creation of
inefficient chains of nested proxies, we create a ``proxy of a proxy'' as
follows: the new proxy directly refers to the target of the existing proxy, its
path set is the union of the new path and the already existing pathset, and its
contract is merged with the contract of the already existing
handler. Fortunately, contracts can be merged easily by using 
the conjunction operator $\apcAnd$, which means that the restrictions enforced along
\emph{all} reaching paths are enforced. Extending the formalization, the handler contains a set
of paths, instead of a single path.

However, naively following this approach leads to two problems. First,
a path update operation has to extend every
path in a set and the redundant parts in a set waste a lot of memory. Second, the
combination of contracts may result in a contract whose parts cancel or subsume one another.
Redundant parts in a combined contract also waste memory and make the computation of the derivative
more expensive.
As a result, in our initial experiments, test cases with many objects could not be analyzed in a sensible amount of time. 

This section reports some optimizations to reduce memory consumption and to improve the run time.

% _____    _       ___ _               _                
%|_   _| _(_)___  / __| |_ _ _ _  _ __| |_ _  _ _ _ ___ 
%  | || '_| / -_) \__ \  _| '_| || / _|  _| || | '_/ -_)
%  |_||_| |_\___| |___/\__|_|  \_,_\__|\__|\_,_|_| \___|

\subsection{Trie Structure}
\label{sec:trie}

In a first step, we changed the representation of a set of access paths to a trie structure \cite{Fredkin:1960:TM:367390.367400}.
Let $\Label{\TpTrie} ~\ni~ \tpTrie ::= \emptyset ~|~
\tpTrie[\tpProperty\mapsto\tpTrie']$ be a trie structure used as prefix tree to
store paths $\tpPath$. At each node $\tpTrie$, there is at most one association
$[\tpProperty\mapsto\tpTrie']$ for each property $\tpProperty$. It indicates
that there is a path of the form $\tpProperty.\tpPath$ in the trie, where
$\tpPath$ is in $\tpTrie'$. We use $[\tpEmpty \mapsto \emptyset]$ to mark the
end of a path.  

A path $\tpPath$ in a trie $\tpTrie$ is thus represented by the concatenation of
the properties $\tpProperty$ on a path from the root to an edge labeled with the
empty path $\tpEmpty$. 

%$\tpTrie'=\append{\tpTrie}{\tpPath\tpDot\tpPath'}$ with $\tpTrie'=\append{\append{\tpTrie}{\tpPath}}{\tpPath'}$ denotes the appending of path $\tpPath\tpDot\tpPath'$ at $\tpTrie$. We write $\tpPath\in\tpTrie$ for all paths $\tpPath$ represented by $\tpTrie$. $\tpTrie'' = \append{\tpTrie}{\tpTrie'}$ or $\tpTrie'' = \append{\tpTrie}{\tpPath} ~|~ \tpPath\in\tpTrie'$ denotes the union of the trie $\tpTrie$ and $\tpTrie'$ 

Using trie structures enables us to share prefixes in path sets, which
reduces the memory usage significantly.  In particular, linked data structures
like lists and trees give rise to many shared prefixes which are represented efficiently by
the trie structure.

%  ___         _               _     ___                _ _   _           
% / __|___ _ _| |_ _ _ __ _ __| |_  | _ \_____ __ ___ _(_) |_(_)_ _  __ _ 
%| (__/ _ \ ' \  _| '_/ _` / _|  _| |   / -_) V  V / '_| |  _| | ' \/ _` |
% \___\___/_||_\__|_| \__,_\__|\__| |_|_\___|\_/\_/|_| |_|\__|_|_||_\__, |
%                                                                   |___/ 

\subsection{Contract Rewriting}
\label{sec:contract_reduction}

Besides the normalization of contracts (Section \ref{sec:contracts})
the containment reduction of contracts is important to reduce memory consumption.
Containment reduction means that a contract $\apcContract\apcOr\apcContract'$ can be reduced
to $\apcContract$ if $\lang{\apcContract'}\subseteq\lang{\apcContract}$. Similarly,
$\apcContract\apcAnd\apcContract'$ can be reduced to $\apcContract'$ if
$\lang{\apcContract'}\subseteq\lang{\apcContract}$. The implementation also accounts for
commutativity. 

First, we define a semantic containment relation on contracts.

\begin{definition}[Containment]\label{def:containment}
   A contract $\apcContract$ is contained in another contract $\apcContract'$, written as
   $\isSubSetOf{\apcContract}{\apcContract'}$,  iff   
   $\lang{\apcContract}\subseteq\lang{\apcContract'}$.
   % We follow that:
   % \begin{gather}
   %    \isSubSetOf{\apcContract}{\apcContract'} ~\Leftrightarrow~ 
   %    \lang{\apcContract}\subseteq\lang{\apcContract'}
   % \end{gather}
\end{definition}
The containment relation on contracts is reflexive and transitive and thus forms a preorder.
As it is defined semantically, we need a syntactic decision procedure for containment of
contracts to put it to use. This procedure is inspired by Antimirov's calculus \cite{Antimirov95:0}, which
provides a non-deterministic decision procedure to solve the containment problem for
ordinary regular expressions.

From the definition of containment (Definition~\ref{def:containment}),
we obtain that $\apcContract \sqsubseteq \apcContract'$ iff for all paths $\tpPath\in\lang{\apcContract}$ the derivative of contract $\apcContract'$ with respect to path $\tpPath$ is nullable $\nullable{(\deriv{\tpPath}{\apcContract'})}$.

\begin{lemma}[Containment]\label{thm:containment_nullable}
  %%% PJT: you just said that
   % A contract $\apcContract$ is subset or equals to contract $\apcContract'$ iff the
   % derivation of $\apcContract'$ with respect to all path $\tpPath\in\lang{\apcContract}$
   % defined in the language of $\apcContract$ is nullable. 
   \begin{gather}
	  \isSubSetOf{\apcContract}{\apcContract'} ~\Leftrightarrow~ 
	  \nullable{(\deriv{\tpPath}{\apcContract'})} \text{~for all~} \tpPath\in\lang{\apcContract}
   \end{gather}
\end{lemma}

The proof is by Lemma~\ref{thm:nullable} and~\ref{thm:derivation}.
%%% PJT: this is trivial by definition of containment
%We observe that all paths $\tpPath$ in $\lang{\apcContract}$ are also element of $\lang{\apcContract'}$.

As this lemma does not yield an effective way of deciding containment,
we aim to enumerate the access paths of $\apcContract$ by iteratively
extracting its possible first 
properties and forming the derivatives on both sides as in Antimirov's procedure.

\begin{definition}[$\firstP$] \label{def:firstp}
   The function $\firstP:\apcContract\rightarrow\apcAlphabet$ returns the
   first properties $\tpProperty$ of path elements $\tpPath$ in $\lang{\apcContract}$.
   \begin{gather}
	  \getFirstP{\apcContract} ~=~ \{\tpProperty ~|~ \tpProperty\tpDot\tpPath\in\lang{\apcContract}\}
   \end{gather}
\end{definition}

From Antimirov's theorem \cite{Antimirov95:0}, we obtain that
$\isSubSetOf{\apcContract}{\apcContract'}$ iff $\isNullable{\apcContract}$
implies $\isNullable{\apcContract'}$ and $\forall\tpProperty:$
$\isSubSetOf{\deriv{\tpProperty}{\apcContract}}{\deriv{\tpProperty}{\apcContract'}}$.

\begin{lemma}[Containment2]\label{thm:containment2}
  % PJT: just said that
   % If $\forall\tpProperty\in\getFirstP{\apcContract}:$
   % $\isSubSetOf{\deriv{\tpProperty}{\apcContract}}{\deriv{\tpProperty}{\apcContract'}}$ that $\isSubSetOf{\apcContract}{\apcContract'}$.
   \begin{align}
	  \begin{split}
		 \isSubSetOf{\apcContract}{\apcContract'}
		 ~\Leftrightarrow&~
		 (\forall \tpProperty\in\getFirstP{\apcContract})~
		 \isSubSetOf{\deriv{\tpProperty}{\apcContract}}{\deriv{\tpProperty}{\apcContract'}}\\
		 ~&~\wedge~ (\isNullable{\apcContract}\Rightarrow\isNullable{\apcContract'})
	  \end{split}
   \end{align}
\end{lemma}

The proof combines the assertion of Lemma~\ref{thm:containment_nullable} with the stepwise
derivative of Lemma~\ref{thm:derivation}.

Unfortunately, it is in general impossible to construct  all derivatives with respect to all
first properties. In case of a question mark literal or a negation, there may be infinitely
many first literals. Thus, the implementation needs to apply some kind of approximation.

% _    _ _                _        _         _          _   _                 __    ___         _               _      
%| |  (_) |_ ___ _ _ __ _| |___ __| |___ _ _(_)_ ____ _| |_(_)___ _ _    ___ / _|  / __|___ _ _| |_ _ _ __ _ __| |_ ___
%| |__| |  _/ -_) '_/ _` | |___/ _` / -_) '_| \ V / _` |  _| / _ \ ' \  / _ \  _| | (__/ _ \ ' \  _| '_/ _` / _|  _(_-<
%|____|_|\__\___|_| \__,_|_|   \__,_\___|_| |_|\_/\__,_|\__|_\___/_||_| \___/_|    \___\___/_||_\__|_| \__,_\__|\__/__/

\subsubsection{Literal-derivative of Contracts}
\label{sec:literal_derivation}

To abstract from the derivative of a contract with respect to a property, we
introduce a literal-based derivative which forms the derivative of a
contract with respect to a literal $\apcLiteral$. This operation
admits derivatives with respect to the contract literals
$\apcAT$, $\apcQMark$, $\apcRegEx$, and
$\apcNeg\apcRegEx$, and performs some approximation by
returning a contract that is contained in the 
derivatives resulting from expanding the literals. 

\begin{figure}
   \centering
   \begin{displaymath}
	  \begin{array}{lll}
		 \lderiv{\apcLiteral}{\apcAT} &=&  \begin{cases}
			\apcEmpty, & \apcLiteral=\apcAT\\
			\apcEmptySet, & \text{otherwise}
		 \end{cases}\\

		 \lderiv{\apcLiteral}{\apcQMark} &=&             \apcEmpty\\

		 \lderiv{\apcLiteral}{\apcRegEx} &=& \begin{cases}
			\apcEmpty, & \regexIsSubSetOf{\apcLiteral}{\apcRegEx}\\
			\apcEmptySet, & \text{otherwise}
		 \end{cases}\\

		 \lderiv{\apcLiteral}{\apcNeg\apcRegEx} &=&  \begin{cases}
			\apcEmpty, &
			\regexIsCap{\apcLiteral}{\apcRegEx} = \apcEmptySet\\
			\apcEmptySet, & \text{otherwise}
		 \end{cases}\\

		 \lderiv{\apcLiteral}{\apcEmptySet} &=& \apcEmptySet\\

		 \lderiv{\apcLiteral}{\apcEmpty} &=& \apcEmptySet\\

%space         \lderiv{\apcLiteral}{\apcContract\apcQMark} &=& \lderiv{\apcLiteral}{\apcContract}\\

		 \lderiv{\apcLiteral}{\apcContract\apcStar} &=& 
		 \lderiv{\apcLiteral}{\apcContract}\apcDot\apcContract\apcStar\\

		 \lderiv{\apcLiteral}{\apcContract\apcOr\apcContract'} &=&
		 \lderiv{\apcLiteral}{\apcContract} \apcOr \lderiv{\apcLiteral}{\apcContract'}\\

		 \lderiv{\apcLiteral}{\apcContract\apcAnd\apcContract'} &=& 
		 \lderiv{\apcLiteral}{\apcContract} \apcAnd \lderiv{\apcLiteral}{\apcContract'}\\

		 \lderiv{\apcLiteral}{\apcContract\apcDot\apcContract'} &=& \begin{cases}
			\lderiv{\apcLiteral}{\apcContract}\apcDot\apcContract'\apcOr\lderiv{\apcLiteral}{\apcContract'}, &\isNullable{\apcContract}\\
			\lderiv{\apcLiteral}{\apcContract}\apcDot\apcContract' , &\text{otherwise}
		 \end{cases}\\

	  \end{array}
   \end{displaymath}
   \caption{Derivative of a contract by a contract literal.}
   \label{fig:literal_derivation}
\end{figure}

As the literals include character-level regular expressions that we
treat as abstract, we rely on a relation  $\regexSubsetof$ that
decides containment of literals and on an operation $\regexCap$ that
builds the intersection of two literals. Both are considered as
abstract operators on the language of regular expression literals and
we do not specify their implementation (but similar methods as for
contracts apply).
\begin{align}
   \regexIsSubSetOf{\apcLiteral}{\apcRegEx} ~&\Leftrightarrow~ \lang{\apcLiteral}\subseteq\lang{\apcRegEx}\\
   \lang{\regexIsCap{\apcLiteral}{\apcRegEx}} ~&=~ \lang{\apcLiteral}\cap\lang{\apcRegEx}
\end{align}

Figure~\ref{fig:literal_derivation} contains the definition of the
syntactic derivative with respect to a literal. 
% The literal-based derivation of contract $\apcContract$ w.r.t. literal $\apcLiteral$ is an approximation of the derivation of $\apcContract$ w.r.t. to property $\tpProperty\in\lang{\apcLiteral}$. It is based on the same fundamentals as the path-based derivation, meaning that the resulting contract describes the remaining paths after reading an element defined by the language of $\apcLiteral$. 
%
Deriving a literal $\apcLiteral$ w.r.t. itself results in the empty
contract $\apcEmpty$. Applying any derivative to the empty set $\apcEmptySet$ yields the empty set $\apcEmptySet$.
The derivative of a regular expression literal $\apcRegEx$ w.r.t. a literal
$\apcLiteral$ is the empty contract $\apcEmpty$ if the language of
$\apcLiteral$ is subsumed by the regular expression $\apcRegEx$, that
is, if we can make a step with each character in the
literal. Similarly, the derivative of a negated regular expression
literal $\apcNeg\apcRegEx$ w.r.t. a literal $\apcLiteral$ is the empty
contract $\apcEmpty$ if no property of the language of $\apcLiteral$
can make a step in $\apcRegEx$. Otherwise, the result is
$\apcEmptySet$. The remaining cases are exactly as in the derivative
with respect to a property.

The following lemma states the connection between the derivative by
contract literal and the derivative by a property.

\begin{lemma}[Syntactic derivative of contracts]\label{def:literal_derivation}  $\forall\apcLiteral:$
   \begin{gather}
	  \lang{\lderiv{\apcLiteral}{\apcContract}} ~\subseteq~ \bigcap_{\tpProperty\in\lang{\apcLiteral}} \lang{\deriv{\tpProperty}{\apcContract}}
   \end{gather}
\end{lemma}

%\begin{lemma}[Correctness]\label{def:correctness}
%The literal-derivation is defined in terms of the property-derivation.
%   \begin{gather}
%         \lang{\apcLiteral\apcDot\lderiv{\apcLiteral}{\apcContract}}\subseteq\lang{\apcContract}
%   \end{gather}
%\end{lemma}

%\todoInline{Correctness relation/ Abstraction}
%\todoInline{Write something about the relation to the normal derivation/ Bi-simulation}

%  ___         _        _                    _      ___      _         _         
% / __|___ _ _| |_ __ _(_)_ _  _ __  ___ _ _| |_   / __|__ _| |__ _  _| |_  _ ___
%| (__/ _ \ ' \  _/ _` | | ' \| '  \/ -_) ' \  _| | (__/ _` | / _| || | | || (_-<
% \___\___/_||_\__\__,_|_|_||_|_|_|_\___|_||_\__|  \___\__,_|_\__|\_,_|_|\_,_/__/

\subsubsection{Containment}
\label{sec:containmelnt}

Using the syntactic derivation, we are able to abstract the derivative
of a contract w.r.t. an infinite property set by forming derivatives with respect to literals. Before combining this abstraction with the containment lemma (Lemma \ref{thm:containment2}), we define a function to obtain the
first literals of a contract as shown in Figure~\ref{fig:firstc}.

\begin{figure}
   \centering
   \begin{displaymath}
	  \begin{array}{lll}

		 \getFirstC{\apcAT} &=& \{\apcAT\}\\
		 \getFirstC{\apcQMark} &=& \{\apcQMark\}\\
		 \getFirstC{\apcRegEx} &=& \{\apcRegEx\}\\
		 \getFirstC{\apcNeg\apcRegEx} &=& \{\apcNeg\apcRegEx\}\\
		 \getFirstC{\apcEmptySet} &=& \{\}\\
		 \getFirstC{\apcEmpty} &=& \{\}\\
%space                                 \getFirstC{\apcContract\apcQMark} &=& \getFirstC{\apcContract}\\
		 \getFirstC{\apcContract\apcStar} &=& \getFirstC{\apcContract}\\
		 \getFirstC{\apcContract\apcOr\apcContract'} &=& \getFirstC{\apcContract} \cup \getFirstC{\apcContract'}\\
		 \getFirstC{\apcContract\apcAnd\apcContract'} &=& 
		 \{ \regexIsCap{\apcLiteral}{\apcLiteral'} ~|~ \apcLiteral\in\getFirstC{\apcContract}, \apcLiteral'\in\getFirstC{\apcContract'} \}\\
		 \getFirstC{\apcContract\apcDot\apcContract'} &=& \begin{cases}
			\getFirstC{\apcContract}\cup\getFirstC{\apcContract'}, & \isNullable{\apcContract}\\
			\getFirstC{\apcContract}, & \text{otherwise}
		 \end{cases}\\

	  \end{array}
   \end{displaymath}
   \caption{$\firstC$ on contracts.}
   \label{fig:firstc}
\end{figure}

The first literal of a literal $\apcLiteral$ is
$\apcLiteral$. $\apcEmptySet$ and $\apcEmpty$ have no first literals. The first literals
of
% space% an option contract $\apcContract\apcQMark$ or
a Kleene star
contract $\apcContract\apcStar$ are the first literals of its
subcontract. The first literals of a disjunction are the union of the
first literals of its subcontracts. For a conjunction,  the set of
first literals is the set of all intersections
$\regexIsCap{\apcLiteral}{\apcLiteral'}$ of the first literals of both conjuncts. The first literals of a concatenation are the first literals of its first subcontract if the first subcontract is not nullable. Otherwise, it is the union of the first literals of both subcontracts.

The language of the first literals is defined to be the union of the languages of its literals.
\begin{equation}
   \lang{\getFirstC{\apcContract}}=\bigcup_{\apcLiteral\in\getFirstC{\apcContract}} \lang{\apcLiteral}
\end{equation}

\begin{lemma}[$\firstC$] \label{def:firstc}
   % The function $\firstC:\apcContract\rightarrow\{\apcLiteral\ldots\}$ returns the first literals $\apcLiteral$ of contract $\apcContract$.
   % The set of first-properties $\getFirstP{\apcContract}$ is subset or equals to the language defined by the first-literals $\lang{\getFirstC{\apcContract}}$.
   \begin{gather}
	  \getFirstP{\apcContract} ~=~ \lang{\getFirstC{\apcContract}}
   \end{gather}
\end{lemma}

% RMK: definition not required/ not used anymore!
%We also have to gather first literals for inqualities. The first literals of the inequality $\isSubSetOf{\apcContract}{\apcContract'}$ are the first literals of the left-hand side contract.
%%
%%the union of the first literals of the left-hand side contract and the first literals of the negated contract.
%%
%\begin{equation}
%   \getFirstC{\isSubSetOf{\apcContract}{\apcContract'}} ~=~ \getFirstC{\apcContract} %\cup\getFirstC{\apcContract'}
%\end{equation}

% This is required because the set of first properties of a negation may be infinite. A negation reduces the universe by the elements of its subcontract. Therefore, the excluded paths--of the right-hand side--have to be involved. 

\begin{lemma}[Syntactic derivative of contracts 2]\label{def:literal_derivation2} $\forall \apcLiteral\in\getFirstC{\apcContract}:$
   \begin{gather}
	  \lang{\lderiv{\apcLiteral}{\apcContract}} ~=~ \bigcap_{\tpProperty\in\lang{\apcLiteral}} \lang{\deriv{\tpProperty}{\apcContract}}
   \end{gather}
\end{lemma}

\begin{theorem}[Containment]\label{thm:containment}
   % A contract $\apcContract$ is subset $\subsetof$ of a contract $\apcContract'$, or $\apcContract'$ is superset $\supersetof$ of $\apcContract$, iff the subset relation is satisfied for all derivations of $\apcContract$ and $\apcContract'$ with respect to the first literals of $\apcContract$ and $\apcContract'$.
   \begin{gather}
	  \begin{split}
		 \isSubSetOf{\apcContract}{\apcContract'} ~\Leftarrow~ 
		 &(\forall
		 \apcLiteral\in\getFirstC{\isSubSetOf{\apcContract}{\apcContract'}})~
		 \isSubSetOf{\lderiv{\apcLiteral}{\apcContract}}{\lderiv{\apcLiteral}{\apcContract'}}
		 \\
		 ~&              ~\wedge~ (\isNullable{\apcContract}\Rightarrow\isNullable{\apcContract'})
	  \end{split}
   \end{gather}
\end{theorem}

% The proof uses Lemmas~\ref{thm:containment2}
% and~\ref{def:literal_derivation} and performs a bisimulation
% w.r.t. the first literals to determine if all elements of the
% left-hand side are subsumed by the right-hand side.  

%  ___         _        _                    _     ___                     _   _       
% / __|___ _ _| |_ __ _(_)_ _  _ __  ___ _ _| |_  / __| ___ _ __  __ _ _ _| |_(_)__ ___
%| (__/ _ \ ' \  _/ _` | | ' \| '  \/ -_) ' \  _| \__ \/ -_) '  \/ _` | ' \  _| / _(_-<
% \___\___/_||_\__\__,_|_|_||_|_|_|_\___|_||_\__| |___/\___|_|_|_\__,_|_||_\__|_\__/__/

\subsubsection{Containment Semantics}
\label{sec:containmelnt_semantics}

Based on the containment theorem (Theorem~\ref{thm:containment}) and the syntactic derivative, we present an algorithm that approximates the containment relation of contracts.
%decision procedure for the containment relation of contracts. 
%%% PJT: is this a decision procedure?

To recapitulate, a path $\tpPath$ is an element of the language defined by a contract $\apcContract$ iff the derivative of contract $\apcContract$ w.r.t. path $\tpPath$ is nullable $\nullable{(\deriv{\tpPath}{\apcContract})}$. If a contract $\apcContract$ is not contained in a contract $\apcContract'$, then there exists at least one access path $\tpPath\in\lang{\apcContract}$ that derives a non-nullable contract from $\apcContract'$.

Define containment expressions by $\ccExp ~::=~
\isSubSetOf{\apcContract}{\apcContract'}$ and let a context
$\ccContext$ be a set of previous visited containment expressions. The
context lookup $\inCcContext{\ccExp}$ determines if an expression is
already calculated in this chain. The derivation of a containment
expression
$\lderiv{\apcLiteral}{\isSubSetOf{\apcContract}{\apcContract'}}$ is
forwarded to its subcontracts as $\isSubSetOf{\lderiv{\apcLiteral}{\apcContract}}{\lderiv{\apcLiteral}{\apcContract'}}$.

%
%\begin{gather}
%   \inCcContext{\ccExp} ~::=~ \begin{cases}
%      \perp, & \ccContext=\emptyset\\
%      \top, & \ccContext=\langle\ccContext',\ccExp'\rangle ~\wedge~ \ccExp=\ccExp'\\
%      \inCcContextPrime{\ccExp}, & \ccContext=\langle\ccContext',\ccExp'\rangle ~\wedge~ \ccExp\neq\ccExp'
%   \end{cases}
%\end{gather}
%
The decision procedure is defined by a judgment of the form $\ccContext ~\entails~ \ccExp ~:~ \{\top, \perp\}$. The evaluation of expression $\ccExp$ in context $\ccContext$ results either in true $\top$ or false $\perp$. 
%
%\begin{figure}
%   \centering
%   \begin{displaymath}
%      \begin{array}{lll}
%
%        \lderiv{\apcLiteral}{\isSubSetOf{\apcContract}{\apcContract'}}
%        &=& 
%        \isSubSetOf{\lderiv{\apcLiteral}{\apcContract}}{\lderiv{\apcLiteral}{\apcContract'}}
%         \\
%
%         \lderiv{L\subseteq\ApcLiteral}{\ccExp} &=& \bigwedge_{\apcLiteral\in L} \lderiv{\apcLiteral}{\ccExp}
%
%      \end{array}
%   \end{displaymath}
%   \caption{Containment (notation).}
%   \label{fig:cc_notations}
%\end{figure}
%
%Figure~\ref{fig:cc_notations} defines some abbreviations. A derivation of an inquality is forwarded to its subcontracts.
%The approximative derivative of a set of literals is the conjunction of all partial derivatives.  
%
\begin{figure}
   \centering
   \begin{mathpar}
	  \inferrule [\RuleCCDisprove]
	  {
		 \isNullable{\apcContract}\\
		 \neg\isNullable{\apcContract'}
	  }
	  {
		 \ccContext ~\entails~ \isSubSetOf{\apcContract}{\apcContract'} ~:~ \perp
	  }\and
	  \inferrule [\RuleCCContext]
	  {
		 \inCcContext{\ccExp}
	  }
	  {
		 \ccContext ~\entails~ \ccExp ~:~ \top
	  }\and
	  \inferrule [\RuleCCUnfoldTrue]
	  {
		 \negInCcContext{\ccExp}\\
		 \forall\apcLiteral\in\getFirstC{\ccExp}:~ \langle\ccContext,\ccExp\rangle ~\entails~ \lderiv{\apcLiteral}{\ccExp} ~:~ \top
	  }
	  {
		 \ccContext ~\entails~ \ccExp ~:~ \top
	  }\and
	  \inferrule [\RuleCCUnfoldFalse]
	  {
		 \negInCcContext{\ccExp}\\
		 \exists\apcLiteral\in\getFirstC{\ccExp}:~ \langle\ccContext,\ccExp\rangle ~\entails~ \lderiv{\apcLiteral}{\ccExp} ~:~ \perp
	  }
	  {
		 \ccContext ~\entails~ \ccExp ~:~ \perp
	  }
   \end{mathpar}
   \caption{Unfolding axioms and rules.}
   \label{fig:cc_unfold}
\end{figure}
In Figure~\ref{fig:cc_unfold}, rule \Rule{\RuleCCDisprove} shows the
generalized disproving axiom. If the first contract $\apcContract$ is
nullable and the second contract $\apcContract'$ is not nullable, then
there exists at least one element--the empty access path
$\apcEmpty$--in the language of $\lang{\apcContract}$ which is not
element of $\lang{\apcContract'}$. This condition is sufficient to
disprove the inequality, so the rule returns false.
Rule \Rule{\RuleCCContext} returns true if the evaluated expression is
already subsumed by the context. Further derivatives of $\ccExp$ would
not contribute new information. \Rule{\RuleCCUnfoldTrue} and \Rule{\RuleCCUnfoldFalse} applies only if
$\ccExp$ is not in the context. It applies all derivatives according to $\getFirstC{\ccExp}$ and conjoins them together. 
\begin{theorem}[Correctness]\label{thm:correctness}
   % We assume $\forall\ccContext,\apcContract,\apcContract'$ if 
   % $\ccContext ~\entails~ \isSubSetOf{\apcContract}{\apcContract'} ~:~ \top$ than $\isSubSetOf{\apcContract}{\apcContract'}$.
   \begin{gather}
	  \ccContext ~\entails~ \isSubSetOf{\apcContract}{\apcContract'} ~:~ \top
	  ~\Rightarrow~
	  \isSubSetOf{\apcContract}{\apcContract'}
   \end{gather}
\end{theorem}
%

%\begin{lemma}[Unfolding]\label{thm:unfolding}
%Suppose that $\inCcContext{\ccExp}=\top$ than $\ccContext,\ccExp ~\entails~ \top$.
%Every unfolding chain of containment expressions $\ccExp$ recurring in $\ccExp$ fulfills the containment calculus \eqref{thm:cc-membership}.
%\begin{gather}
%\ccContext ~\entails~ \ccExp ~:~ \langle\ccContext,\ccExp\rangle ~\entails~ \ccExp' ~:~ \ldots ~:~
%\ccContext' ~\entails~ \ccExp ~\entails~ \top
%\end{gather}
%\end{lemma}

In addition to  the rules from Figure~\ref{fig:cc_unfold}, we add auxiliary axioms to detect
trivially consistent (inconsistent) inequalities early
(Figures~\ref{fig:cc_prove} for consistent inequalities
and~\ref{fig:cc_disprove} for inconsistent ones). They decide
containment directly without unfolding.
The axioms rely on four
functions that operate on the contract syntax: $\blankReducible$, $\emptyReducible$,
$\indifferentReducible$, and $\universalReducible$. Each of them is
correct, but not complete.
For instance, the function $\emptyReducible$ can only approximate whether the
language of a contract of the form $\apcContract\apcAnd\apcContract'$
denotes the empty set.  The following definition specifies these functions, their
actual definitions are straightforward and thus elided.

\begin{definition}
   A contract $\apcContract$ is \dots
   \begin{description}
	  \item[\emph{blank}] if $\lang{\apcContract}=\{\tpBlank\}$,
		 let    $\isBlank{\apcContract}  =  \top$ imply $\apcContract$ is blank;
	  \item[\emph{empty}] if $\lang{\apcContract}=\emptyset$,
		 let $\isEmpty{\apcContract} = \top  \Rightarrow$ imply $\apcContract$ is empty;
	  \item[\emph{indifferent}] if
		 $\lang{\apcContract}=\apcAlphabet$,
		 let $\isIndifferent{\apcContract} = \top$ imply $\apcContract$ is indifferent;
	  \item[\emph{universal}] if $\lang{\apcContract}=\apcWords$,
		 let $\isUniversal{\apcContract} = \top$ imply $\apcContract$ is universal.
   \end{description}
\end{definition}

\begin{figure}
   \centering
   \begin{mathpar}
	  \inferrule [\RuleCCIdentity]
	  {
	  }
	  {
		 \ccContext ~\entails~ \isSubSetOf{\apcContract}{\apcContract} ~:~ \top
	  }\and
	  \inferrule [\RuleCCEmpty]
	  {
		 \isEmpty{\apcContract}\vee \isUniversal{\apcContract'}
	  }
	  {
		 \ccContext ~\entails~ \isSubSetOf{\apcContract}{\apcContract'} ~:~ \top
	  }\and
	  \inferrule [\RuleCCNullable]
	  {
		 \isNullable{\apcContract'}
	  }
	  {
		 \ccContext ~\entails~ \isSubSetOf{\apcEmpty}{\apcContract'} ~:~ \top
	  }
	  %  \and
	  % \inferrule [\RuleCCUniversal]
	  % {
	  %    \isUniversal{\apcContract'}
	  % }
	  % {
	  %    \ccContext ~\entails~ \isSubSetOf{\apcContract}{\apcContract'} ~:~ \top
	  % }
   \end{mathpar}
   \caption{Prove axioms.}
   \label{fig:cc_prove}
\end{figure}

\begin{figure}
   \centering
   \begin{mathpar}
	  \inferrule [\RuleCCBlank]
	  {
		 \neg\isEmpty{\apcContract}\\
		 \isEmpty{\apcContract'}
	  }
	  {
		 \ccContext ~\entails~ \isSubSetOf{\apcContract}{\apcContract'} ~:~ \perp
	  }\and
	  \inferrule [\RuleCCIndifferentTwo]
	  { 
		 \isIndifferent{\apcContract}\vee \isUniversal{\apcContract}\\
		 \isBlank{\apcContract'}
	  }
	  {
		 \ccContext ~\entails~ \isSubSetOf{\apcContract}{\apcContract'} ~:~ \perp
	  }
	  %  \and
	  % \inferrule [\RuleCCUniversalTwo]
	  % { 
	  %    \isUniversal{\apcContract}\\
	  %    \isBlank{\apcContract'}
	  % }
	  % {
	  %    \ccContext ~\entails~ \isSubSetOf{\apcContract}{\apcContract'} ~:~ \perp
	  % }\and
   \end{mathpar}
   \caption{Disprove axioms.}
   \label{fig:cc_disprove}
\end{figure}

% With these functions we can state some auxiliary axioms
% to decide the 
% containment relation without unfolding.
% The self-explanatory rules detect trivially consistent (inconsistent) inequalities which are
% directly mapped to true (false).

% ___        _         _   _               ___      _        
%| _ \___ __| |_  _ __| |_(_)___ _ _  ___ | _ \_  _| |___ ___
%|   / -_) _` | || / _|  _| / _ \ ' \(_-< |   / || | / -_|_-<
%|_|_\___\__,_|\_,_\__|\__|_\___/_||_/__/ |_|_\\_,_|_\___/__/

\subsubsection{Reductions Rules}
\label{sec:js_reductions}

Finally, we apply the preceding machinery to define a reduction function $\reduce{\cdot}$ on contracts in Figure~\ref{fig:reduction_rules}. Reduction produces an equivalent contract, which is smaller than its input. 

% The reduction $\contractReduct:~ \apcContract\rightarrow\apcContract'$ is defined as minimization of contract $\apcContract$ to a language equivalent contract $\apcContract'$. For short we use $\reduce{\apcContract}$ as $\contractReduct(\apcContract)$.

\begin{figure}
   \centering
   \begin{displaymath}
	  \begin{array}{lll}

% MK: reduced for space reasons
%         \reduce{\apcEmptySet} &=& \apcEmptySet\\
%         \reduce{\apcAT} &=& \apcAT\\
%         \reduce{\apcQMark} &=& \apcQMark\\
%         \reduce{\apcRegEx} &=& \apcRegEx\\
%                 \reduce{\apcNeg\apcRegEx} &=& \apcNeg\apcRegEx\\
%
%         \reduce{\apcEmpty} &=& \apcEmpty\\

%space
		 % \reduce{\apcContract\apcQMark} &=& \begin{cases}
		 %    \apcEmpty, & \isEmpty{\apcContract} ~\vee~ \isBlank{\apcContract}\\
		 %    \reduce{\apcContract}\apcQMark, & \text{otherwise}
		 % \end{cases}\\

		 \reduce{\apcContract\apcStar} &=& \begin{cases}
			\apcEmpty, & \isEmpty{\apcContract} ~\vee~ \isBlank{\apcContract}\\
			\reduce{\apcContract}\apcStar, & \text{otherwise}
		 \end{cases}\\

		 \reduce{\apcContract\apcOr\apcContract'} &=& \begin{cases}
			\apcEmptySet, & \isEmpty{\apcContract}\wedge\isEmpty{\apcContract'}\\
			\apcAT, & \isBlank{\apcContract}\wedge\isBlank{\apcContract'}\\
			\reduce{\apcContract}, & \isSuperSetOf{\apcContract}{\apcContract'}\\
			\reduce{\apcContract'}, & \isSubSetOf{\apcContract}{\apcContract'}\\
			\reduce{\apcContract}\apcOr\reduce{\apcContract'}, & \text{otherwise}
		 \end{cases}\\

		 \reduce{\apcContract\apcAnd\apcContract'} &=& \begin{cases}
			\apcEmptySet, & \isEmpty{\apcContract}\vee\isEmpty{\apcContract'}\\
			\apcAT, & \isBlank{\apcContract}\vee\isBlank{\apcContract'}\\
			\reduce{\apcContract}, & \isSubSetOf{\apcContract}{\apcContract'}\\
			\reduce{\apcContract'}, & \isSuperSetOf{\apcContract}{\apcContract'}\\
			\reduce{\apcContract}\apcAnd\reduce{\apcContract'}, & \text{otherwise}
		 \end{cases}\\

		 \reduce{\apcContract\apcDot\apcContract'} &=& \begin{cases}
			\apcEmptySet, & \isEmpty{\apcContract}\\
			\apcAT, & \isBlank{\apcContract}\\
			\reduce{\apcContract}\apcDot\reduce{\apcContract'}, & \text{otherwise}
		 \end{cases}

	  \end{array}
   \end{displaymath}
   \caption{Reduction rules.}
   \label{fig:reduction_rules}
\end{figure}

Literals and empty contracts are not further reducible. A
% space% optional contract $\apcContract\apcQMark$ and a
Kleene star contract $\apcContract\apcStar$ is reduced to the empty contract $\apcEmpty$ if the subcontract is either empty or blank. The disjunction contract $\apcContract\apcOr\apcContract'$ is reduced to the empty set $\apcEmptySet$ if both contracts are empty or it reduces to the empty literal $\apcAT$ if both contracts are blank. If one of the subcontracts is subsumed by the other one, the subsuming contract is used. Similarly, the conjunction contract $\apcContract\apcAnd\apcContract'$ is reduced to the empty set $\apcEmptySet$ or the empty literal $\apcAT$  if one of the subcontracts is empty or blank. If one contract subsumes the other, then the subsumed contract is used. The concatenation $\apcContract\apcDot\apcContract'$ is reduced to the empty set $\apcEmptySet$ if the first subcontract is empty or it is reduced to the empty literal $\apcAT$ if the first sub-contract is blank.

% _____                 _                           _        _   _             
%|_   _|               | |                         | |      | | (_)            
%  | |  _ __ ___  _ __ | | ___ _ __ ___   ___ _ __ | |_ __ _| |_ _  ___  _ __  
%  | | | '_ ` _ \| '_ \| |/ _ \ '_ ` _ \ / _ \ '_ \| __/ _` | __| |/ _ \| '_ \ 
% _| |_| | | | | | |_) | |  __/ | | | | |  __/ | | | || (_| | |_| | (_) | | | |
%|_____|_| |_| |_| .__/|_|\___|_| |_| |_|\___|_| |_|\__\__,_|\__|_|\___/|_| |_|
%                | |                                                           
%                |_|                                                           

\section{Implementation}
\label{sec:implementation}

The implementation is based on the JavaScript Proxy API
\cite{VanCutsem:2010:PDP:1869631.1869638, van2012design}, a proposed
addition to the JavaScript standard.  This API is implemented in
Firefox since version 18.0 and in Chrome V8 since version 3.5.
We developed the implementation using the \emph{SpiderMonkey JavaScript-C 1.8.5
(2011-03-31)} JavaScript engine.
% JavaScript proxies can be used to enhance the functionality of objects--e.g. by tampering with property handling and function calls. A \emph{trap}-functions is available for each kind of operation callable on objects.

\subsection{Description}
\label{sec:description}

The implementation provides a proxy handler \texttt{AccessHandler} that overrides all
trap functions. The traps implement the access control
mechanism as well as path monitoring. They either interrupt the
operation, if it is not permitted, or forward it to the target object.
They maintain the path set and contract data structures using the
fly-weight pattern to minimize memory consumption.   

Our framework can easily be included in existing JavaScript software
projects. Its functionality is encapsulated in a facade whose
interface--the function \texttt{permit}--can be used to wrap objects.  

The framework provides two evaluation modes, \emph{Observer Mode} and
\emph{Protector Mode}. The \emph{Observer Mode} 
performs only path and violation logging without changing the semantics of the
underlying program. Thus, if a program reads multiple properties along a
prohibited path, then each individual read is logged as a violation.
For example, suppose an object is protected by the contract
\lstinline{|'b+c'|}. Reading property \lstinline$a$ results in a violation with
access path \lstinline{a} and a subsequent read of \lstinline$a.b$ results in a
violation of \lstinline$a.b$, and so on.

The \emph{Protector Mode} follows the scripting-language philosophy as implemented
in the rest of JavaScript. If a read access violates the contract of
an object, the value \texttt{undefined} is returned instead of an
abnormal termination. Forbidden write accesses are simply omitted.
Thus, only top-level violations are visible.  

% This behavior. A script should continue, whenever that is
% possible. It should only stop if no sensible continuation is
% possible. The motivation for this philosophy is that  it may not be advisable
% in big software systems to stop the whole system because of a local
% error. Returning \texttt{undefined} opens up the possibility that the
% system reacts to the local error, potentially corrects it, and continues to
% run. In our case,  the protected value remains inaccessible and only
% values depending on the ``illegal'' accesses are damaged. 

Our framework comes with a JavaScript-based GUI. Included in a web
page, the interface shows all accessed paths as well as all incurred
contract violations. A heuristic allows us to generate short effect
descriptions from the gathered path sets using the approach reported
elsewhere
\cite{Heidegger:2012:APC:2103656.2103671,Heidegger:2010:CTJ:1894386.1894395}. 

\subsection{Limitations}

Because of the browser's sandbox, JSConTest2 cannot directly protect
DOM objects with access permission contracts. The security mechanism
forbids to replace the references to the \lstinline$window$ and
\lstinline$document$ objects by suitably contracted proxies. This deficiency can 
be partially addressed by embedding an entire script in a scope which substitutes the
global object by a suitable proxy.

The use of proxies for access control has one unfortunate consequence:
the equality operators \texttt{==} and \texttt{===} do not work
correctly, anymore. Dependening on the access path, the same target
object may have different access rights and hence distinct proxies
that enforce these rights. Comparing these distinct proxies returns
false even though the underlying target is the same. Similarly, an
unwrapped target object may be compared with its contracted version,
which should be true, but yields false.

Here is an example illustrating the problem.
\begin{lstlisting}
var ch = { c : 42 }
var root = __APC.permit (|'a.@+b.c'|, { a : ch, b : ch })
var same_acc = (root.a === root.b)
var same_unw = (ch === root.b)
\end{lstlisting}
With our implementation, both \lstinline|same_acc| and
\lstinline|same_unw| are false although they are true without the
\lstinline|permit| operation.

Unfortunately, there is no easy way to address this shortcoming. One
possibility is to assign each target a unique proxy, which requires a
potentially unintuitive merge of different access contract. Another
idea would be to trap the equality operation, which is not supported
by the proxy API. However, neither the unique proxy nor trapping the
equality operation would solve the problem with comparing the proxy with its
target as in \lstinline|same_unw| (just consider \texttt{===} as a
method call on the unwrapped target object \texttt{ch} in line~4). 

The best solution would be to provide two proxy-aware equality functions
and replace all uses of \texttt{==} and \texttt{===} by these
functions. This solution would require some light rewriting of the
source code (also at run time to support eval), which is much less
intrusive than the rewriting of the original JSConTest
implementation. 
Currently, we do not supply this rewriting because none of the
programs we examined in our evaluation were affected by the problem.

% _____                _   _           _   _____                 _ _       
%|  __ \              | | (_)         | | |  __ \               | | |      
%| |__) | __ __ _  ___| |_ _  ___ __ _| | | |__) |___  ___ _   _| | |_ ___ 
%|  ___/ '__/ _` |/ __| __| |/ __/ _` | | |  _  // _ \/ __| | | | | __/ __|
%| |   | | | (_| | (__| |_| | (_| (_| | | | | \ \  __/\__ \ |_| | | |_\__ \
%|_|   |_|  \__,_|\___|\__|_|\___\__,_|_| |_|  \_\___||___/\__,_|_|\__|___/

\section{Evaluation}
\label{sec:practical_results}

This section reports on our experiences with applying JSConTest2 to selected programs.
All benchmarks were run on a MacBook Pro with a 2 GHz
Intel Core i7 processor with 8 GB memory.  
All example runs and timings reported in
this paper were obtained with version 23.0a2 (2013-05-21) of the \emph{Firefox
Aurora} browser.  

\subsection{Benchmark Programs}
\label{sec:programs}

To evaluate our implementation, we applied it to a range of JavaScript
programs: the Google V8 Benchmark Suite\footnote{\webpageGoogleBenchmarks} and a
selection of benchmarks 
accompanying the TAJS system \cite{JensenMoellerThiemann2009}. 

The Google V8 Benchmark Suite consists of a webpage with several
JavaScript programs, which are listed in
Figure~\ref{fig:googleV8Runtime}. The benchmarks range from about 400
to 5000 lines of code implementing an OS kernel simulation, constraint
solving, encryption, ray tracing, parsing, regular expression
operations, benchmarking data structures, and solving differential
equations. The suite was originally composed to evaluate
the performance of JavaScript engines.  It is designed to stress
various aspects of the implementation of a JavaScript engine, but the
programs it contains are not necessarily representative of the typical
programs run in a browser.

The TAJS benchmarks consist of JavaScript programs and dumped Web pages collected
in the wild to test the static analysis system TAJS. To easily run the tests in
the Aurora web browser, we selected all programs that came packaged with a
webpage:  3dmodel, countdown, oryx, ajaxtabscontent, arkanoid, ball\_pool,
bunnyhunt, gamespot, google\_pacman, jscalc, jscrypto, logo, mceditor,
minesweeper, msie9, simple\_calc, wala. The selection further contains
programs like a calculator or a simple browser game 
as well as libraries extending the functionality of JavaScript like
jQuery, a linked list data type, or an MD5 hashing library. We also
applied our system to a number of dumped web pages like
\emph{youtube}, \emph{twitter}, or \emph{imageshack}. 

\subsection{Methodology}
\label{sec:methodology}

To evaluate our implementation with the Google Benchmarks, we manually examined
their source code, identified frequently used objects, and marked them with an
empty contract \lstinline$|@|$. Each access to those objects generated an
access violation, which was logged.

In the TAJS benchmarks, 
we looked for interesting objects and functions,
non-locally used data, and uses of external libraries like jQuery. In a first run, we
augmented these objects with a universal contract
(e.g. \lstinline{|?*|}) to monitor the accessed properties. Based on
the generated protocol we prepared customized contracts to protect
these objects. To exercise the customized contracts, we extended the source code
with additional, nonconforming operations to provoke violations.

% ___          __                                  ___                     _   
%| _ \___ _ _ / _|___ _ _ _ __  __ _ _ _  __ ___  |_ _|_ __  _ __  __ _ __| |_ 
%|  _/ -_) '_|  _/ _ \ '_| '  \/ _` | ' \/ _/ -_)  | || '  \| '_ \/ _` / _|  _|
%|_| \___|_| |_| \___/_| |_|_|_\__,_|_||_\__\___| |___|_|_|_| .__/\__,_\__|\__|
%                                                           |_|                

\subsection{Results}
\label{sec:performance_impact}

%\begin{figure}
%\centering
%\small
%\begin{tabular}{ l || l | l| l| l}
%\toprule
%\textbf{Benchmark}&
%\textbf{Full}&
%\textbf{Without}&
%\textbf{Contracts}&
%\textbf{Baseline}\\
%&
%&
%\textbf{logging}&
%\textbf{only}\\
%\midrule
%Richards& % \scriptsize(539 lines)&
%1h 22min&
%59min&
%3.8sec&
%2.6sec\\
%DeltaBlue& % \scriptsize(880 lines)&
%12.7sec&
%10.6sec&
%2.7sec&
%2.6sec\\
%Crypto& % \scriptsize(1689 lines)&
%3h 39min&
%2h 41min&
%2min&
%4.7sec\\
%RayTrace& % \scriptsize(3418 lines)&
%3min 2sec&
%2min 30sec&
%3.3sec&
%2.5sec\\
%EarleyBoyer& % \scriptsize(4682 lines)&
%4.8sec&
%4.7sec&
%4.6sec&
%4.8sec\\
%RegExp& % \scriptsize(1761 lines)&
%3.6sec&
%3.6sec&
%3.6sec&
%3.7sec\\
%Splay& % \scriptsize(394 lines)&
%?&
%?&
%9.6sec&
%2.8sec\\
%NavierStokes& % \scriptsize(387 lines)&
%2.5sec&
%2.5sec&
%2.5sec&
%2.5sec\\
%\bottomrule
%\end{tabular}
%\caption{Google V8 Benchmark Suite.}
%\label{fig:googleV8Runtime}
%\end{figure}

\begin{figure}
   \centering
   \small
   \begin{tabular}{ l || l | l| l| l}
	  \toprule
	  \textbf{Benchmark}&
	  \textbf{Full}&
	  \textbf{Without}&
	  \textbf{Contracts}&
	  \textbf{Baseline}\\
	  &
	  &
	  \textbf{logging}&
	  \textbf{only}\\
	  \midrule
	  Richards& % \scriptsize(539 lines)&
	  22.5min&
	  18.6min&
	  3.3sec&
	  2.3sec\\
	  DeltaBlue& % \scriptsize(880 lines)&
	  9.8sec&
	  9.5sec&
	  3.3sec&
	  2.3sec\\
	  Crypto& % \scriptsize(1689 lines)&
	  4.2h&
	  2.5h&
	  2.6min&
	  4.4sec\\
	  RayTrace& % \scriptsize(3418 lines)&
	  1.2h&
	  1.1h&
	  1.6min&
	  2.3sec\\
	  EarleyBoyer& % \scriptsize(4682 lines)&
	  4.4sec&
	  4.4sec&
	  4.4sec&
	  4.3sec\\
	  RegExp& % \scriptsize(1761 lines)&
	  2.4sec&
	  2.4sec&
	  2.4sec&
	  2.4sec\\
	  Splay& % \scriptsize(394 lines)&
	  -&
	  -&
	  2.3sec&
	  2.3sec\\
	  NavierStokes& % \scriptsize(387 lines)&
	  2.3sec&
	  2.3sec&
	  2.3sec&
	  2.3sec\\
	  \bottomrule
   \end{tabular}
   \caption{Google V8 Benchmark Suite.}
   \label{fig:googleV8Runtime}
\end{figure}

% ALL OBSERVER MODE
% PROTECTOR IS EASY ...
% Richards: Contract=@; Object=new Packet(null, ID_WORKER), KIND_WORK)
% DeltaBlue: Contract=@ ; Object=new Planer()
% Crypto: Contract=@; Object=new RSAKey(); in encrypt() and decrypt()
% RayTrace: Contract=@ ; Object=new Flog.RayTrace.Scene()
% EarleyBoyer: Contract=@; Object=var cost_earley
% RegEx: Contract@l Object=var regExBenchmark
% Splay: Contract=@ ; Object=var splayTree.root
% Navier Stokes: Contract=@ ; Object=new FluidiField()

With our initial implementation, contract enforcement for programs in the Google
V8 Benchmark Suite was not possible, because the browser quickly ran out of
memory. 
Our reimplementation based on the ideas described in Section~\ref{sec:reduction}
enabled us to cut down memory consumption dramatically.\footnote{Unfortunately,
we did not find a way to measure memory consumption.}
%The memory use of contract enforcement of a particular run stays almost constant.
The reduction in memory use of path logging comes at the expense of higher computational
cost. The reimplemented system successfully applies contract enforcement to
all programs in the Google V8 Benchmark Suite; for the \emph{Splay}
benchmark, we have no numbers for path set collection and logging
because it did not terminate  within four hours.

%Despite the reduction in memory use,
%some benchmarks produce a vast number of logged paths.  

Figure~\ref{fig:googleV8Runtime} contains the run times for all V8
benchmark programs in different configurations. The column \emph{Full} contains the run time for contract
enforcement, path set collection, and log output. The effect heuristic to
condense the resulting set of paths to a short effect description is disabled.
The column \emph{Without logging} shows the time used for contract enforcement
and path set collection, but without logging. The column \emph{Contracts only}
shows the time for contract enforcement, without any path set
generation. The last column \emph{Baseline} shows the baseline for a run without
JSConTest2.

Using forwarding proxies instead of normal objects did not have a measurable
effect: in addition to the above configurations, we ran the
benchmarks with forwarding proxies, where the handler intercepts all
operations but the trap functions forward the operation to
the target object, as shown in Figure~\ref{fig:prxie_pattern}.  The
resulting run times exhibit no measurable difference to the numbers in
column \emph{Baseline}. 

In most benchmarks, the run-time difference between contract
enforcement and the baseline is negligible, so monitoring is
cheap. The exceptions are \emph{Crypto} and \emph{RayTrace} where the contract is
applied to the main API object.

The run times for the programs in the V8 Benchmark Suite range from
few seconds up to four hours when running with contract monitoring fully engaged. This run
time depends on the objects chosen for contract monitoring: contracting heavily
used objects causes more overhead (viz.\ \emph{Crypto} and \emph{RayTrace}),
contracting the root of a tree (for example in \emph{Splay}) 
also causes overhead because the membrane implementation creates a shadow tree
populated with proxies while the program runs.

Unfortunately, the most expensive benchmark (\emph{Splay}) increases
the size of the trie structure in a way that the contract
implementation was not able to handle efficiently. Further
optimizations like condensing paths are required to run this benchmark
to completion. 

The second most expensive benchmark (\emph{Crypto}) produces more than 5GB output of logged paths,
depending on the selected object. For comparison, the benchmark
\emph{Richards} requires  
approximately 22.5 minutes to calculate slightly more than 1GB of logged paths. In \emph{Crypto}, a significant percentage
of the memory consumption and the computation time is due to the path
recovery at the end of the run. This mechanism flattens a trie structure to a
list of paths, which removes all sharing from the structure. It accounts for
much of the difference between columns \emph{Full} and \emph{Without logging} in
Figure~\ref{fig:googleV8Runtime}.

These examples show that the run time impact of monitoring is highly
dependent on the program and on the particular values that are
monitored. While some programs are heavily affected (\emph{Crypto},
\emph{Richards}, \emph{RayTrace}, \emph{Splay}), others are almost
unaffected: \emph{EarleyBoyer}, \emph{RegEx}, \emph{NavierStokes}. 

The numbers also show that the path logging accounts for most of the run-time overhead: 
the biggest fraction of the total run time is used for path generation, which comprises
appending of trie structures and merging tries. The remaining time is spent for path 
reconstruction, logging, and log output. 

The evaluation of contracts themselves is negligible in many cases, but
occasionally it may create an overhead of 35x (\emph{Crypto}) to 41x
(\emph{RayTrace}).  On the other hand, these programs are an artificial
selection to stress the JavaScript implementation.

The run times of the more realistic collection of TAJS benchmark
programs are all much shorter (less than one second) than the run
times for the V8 benchmarks. Furthermore, the difference between the
run times of the four configurations listed in
Figure~\ref{fig:googleV8Runtime} is negligible for the TAJS
benchmarks. These findings indicate that contract monitoring seems
feasible for realistic programs.

\subsection{General Observations}
\label{sec:general-observations}

The benchmarks show that the most time-consuming parts are path
logging and contract derivation. Their overhead is influenced by
several factors: the number and frequency of proxy calls and the length
of access chains. 

The length of the access chains determines the number of derivation steps and the size of the trie structure. Further, the number of nested proxies influences the number of merge operations and can cause a blowup of the data structures. Path extension and computing the contract derivative is more expensive on merged handlers. Also, the structure of the contract affects the performance. Wide and complex contracts require more derivation steps than deep contracts. In addition, a derived contract sometimes gets bigger than the original contract. For example, the contract $a \apcStar \apcOr a \apcStar\apcDot a \apcStar$ is the result of deriving  $a \apcStar\apcDot a \apcStar$ by $a$. All these factors contribute to the time and space complexity of contract monitoring.

% _____      _       _           _  __          __        _    
%|  __ \    | |     | |         | | \ \        / /       | |   
%| |__) |___| | __ _| |_ ___  __| |  \ \  /\  / /__  _ __| | __
%|  _  // _ \ |/ _` | __/ _ \/ _` |   \ \/  \/ / _ \| '__| |/ /
%| | \ \  __/ | (_| | ||  __/ (_| |    \  /\  / (_) | |  |   < 
%|_|  \_\___|_|\__,_|\__\___|\__,_|     \/  \/ \___/|_|  |_|\_\

\section{Related Work}
\label{sec:related_work}

JSConTest
\cite{Heidegger:2012:APC:2103656.2103671,Heidegger:2010:CTJ:1894386.1894395}
is a framework for logging side effects and enforcing path-based
access permission contracts. It comes with an algorithm
\cite{Heidegger:2011:HAC:2025896.2025908} that infers a concise effect
description from a set of access paths. Access permission contract
enable the specification of effects to restrict the access to the
object graph by defining a set of permitted access paths. JSConTest is
based on an offline code transformation.  Its implementation is
restricted to a subset of the language, it does not scale to large
programs, and it is hard to guarantee full interposition.

Access permission contracts are closely related to extended regular
expression. Permissions are computed from the iterated derivative of a
contracts by the current access path. Derivatives of extended regular
expressions and their properties are well known from the literature
\cite{Brzozowski1964,Owens:2009:RDR:1520284.1520288,Antimirov95partialderivatives}. Computing 
contract subsumption is related to solving regular expression
inequalities and checking regular expression equivalence, which has
been addressed in several places
\cite{Antimirov95:0,Henglein:2011:REC:1926385.1926429,komendantsky:inria-00614042}. Most
approaches rely on NFA checking
\cite{Bonchi:2013:CNE:2480359.2429124,Krauss:ProofPearl} or on
rewriting
\cite{Antimirov93rewritingextended,Almeida:RewriteRevisited,Viswanathan03testingextended}. For
checking contract subsumption, we adapted Antimirov's approach to
obtain a reasonably fast algorithm. We extended Antimirov's
algorithm to an infinite alphabet and to extended regular expressions including negation and intersection.

The JavaScript Reflection API
\cite{VanCutsem:2010:PDP:1869631.1869638,van2012design} enables
developers to easily enhance the functionality of objects and
functions. The implementation of proxies opens up the means to fully
interpose operations applied to objects and functions calls. Proxies
have already been used for dynamic effects systems
\cite{An:2011:DIS:1926385.1926437}. Other common uses for proxies,
e.g. \cite{Austin:2011:VVL:2048066.2048136, Brant98wrappersto,
   Bracha:2004:MDP:1035292.1029004, RobustComposition,
   Eugster:2006:UPJ:1167515.1167485,
%  Datta02proxy-basedacceleration,
Wernli:2012:OFC:2384577.2384589},
are meta-level extension, behavioral reflection, security, or
concurrency control.

There are further proposals to limit effects on heap-allocated objects
both statically and dynamically.
An \emph{effect system} is a static analysis that partitions the heap into disjoint regions and
annotates the type of a heap reference with the region in which the
reference points \cite{GiffordLucassen1986}.  Although initially
developed for functional languages, region-based effects have been
transposed to object-oriented languages \cite{GreenhouseBoyland1999}.
A notable proposal targeting Java is the type and effect system of DPJ
\cite{BocchinoAdveDigAdveHeumannKomuravelliOverbeySimmonsSungVakilian2009}.
DPJ targets parallel execution and provides by default a deterministic
semantics.

Also, specification languages like JML
\cite{LeavensBakerRuby1999,BurdyCheonCokErnstKiniryLeavensLeinoPoll2005}
include a mechanism for specifying side effects, the
\texttt{assignable} clause.
While the JML toolchain supports verification as well as run-time
monitoring \cite{LeavensCheonCliftonRubyCok2005,Cheon2003,Lehner2011},
\texttt{assignable} clauses are not widely used, partly because their semantics
has not been formally and unanimously defined until recently
\cite{Lehner2011}, and partly because support for \texttt{assignable} clauses is present
in only a few tools that perform run-time monitoring for JML
\cite{LehnerMueller2010} and then not always in full
generality\cite{Cheon2003}.

Our system may also be useful to guarantee security aspects like
confidentiality or integrity of information. In JavaScript, static
approaches are often lacking because of the dynamicity of the
language. However, the approaches range from static and dynamic
control of information flow control
\cite{Just:2011:IFA:2093328.2093331, hedin2012information,
Chugh:2009:SIF:1542476.1542483} over restricting the functionality
\cite{miller2008safe} to the isolations of scopes \cite{Phung:2012:TSA:2307720.2307721}.

% Data structures to improve access time and to take advantages of
% redundant information are already elaborated and versatile in
% use. Trie structures \cite{Fredkin:1960:TM:367390.367400,
% Aoe92anefficient} gained a lot of popularity. Modifications, like
% the reduction of unary branching nodes
% \cite{Morrison:1968:PAR:321479.321481,Szpankowski:1990:PTA:96559.214080,Clark:PATRICIAII}, reach a improvement of its performance properties. Such a kind of a trie structure is used for path logging to avoid an excessive memory bloat and to keep updating efficient. 

%  _____                 _           _             
% / ____|               | |         (_)            
%| |     ___  _ __   ___| |_   _ ___ _  ___  _ __  
%| |    / _ \| '_ \ / __| | | | / __| |/ _ \| '_ \ 
%| |___| (_) | | | | (__| | |_| \__ \ | (_) | | | |
% \_____\___/|_| |_|\___|_|\__,_|___/_|\___/|_| |_|

\section{Conclusion}
\label{sec:conclusion}

We successfully applied JavaScript proxies to the implementation of
effect logging and dynamic enforcement of access permission contracts,
which specify the allowed side effects using access paths in the
object graph.
The implementation avoids the shortcomings of an earlier
implementation in the JSConTest system, which is based on an offline
code transformation. The proxy-based approach handles the full
JavaScript language, including the \texttt{with}-statement,
\texttt{eval}, and arbitrary dynamic code loading techniques. Contrary
to the earlier implementation, the proxy-based approach guarantees
full interposition.

This reimplementation presents a major step towards practical
applicability of access permission contracts. The
run-time overhead and the additional memory consumption of pure
contract enforcement is negligible. 
Hence, we believe that this implementation can provide encapsulation in
realistic applications, as demonstrated with our examples and case studies.
Full effect logging, on the other
hand, incurs quite some overhead, but we regard it primarily as a tool
for program understanding and debugging.

%\acks
%Acknowledgments, if needed.

% We recommend abbrvnat bibliography style.

\bibliographystyle{abbrvnat}

\begin{thebibliography}{}
%\softraggedright
	  \providecommand{\natexlab}[1]{#1}
	  \providecommand{\url}[1]{\texttt{#1}}
	  \expandafter\ifx\csname urlstyle\endcsname\relax
	  \providecommand{\doi}[1]{doi: #1}\else
	  \providecommand{\doi}{doi: \begingroup \urlstyle{rm}\Url}\fi


	  \bibitem[Almeida et~al.(2009)Almeida, Moreira, and
	  Reis]{Almeida:RewriteRevisited}
	  M.~Almeida, N.~Moreira, and R.~Reis.
	  \newblock Antimirov and mosses's rewrite system revisited.
	  \newblock \emph{Int. J. Found. Comput. Sci.}, 20\penalty0 (4):\penalty0
	  669--684, 2009.

	  \bibitem[Antimirov(1995{\natexlab{a}})]{Antimirov95:0}
	  V.~M. Antimirov.
	  \newblock Rewriting regular inequalities (extended abstract).
	  \newblock In H.~Reichel, editor, \emph{FCT}, volume 965 of \emph{Lecture Notes
	  in Computer Science}, pages 116--125. Springer, 1995{\natexlab{a}}.

	  \bibitem[Antimirov(1995{\natexlab{b}})]{Antimirov95partialderivatives}
	  V.~M. Antimirov.
	  \newblock Partial derivates of regular expressions and finite automata
	  constructions.
	  \newblock In \emph{STACS}, pages 455--466, 1995{\natexlab{b}}.

	  \bibitem[Antimirov and Mosses(1993)]{Antimirov93rewritingextended}
	  V.~M. Antimirov and P.~D. Mosses.
	  \newblock Rewriting extended regular expressions.
	  \newblock In \emph{Developments in Language Theory}, pages 195--209, 1993.

	  \bibitem[Austin et~al.(2011)Austin, Disney, and
	  Flanagan]{Austin:2011:VVL:2048066.2048136}
	  T.~H. Austin, T.~Disney, and C.~Flanagan.
	  \newblock Virtual values for language extension.
	  \newblock In C.~V. Lopes and K.~Fisher, editors, \emph{OOPSLA}, pages 921--938,
	  Portland, OR, USA, 2011. ACM.
	  \newblock ISBN 978-1-4503-0940-0.

	  \bibitem[{Bocchino Jr.} et~al.(2009){Bocchino Jr.}, Adve, Dig, Adve, Heumann,
		 Komuravelli, Overbey, Simmons, Sung, and
	  Vakilian]{BocchinoAdveDigAdveHeumannKomuravelliOverbeySimmonsSungVakilian2009}
	  R.~L. {Bocchino Jr.}, V.~S. Adve, D.~Dig, S.~V. Adve, S.~Heumann,
	  R.~Komuravelli, J.~Overbey, P.~Simmons, H.~Sung, and M.~Vakilian.
	  \newblock A type and effect system for deterministic parallel {Java}.
	  \newblock In S.~Arora and G.~T. Leavens, editors, \emph{Proceedings of the 24th
		 {ACM} {SIGPLAN} Conference on Object Oriented Programming, Systems,
	  Languages, and Applications}, pages 97--116, Orlando, Florida, USA, 2009. ACM
	  Press, New York.
	  \newblock ISBN 978-1-60558-766-0.

	  \bibitem[Bonchi and Pous(2013)]{Bonchi:2013:CNE:2480359.2429124}
	  F.~Bonchi and D.~Pous.
	  \newblock Checking {NFA} equivalence with bisimulations up to congruence.
	  \newblock In R.~Giacobazzi and R.~Cousot, editors, \emph{POPL}, pages 457--468.
	  ACM, 2013.

	  \bibitem[Bracha and Ungar(2004)]{Bracha:2004:MDP:1035292.1029004}
	  G.~Bracha and D.~Ungar.
	  \newblock Mirrors: design principles for meta-level facilities of
	  object-oriented programming languages.
	  \newblock In J.~M. Vlissides and D.~C. Schmidt, editors, \emph{OOPSLA}, pages
	  331--344. ACM, 2004.

	  \bibitem[Brant et~al.(1998)Brant, Foote, Johnson, and
	  Roberts]{Brant98wrappersto}
	  J.~Brant, B.~Foote, R.~E. Johnson, and D.~Roberts.
	  \newblock Wrappers to the rescue.
	  \newblock In E.~Jul, editor, \emph{ECOOP}, volume 1445 of \emph{Lecture Notes
	  in Computer Science}, pages 396--417. Springer, 1998.

	  \bibitem[Brzozowski(1964)]{Brzozowski1964}
	  J.~A. Brzozowski.
	  \newblock Derivatives of regular expressions.
	  \newblock \emph{Journal of the {ACM}}, 11\penalty0 (4):\penalty0 481--494,
	  1964.

	  \bibitem[Burdy et~al.(2005)Burdy, Cheon, Cok, Ernst, Kiniry, Leavens, Leino,
	  and Poll]{BurdyCheonCokErnstKiniryLeavensLeinoPoll2005}
	  L.~Burdy, Y.~Cheon, D.~R. Cok, M.~D. Ernst, J.~R. Kiniry, G.~T. Leavens,
	  K.~R.~M. Leino, and E.~Poll.
	  \newblock An overview of {JML} tools and applications.
	  \newblock \emph{Int. J. Softw. Tools Technol. Transf.}, 7\penalty0
	  (3):\penalty0 212--232, 2005.
	  \newblock ISSN 1433-2779.
	  \newblock \doi{http://dx.doi.org/10.1007/s10009-004-0167-4}.

	  \bibitem[Cheon(2003)]{Cheon2003}
	  Y.~Cheon.
	  \newblock \emph{A Runtime Assertion Checker for the {Java} Modeling Language}.
	  \newblock PhD thesis, Iowa State University, Apr. 2003.
	  \newblock TR \#03-09.

	  \bibitem[Chugh et~al.(2009)Chugh, Meister, Jhala, and
	  Lerner]{Chugh:2009:SIF:1542476.1542483}
	  R.~Chugh, J.~A. Meister, R.~Jhala, and S.~Lerner.
	  \newblock Staged information flow for {JavaScript}.
	  \newblock In M.~Hind and A.~Diwan, editors, \emph{PLDI}, pages 50--62. ACM,
	  2009.

	  \bibitem[Cutsem and Miller(2010)]{VanCutsem:2010:PDP:1869631.1869638}
	  T.~V. Cutsem and M.~S. Miller.
	  \newblock Proxies: design principles for robust object-oriented intercession
	  {APIs}.
	  \newblock In W.~D. Clinger, editor, \emph{DLS}, pages 59--72. ACM, 2010.
	  \newblock ISBN 978-1-4503-0405-4.

	  \bibitem[Eugster(2006)]{Eugster:2006:UPJ:1167515.1167485}
	  P.~T. Eugster.
	  \newblock Uniform proxies for {Java}.
	  \newblock In P.~L. Tarr and W.~R. Cook, editors, \emph{OOPSLA}, pages 139--152.
	  ACM, 2006.

	  \bibitem[F{\"a}hndrich et~al.(2010)F{\"a}hndrich, Barnett, and
	  Logozzo]{FaehndrichBarnettLogozzo2010}
	  M.~F{\"a}hndrich, M.~Barnett, and F.~Logozzo.
	  \newblock Embedded contract languages.
	  \newblock In S.~Y. Shin, S.~Ossowski, M.~Schumacher, M.~J. Palakal, and C.-C.
	  Hung, editors, \emph{SAC}, pages 2103--2110, Sierre, Switzerland, 2010. ACM.
	  \newblock ISBN 978-1-60558-639-7.

	  \bibitem[Fredkin(1960)]{Fredkin:1960:TM:367390.367400}
	  E.~Fredkin.
	  \newblock Trie memory.
	  \newblock \emph{Commun. ACM}, 3\penalty0 (9):\penalty0 490--499, Sept. 1960.

	  \bibitem[Gifford and Lucassen(1986)]{GiffordLucassen1986}
	  D.~Gifford and J.~Lucassen.
	  \newblock Integrating functional and imperative programming.
	  \newblock In \emph{Proceedings of the 1986 {ACM} Conf. on {Lisp} and Functional
   Programming}, pages 28--38, Cambridge, Massachusetts, United States, 1986.
   {ACM} Press.

   \bibitem[Greenhouse and Boyland(1999)]{GreenhouseBoyland1999}
   A.~Greenhouse and J.~Boyland.
   \newblock An object-oriented effects system.
   \newblock In R.~Guerraoui, editor, \emph{13th European Conference on
   Object-Oriented Programming}, volume 1628 of \emph{Lecture Notes in Computer
   Science}, pages 205--229, Lisbon, Portugal, June 1999. Springer-Verlag.
   \newblock ISBN 3-540-66156-5.

   \bibitem[Guha et~al.(2010)Guha, Saftoiu, and
   Krishnamurthi]{Guha:2010:EJ:1883978.1883988}
   A.~Guha, C.~Saftoiu, and S.~Krishnamurthi.
   \newblock The essence of {JavaScript}.
   \newblock In T.~D'Hondt, editor, \emph{ECOOP}, volume 6183 of \emph{Lecture
   Notes in Computer Science}, pages 126--150. Springer, 2010.

   \bibitem[Hedin and Sabelfeld(2012)]{hedin2012information}
   D.~Hedin and A.~Sabelfeld.
   \newblock Information-flow security for a core of {JavaScript}.
   \newblock In S.~Chong, editor, \emph{CSF}, pages 3--18. IEEE, 2012.

   \bibitem[Heidegger and Thiemann(2010)]{Heidegger:2010:CTJ:1894386.1894395}
   P.~Heidegger and P.~Thiemann.
   \newblock Contract-driven testing of {JavaScript} code.
   \newblock In J.~Vitek, editor, \emph{TOOLS (48)}, volume 6141 of \emph{Lecture
   Notes in Computer Science}, pages 154--172, M{\'a}laga, Spain, June 2010.
   Springer.
   \newblock ISBN 978-3-642-13952-9.

   \bibitem[Heidegger and Thiemann(2011)]{Heidegger:2011:HAC:2025896.2025908}
   P.~Heidegger and P.~Thiemann.
   \newblock A heuristic approach for computing effects.
   \newblock In J.~Bishop and A.~Vallecillo, editors, \emph{TOOLS (49)}, volume
   6705 of \emph{Lecture Notes in Computer Science}, pages 147--162, Zurich,
   Switzerland, June 2011. Springer.
   \newblock ISBN 978-3-642-21951-1.

   \bibitem[Heidegger et~al.(2012)Heidegger, Bieniusa, and
   Thiemann]{Heidegger:2012:APC:2103656.2103671}
   P.~Heidegger, A.~Bieniusa, and P.~Thiemann.
   \newblock Access permission contracts for scripting languages.
   \newblock In \emph{POPL}, pages 111--122, Philadelphia, USA, Jan. 2012. {ACM}
   Press.

   \bibitem[Henglein and Nielsen(2011)]{Henglein:2011:REC:1926385.1926429}
   F.~Henglein and L.~Nielsen.
   \newblock Regular expression containment: coinductive axiomatization and
   computational interpretation.
   \newblock In  \citet{DBLP:conf/popl/2011}, pages 385--398.
   \newblock ISBN 978-1-4503-0490-0.

   \bibitem[hoon (David)~An et~al.(2011)hoon (David)~An, Chaudhuri, Foster, and
   Hicks]{An:2011:DIS:1926385.1926437}
   J.~hoon (David)~An, A.~Chaudhuri, J.~S. Foster, and M.~Hicks.
   \newblock Dynamic inference of static types for {Ruby}.
   \newblock In  \citet{DBLP:conf/popl/2011}, pages 459--472.
   \newblock ISBN 978-1-4503-0490-0.

   \bibitem[International(2009)]{ecma:262}
   E.~International.
   \newblock \emph{Standard ECMA-262}, volume~5.
   \newblock 2009.

   \bibitem[Jensen et~al.(2009)Jensen, M{\o}ller, and
   Thiemann]{JensenMoellerThiemann2009}
   S.~H. Jensen, A.~M{\o}ller, and P.~Thiemann.
   \newblock Type analysis for {JavaScript}.
   \newblock In \emph{Proc. 16th International Static Analysis Symposium,
   SAS~'09}, volume 5673 of \emph{Lecture Notes in Computer Science}, pages
   238--255, Los Angeles, CA, USA, Aug. 2009. Springer-Verlag.

   \bibitem[Just et~al.(2011)Just, Cleary, Shirley, and
   Hammer]{Just:2011:IFA:2093328.2093331}
   S.~Just, A.~Cleary, B.~Shirley, and C.~Hammer.
   \newblock Information flow analysis for javascript.
   \newblock In \emph{Proceedings of the 1st ACM SIGPLAN international workshop on
   Programming language and systems technologies for internet clients}, PLASTIC
   '11, pages 9--18, New York, NY, USA, 2011. ACM.

   \bibitem[Keil and Thiemann(2013)]{proxy2013:TechnicalReport}
   M.~Keil and P.~Thiemann.
   \newblock Efficient access analysis using {JavaScript} proxies.
   \newblock Technical report, Institute for Computer Science, University of
   Freiburg, 2013.

   \bibitem[Komendantsky(2011)]{komendantsky:inria-00614042}
   V.~Komendantsky.
   \newblock {Regular expression containment as a proof search problem}.
   \newblock In \emph{{PSATTT'11: International Workshop on Proof-Search in
   Axiomatic Theories and Type Theories}}, Wroclaw, Pologne, 2011. Germain
   Faure, St{\'e}phane Lengrand, Assia Mahboubi.

   \bibitem[Krauss and Nipkow(2012)]{Krauss:ProofPearl}
   A.~Krauss and T.~Nipkow.
   \newblock Proof pearl: Regular expression equivalence and relation algebra.
   \newblock \emph{J. Autom. Reasoning}, 49\penalty0 (1):\penalty0 95--106, 2012.

   \bibitem[Leavens et~al.(1999)Leavens, Baker, and Ruby]{LeavensBakerRuby1999}
   G.~T. Leavens, A.~L. Baker, and C.~Ruby.
   \newblock {JML}: A notation for detailed design.
   \newblock In H.~Kilov, B.~Rumpe, and I.~Simmonds, editors, \emph{{Behavioral
   Specifications of Businesses and Systems}}, pages 175--188, Norwell, MA, USA,
   1999. Kluwer Academic Publishers.

   \bibitem[Leavens et~al.(2005)Leavens, Cheon, Clifton, Ruby, and
   Cok]{LeavensCheonCliftonRubyCok2005}
   G.~T. Leavens, Y.~Cheon, C.~Clifton, C.~Ruby, and D.~R. Cok.
   \newblock How the design of {JML} accommodates both runtime assertion checking
   and formal verification.
   \newblock \emph{Science of Computer Programming}, 55\penalty0 (1-3):\penalty0
   185--208, 2005.

   \bibitem[Lehner(2011)]{Lehner2011}
   H.~Lehner.
   \newblock \emph{A Formal Definition of {JML} in {Coq} and its Application to
Runtime Assertion Checking}.
\newblock PhD thesis, ETH Zurich, Switzerland, 2011.

\bibitem[Lehner and M{\"u}ller(2010)]{LehnerMueller2010}
H.~Lehner and P.~M{\"u}ller.
\newblock Efficient runtime assertion checking of assignable clauses with
datagroups.
\newblock In D.~S. Rosenblum and G.~Taentzer, editors, \emph{FASE}, volume 6013
of \emph{Lecture Notes in Computer Science}, pages 338--352, Paphos, Cyprus,
2010. Springer.
\newblock ISBN 978-3-642-12028-2.

\bibitem[Miller(2006)]{RobustComposition}
M.~S. Miller.
\newblock \emph{Robust Composition: Towards a Unified Approach to Access
Control and Concurrency Control}.
\newblock PhD thesis, Johns Hopkins University, Baltimore, Maryland, USA, May
2006.

\bibitem[Miller et~al.(2008)Miller, Samuel, Laurie, Awad, and
Stay]{miller2008safe}
M.~S. Miller, M.~Samuel, B.~Laurie, I.~Awad, and M.~Stay.
\newblock Safe active content in sanitized {JavaScript}.
\newblock Technical report, Tech. Rep., Google, Inc, 2008.

\bibitem[Owens et~al.(2009)Owens, Reppy, and
Turon]{Owens:2009:RDR:1520284.1520288}
S.~Owens, J.~H. Reppy, and A.~Turon.
\newblock Regular-expression derivatives re-examined.
\newblock \emph{J. Funct. Program.}, 19\penalty0 (2):\penalty0 173--190, 2009.

\bibitem[Phung and Desmet(2012)]{Phung:2012:TSA:2307720.2307721}
P.~H. Phung and L.~Desmet.
\newblock A two-tier sandbox architecture for untrusted {JavaScript}.
\newblock In \emph{Proceedings of the Workshop on JavaScript Tools}, JSTools
'12, pages 1--10, New York, NY, USA, 2012. ACM.

\bibitem[POPL~2011()]{DBLP:conf/popl/2011}
POPL~2011.
\newblock \emph{Proceedings 38th Annual {ACM} Symposium on Principles of
Programming Languages}, Austin, TX, USA, Jan. 2011. {ACM} Press.
\newblock ISBN 978-1-4503-0490-0.

\bibitem[Rosu and Viswanathan(2003)]{Viswanathan03testingextended}
G.~Rosu and M.~Viswanathan.
\newblock Testing extended regular language membership incrementally by
rewriting.
\newblock In R.~Nieuwenhuis, editor, \emph{RTA}, volume 2706 of \emph{Lecture
Notes in Computer Science}, pages 499--514. Springer, 2003.

\bibitem[Van~Cutsem and Miller(2012)]{van2012design}
T.~Van~Cutsem and M.~S. Miller.
\newblock On the design of the {ECMAScript} reflection {API}.
\newblock Technical report, Technical Report VUB-SOFT-TR-12-03, Vrije
Universiteit Brussel, 2012.

\bibitem[Wernli et~al.(2012)Wernli, Maerki, and
Nierstrasz]{Wernli:2012:OFC:2384577.2384589}
E.~Wernli, P.~Maerki, and O.~Nierstrasz.
\newblock Ownership, filters and crossing handlers: flexible ownership in
dynamic languages.
\newblock In A.~Warth, editor, \emph{DLS}, pages 83--94. ACM, 2012.

\end{thebibliography}

% The bibliography should be embedded for final submission.

\appendix
%\newpage
\section{Crashing rules}\label{sec:crashing_rules}

\begin{figure}
   \centering
   \begin{mathpar}
	  \inferrule [\RuleLjOPropertyReferenceTwo]
	  {
		 \langle \ljLocation', \tpPath, \apcContract \rangle = \ljHeap(\ljLocation)\\
		 \isNotReadable{\apcContract}{\ljStr}\\\\
		 \ljHeap, \ljLocation' ~\entailsPR~ \ljStr ~\eval~ \ljHeap' ~|~ \ljConst ~|~ \monitor\\
	  }
	  {
		 \ljHeap,\ljLocation ~\entailsPR~ \ljStr ~\eval~ \ljHeap' ~|~ \ljConst ~|~ \addReadViolation{\tpPath\tpDot\ljStr}{\apcContract}
	  }\and
	  \inferrule [\RuleLjOPropertyReferenceThree]
	  {
		 \langle \ljLocation', \tpPath, \apcContract \rangle = \ljHeap(\ljLocation)\\
		 \isNotReadable{\apcContract}{\ljStr}\\\\
		 \ljHeap, \ljLocation' ~\entailsPR~ \ljStr ~\eval~ \ljHeap' ~|~ \ljLocation'' ~|~ \monitor\\\\
		 \ljProxy = \langle \ljLocation'', \tpPath\tpDot\ljStr, \deriv{\ljStr}{\apcContract}\rangle\\
		 \ljLocation''' \notin \dom(\ljHeap')
	  }
	  {
		 \ljHeap,\ljLocation ~\entailsPR~ \ljStr ~\eval~ \ljHeap'[\ljLocation'''\mapsto \ljProxy] ~|~ \ljLocation''' ~|~ \addReadViolation{\tpPath\tpDot\ljStr}{\apcContract}
	  }\and
	  \inferrule [\RuleLjOPropertyAssignmentTwo]
	  {
		 \langle \ljLocation', \tpPath, \apcContract \rangle = \ljHeap(\ljLocation)\\
		 \isNotWrireable{\apcContract}{\ljStr}\\\\
		 \ljHeap, \ljLocation' ~\entailsPA~ \ljStr, \ljVal ~\eval~ \ljHeap' ~|~ \ljVal' ~|~ \monitor
	  }
	  {
		 \ljHeap,\ljLocation ~\entailsPA~ \ljStr, \ljVal ~\eval~ \ljHeap' ~|~ \ljVal' ~|~ \addWriteViolation{\tpPath\tpDot\ljStr}{\apcContract}
	  }
   \end{mathpar}
   \caption{Inference rules for \emph{Observer Mode}.}
   \label{fig:inference_observer}
\end{figure}

\begin{figure}
   \centering
   \begin{mathpar}
	  \inferrule [\RuleLjPPropertyReferenceTwo]
	  {
		 \langle \ljLocation', \tpPath, \apcContract \rangle = \ljHeap(\ljLocation)\\
		 \isNotReadable{\apcContract}{\ljStr}
	  }
	  {
		 \ljHeap,\ljLocation ~\entailsPR~ \ljStr ~\eval~ \ljHeap' ~|~ \ljUndefined ~|~ \addReadViolation{\tpPath\tpDot\ljStr}{\apcContract}
	  }\and
	  \inferrule [\RuleLjPPropertyAssignmentTwo]
	  {
		 \langle \ljLocation', \tpPath, \apcContract \rangle = \ljHeap(\ljLocation)\\
		 \isNotWrireable{\apcContract}{\ljStr}
	  }
	  {
		 \ljHeap,\ljLocation ~\entailsPA~ \ljStr, \ljVal ~\eval~ \ljHeap' ~|~ \ljVal ~|~ \addWriteViolation{\tpPath\tpDot\ljStr}{\apcContract}
	  }
   \end{mathpar}
   \caption{Inference rules for \emph{Protector Mode}.}
   \label{fig:inference_protector}
\end{figure}

This section presents the formal semantics of path monitoring and contract enforcement in case of a violated contract.
The rules extend the set of inference rules of section \ref{sec:formalization}.

The implementation covers two different types of violation treatment. Figure \ref{fig:inference_observer} and \ref{fig:inference_protector} contain its evaluation rules. The \emph{Observer Mode} performs only path and violation logging without any interruption. $\addReadViolation{\tpPath}{\apcContract}$ and $\addWriteViolation{\tpPath}{\apcContract}$ extends the monitor to log a violation. The \emph{Protector Mode} returns \texttt{undefined} for an access instead of an abnormal termination or omits the operation.

\section{Extended Membrane}
\label{sec:extended_membrane}

\begin{figure}
   \centering
   \begin{mathpar}
	  \inferrule [\RuleLjTriePropertyReferenceTwo]
	  {
		 \langle \ljLocation', \tpTrie, \apcContract \rangle = \ljHeap(\ljLocation)\\
		 \isReadable{\apcContract}{\ljStr}\\\\
		 \ljHeap, \ljLocation' ~\entailsPR~ \ljStr ~\eval~ \ljHeap' ~|~ \ljConst ~|~ \monitor\\
	  }
	  {
		 \ljHeap,\ljLocation ~\entailsPR~ \ljStr ~\eval~ \ljHeap' ~|~ \ljConst ~|~ \addReadPath{\append{\tpTrie}{\ljStr}}
	  }\and
	  \inferrule [\RuleLjTriePropertyReferenceThree]
	  {
		 \langle \ljLocation', \tpTrie, \apcContract \rangle = \ljHeap(\ljLocation)\\
		 \isReadable{\apcContract}{\ljStr}\\\\
		 \ljHeap, \ljLocation' ~\entailsPR~ \ljStr ~\eval~ \ljHeap' ~|~ \ljLocation'' ~|~ \monitor\\
		 \ljLocation''=\langle \ljObj, \ljClosure, \ljPrototype \rangle\\\\
		 \ljProxy = \langle \ljLocation'', \append{\tpTrie}{\ljStr}, \deriv{\ljStr}{\apcContract}\rangle\\
		 \ljLocation''' \notin \dom(\ljHeap')
	  }
	  {
		 \ljHeap,\ljLocation ~\entailsPR~ \ljStr ~\eval~ \ljHeap'[\ljLocation'''\mapsto \ljProxy] ~|~ \ljLocation''' ~|~ \addReadPath{\append{\tpTrie}{\ljStr}}
	  }\and
	  \inferrule [\RuleLjTriePropertyReferenceFour]
	  {
		 \langle \ljLocation', \tpTrie, \apcContract \rangle = \ljHeap(\ljLocation)\\
		 \isReadable{\apcContract}{\ljStr}\\\\
		 \ljHeap, \ljLocation' ~\entailsPR~ \ljStr ~\eval~ \ljHeap' ~|~ \ljLocation'' ~|~ \monitor\\
		 \ljLocation''=\langle \ljLocation''', \tpTrie', \apcContract' \rangle\\\\
		 \ljProxy = \langle \ljLocation''', \union{\append{\tpTrie}{\ljStr}}{\tpTrie'}, \deriv{\ljStr}{\apcContract}\apcAnd\apcContract'\rangle\\
		 \ljLocation'''' \notin \dom(\ljHeap')
	  }
	  {
		 \ljHeap,\ljLocation ~\entailsPR~ \ljStr ~\eval~ \ljHeap'[\ljLocation''''\mapsto \ljProxy] ~|~ \ljLocation'''' ~|~ \addReadPath{\append{\tpTrie}{\ljStr}}
	  }\and
	  \inferrule [\RuleLjTriePropertyAssignmentTwo]
	  {
		 \langle \ljLocation', \tpTrie, \apcContract \rangle = \ljHeap(\ljLocation)\\
		 \isWrireable{\apcContract}{\ljStr}\\\\
		 \ljHeap, \ljLocation' ~\entailsPA~ \ljStr, \ljVal ~\eval~ \ljHeap' ~|~ \ljVal' ~|~ \monitor
	  }
	  {
		 \ljHeap,\ljLocation ~\entailsPA~ \ljStr, \ljVal ~\eval~ \ljHeap' ~|~ \ljVal' ~|~ \addWritePath{\append{\tpTrie}{\ljStr}}
	  }
   \end{mathpar}
   \caption{Inference rules for extended membranes.}
   \label{fig:inference_membrane}
\end{figure}

Section \ref{sec:reduction} introduces the necessity of merged proxies to avoid inefficient chains of nested proxy calls. 

Figure \ref{fig:inference_membrane} extends the inference rules from section \ref{sec:formalization} with path-tries and merged handlers. The single path $\tpPath$ in a handler $\ljHandler$ gets changed into a trie $\tpTrie$. We write $\tpPath\in\tpTrie$ if path $\tpPath$ is represented by $\tpTrie$. The operator $\oplus$ appends property $\tpProperty$ to all path-endings in trie $\tpTrie$. A trie $\tpTrie'=\append{\tpTrie}{(\tpProperty\tpDot\tpPath')}$ 
is equivalent to $\tpTrie'=\append{\append{\tpTrie}{\tpProperty}}{\tpPath'}$, appending $\tpProperty\tpDot\tpPath'$ to $\tpTrie$. $\union{\tpTrie}{\tpTrie'} = \append{\tpTrie}{\tpPath} ~|~ \forall\tpPath\in\tpTrie'$ denotes the union of the tries $\tpTrie$ and $\tpTrie'$.

Further, the definition of monitor $\monitor$ is extended by $\addReadPath{\tpTrie}$, extending monitor $\monitor$ with all paths $\tpPath\in\tpTrie$. The definitions of $\addWritePath{\tpTrie}$,  $\addReadViolation{\tpTrie}{\apcContract}$, and $\addWriteViolation{\tpTrie}{\apcContract}$ and the extensions to the crashing rules (Figure \ref{fig:inference_observer} ans \ref{fig:inference_protector}) are analogous.

\section{Auxiliary Functions}
\label{sec:auxiliaryfunctions}

This section contains the full definitions of the four auxiliary functions from section \ref{sec:containmelnt_semantics}.

\begin{figure}
   \centering
   \begin{minipage}{0.3\linewidth}
	  \begin{displaymath}
		 \begin{array}{lll}
			\isBlank{\apcAT} &=& \top\\
			\isBlank{\apcQMark} &=& \perp\\
			\isBlank{\apcRegEx} &=& \perp\\
			\isBlank{\apcNeg\apcRegEx} &=& \perp\\
			\isBlank{\apcEmptySet} &=& \perp
		 \end{array}
	  \end{displaymath}
   \end{minipage}
   \begin{minipage}{0.65\linewidth}
	  \begin{displaymath}
		 \begin{array}{lll}
			\isBlank{\apcEmpty} &=& \perp\\
			\isBlank{\apcContract\apcStar} &=& \perp\\
			\isBlank{\apcContract\apcOr\apcContract'} &=& \isBlank{\apcContract} \wedge \isBlank{\apcContract'}\\
			\isBlank{\apcContract\apcAnd\apcContract'} &=& \isBlank{\apcContract} \vee \isBlank{\apcContract'}\\
			\isBlank{\apcContract\apcDot\apcContract'} &=& \isBlank{\apcContract}
		 \end{array}
	  \end{displaymath}
   \end{minipage}
   \caption{The $\blankReducible$ function.}
   \label{fig:is-blank}
\end{figure}

\begin{figure}
   \centering
   \begin{minipage}{0.3\linewidth}
	  \begin{displaymath}
		 \begin{array}{lll}
			\isEmpty{\apcAT} &=& \perp\\
			\isEmpty{\apcQMark} &=& \perp\\
			\isEmpty{\apcRegEx} &=& \perp\\
			\isEmpty{\apcNeg\apcRegEx} &=& \perp\\
			\isEmpty{\apcEmptySet} &=& \top
		 \end{array}
	  \end{displaymath}
   \end{minipage}
   \begin{minipage}{0.65\linewidth}
	  \begin{displaymath}
		 \begin{array}{lll}
			\isEmpty{\apcEmpty} &=& \perp\\
			\isEmpty{\apcContract\apcStar} &=& \perp\\
			\isEmpty{\apcContract\apcOr\apcContract'} &=& \isEmpty{\apcContract} \vee \isEmpty{\apcContract'}\\
			\isEmpty{\apcContract\apcAnd\apcContract'} &=& \getFirstC{\apcContract\apcAnd\apcContract'}=\emptyset\\
			\isEmpty{\apcContract\apcDot\apcContract'} &=& \isEmpty{\apcContract} \vee \isEmpty{\apcContract'}\\
		 \end{array}
	  \end{displaymath}
   \end{minipage}
   \caption{The $\emptyReducible$ function.}
   \label{fig:is-empty}
\end{figure}

\begin{figure}
   \begin{minipage}[t]{0.35\linewidth}
	  \begin{displaymath}
		 \begin{array}{lll}
			\isIndifferent{\apcAT} &=& \perp\\
			\isIndifferent{\apcQMark} &=& \top\\
			\isIndifferent{\apcRegEx} &=& \perp\\
			\isIndifferent{\apcNeg\apcRegEx} &=& \perp\\
			\isIndifferent{\apcEmptySet} &=& \perp\\
			\isIndifferent{\apcEmpty} &=& \perp
		 \end{array}
	  \end{displaymath}
   \end{minipage}
   \begin{minipage}[t]{0.65\linewidth}
	  \begin{displaymath}
		 \begin{array}{lll}
			\isIndifferent{\apcContract\apcStar} &=& \isIndifferent{\apcContract}\\
			\isIndifferent{\apcContract\apcOr\apcContract'} &=& \isIndifferent{\apcContract} \vee \isIndifferent{\apcContract'}\\
			\isIndifferent{\apcContract\apcAnd\apcContract'} &=& \isIndifferent{\apcContract} \wedge \isIndifferent{\apcContract'}\\
			\isIndifferent{\apcContract\apcDot\apcContract'} &=& \begin{cases}
			   \top, &\apcContract=\apcEmpty \wedge \isIndifferent{\apcContract'}\\ 
			   \top, &\isIndifferent{\apcContract} \wedge \apcContract'=\apcEmpty\\
			   \perp, &otherwise
			\end{cases}\\
		 \end{array}
	  \end{displaymath}
   \end{minipage}
   \caption{The $\indifferentReducible$ function.}
   \label{fig:is-indifferent}
\end{figure}

\begin{figure}
   \centering
   \begin{minipage}[t]{0.3\linewidth}
	  \begin{displaymath}
		 \begin{array}{lll}
			\isUniversal{\apcAT} &=& \perp\\
			\isUniversal{\apcQMark} &=& \perp\\
			\isUniversal{\apcRegEx} &=& \perp\\
			\isUniversal{\apcNeg\apcRegEx} &=& \perp\\
			\isUniversal{\apcEmptySet} &=& \perp\\
			\isUniversal{\apcEmpty} &=& \perp
		 \end{array}
	  \end{displaymath}
   \end{minipage}
   \begin{minipage}[t]{0.65\linewidth}
	  \begin{displaymath}
		 \begin{array}{lll}
			\isUniversal{\apcContract\apcStar} &=& \isUniversal{\apcContract} \vee \isIndifferent{\apcContract}\\
			\isUniversal{\apcContract\apcOr\apcContract'} &=& \isUniversal{\apcContract} \vee \isUniversal{\apcContract'}\\
			\isUniversal{\apcContract\apcAnd\apcContract'} &=& \isUniversal{\apcContract} \wedge \isUniversal{\apcContract'}\\
			\isUniversal{\apcContract\apcDot\apcContract'} &=& \begin{cases}
			   \top, &\apcContract=\apcEmpty\wedge\isUniversal{\apcContract'}\\ 
			   \top, &\isUniversal{\apcContract}\wedge\apcContract'=\apcEmpty\\
			   \top, &\isUniversal{\apcContract}\wedge\isUniversal{\apcContract'}\\
			   \perp, &otherwise
			\end{cases}
		 \end{array}
	  \end{displaymath}
   \end{minipage}
   \caption{The $\universalReducible$ function.}
   \label{fig:is-universal}
\end{figure}

\begin{definition}[Blank] \label{def:isblank}
   A contract $\apcContract$ is \emph{blank} if $\lang{\apcContract}=\{\tpBlank\}$.
   The function $\blankReducible:\apcContract\rightarrow\{\top,\perp\}$ (Figure \ref{fig:is-blank}) checks if $\apcContract$ is blank.
\end{definition}

\begin{lemma}[Blank] $\isBlank{\apcContract}  =  \top  \Rightarrow \lang{\apcContract}=\tpBlank$
\end{lemma}

\begin{definition}[Empty] \label{def:isempty}
   A contract $\apcContract$ is \emph{empty} if $\lang{\apcContract}=\emptyset$.
   The function $\emptyReducible:\apcContract\rightarrow\{\top,\perp\}$ (Figure \ref{fig:is-empty}) checks if $\apcContract$ is empty.
\end{definition}

\begin{lemma}[Empty]$\isEmpty{\apcContract} = \top  \Rightarrow \lang{\apcContract}=\emptyset$
\end{lemma}

\begin{definition}[Indifferent] \label{def:isindifferent}
   A contract $\apcContract$ is \emph{indifferent} if $\lang{\apcContract}=\apcAlphabet$. 
   The function $\indifferentReducible:\apcContract\rightarrow\{\top,\perp\}$ (Figure \ref{fig:is-indifferent}) checks if $\apcContract$ is indifferent. 
\end{definition}

\begin{lemma}[Indifferent]$\isIndifferent{\apcContract} = \top  \Rightarrow \lang{\apcContract}=\apcAlphabet$
\end{lemma}

\begin{definition}[Universal] \label{def:isuniversal}
   A contract $\apcContract$ is \emph{universal} if $\lang{\apcContract}=\apcWords$.
   The function $\universalReducible:\apcContract\rightarrow\{\top,\perp\}$ (Figure \ref{fig:is-universal}) checks if $\apcContract$ is universal.
\end{definition}

\begin{lemma}[Universal]$\isUniversal{\apcContract} = \top  \Rightarrow \lang{\apcContract}=\apcWords$
\end{lemma}

\section{Semantic containment}
\label{sec:proof_semantic-containment}

\begin{proof}[Proof of Lemma \ref{thm:containment_nullable}]
   A contract $\apcContract$ is subset of another contract $\apcContract'$ iff for all paths $\tpPath\in\lang{\apcContract}$ the derivation of $\apcContract'$ w.r.t. path $\tpPath$ is nullable.
   For all $\tpPath\in\apcWords$  it holds that $\tpPath\in\lang{\apcContract'}$ iff $\isNullable{\deriv{\tpPath}{\apcContract'}}$. It is trivial to see that
   \begin{align}
	  &~ \isSubSetOf{\apcContract}{\apcContract'}\\
	  ~\Leftrightarrow&~ \lang{\apcContract}\subseteq\lang{\apcContract'}\\
	  ~\Leftrightarrow&~ \forall \tpPath\in\lang{\apcContract}: \tpPath\in\lang{\apcContract}\\
	  ~\Leftrightarrow&~ \forall \tpPath\in\lang{\apcContract}: \isNullable{\deriv{\tpPath}{\apcContract}}
   \end{align}
   holds.
\end{proof}

\begin{proof}[Proof of Lemma \ref{thm:containment2}]
   A contract $\apcContract$ is subset of another contract $\apcContract'$ iff for all properties $\tpProperty$ in $\getFirstP{\apcContract}$ the derivation of $\apcContract$ w.r.t. property $\tpProperty$ is subset of the derivation of $\apcContract'$ w.r.t. $\tpProperty$. 
   By lemma \ref{thm:derivation} we obtain
   \begin{align}
	  \lang{\deriv{\tpProperty}{\apcContract}} = \leftquotient{\tpProperty}{\lang{\apcContract}}	
   \end{align}
   and this leads to 
   \begin{align}
	  \{\apcEmpty ~|~ \isNullable{\apcContract} \} \cup
	  \{ \tpProperty\tpDot\tpPath ~|~ \tpProperty\in\getFirstP{\apcContract}, \tpPath\in\leftquotient{\tpProperty}{\lang{\apcContract}} \} &= \lang{\apcContract}
   \end{align}
   Claim holds because
   \begin{align}
	  &~ \isSubSetOf{\apcContract}{\apcContract'}\\
	  ~\Leftrightarrow&~ \lang{\apcContract}\subseteq\lang{\apcContract'}\\
	  ~\Leftrightarrow&~ \forall \tpPath\in\lang{\apcContract}:~ \tpPath\in\lang{\apcContract'}\\
	  ~\Leftrightarrow&~ \apcEmpty\in\lang{\apcContract} \Rightarrow \apcEmpty\in\lang{\apcContract'} ~\wedge\\
	  &~ \forall \tpProperty, \tpPath:~ \tpProperty.\tpPath\in\lang{\apcContract} \Rightarrow \isNullable{\deriv{\tpProperty.\tpPath}{\apcContract'}}\\
	  ~\Leftrightarrow&~ \isNullable{\apcContract} \Rightarrow \isNullable{\apcContract'} ~\wedge\\
	  &~ \forall \tpProperty\in\getFirstP{\apcContract}, \forall\tpPath:~ \tpProperty.\tpPath\in\lang{\apcContract} \Rightarrow \isNullable{\deriv{\tpPath}{\deriv{\tpProperty}{\apcContract'}}}\\
	  ~\Leftrightarrow&~ \isNullable{\apcContract} \Rightarrow \isNullable{\apcContract'} ~\wedge\\
	  &~ \forall \tpProperty\in\getFirstP{\apcContract}, \forall\tpPath \in\lang{\deriv{\tpProperty}{\apcContract}}:~ \isNullable{\deriv{\tpPath}{\deriv{\tpProperty}{\apcContract'}}}\\
	  ~\Leftrightarrow&~ \isNullable{\apcContract} \Rightarrow \isNullable{\apcContract'} ~\wedge\\
	  &~ \forall \tpProperty\in\getFirstP{\apcContract}:~ \lang{\deriv{\tpProperty}{\apcContract}}\subseteq\lang{\deriv{\tpProperty}{\apcContract'}}\\
	  ~\Leftrightarrow&~ \isNullable{\apcContract} \Rightarrow \isNullable{\apcContract'} ~\wedge\\
	  &~ \forall \tpProperty\in\getFirstP{\apcContract}:~ \isSubSetOf{\deriv{\tpProperty}{\apcContract}}{\deriv{\tpProperty}{\apcContract'}}
   \end{align}
\end{proof}

\section{Syntactic derivative}
\label{sec:proof_syntactic-derivative}

\newcommand{\regexStar}[1]{#1\apcStar}
\newcommand{\regexOr}[2]{#1 \apcOr #2}
\newcommand{\regexAnd}[2]{#1 \apcAnd #2}
\newcommand{\regexConcat}[2]{#1 \apcDot #2}

% ___               __ 
%| _ \_ _ ___  ___ / _|
%|  _/ '_/ _ \/ _ \  _|
%|_| |_| \___/\___/_|  
\begin{proof}[Proof of Lemma \ref{def:literal_derivation}]
   $\forall \apcLiteral:$
   \begin{gather}
	  \lang{\lderiv{\apcLiteral}{\apcContract}} ~\subseteq~ \bigcap_{\tpProperty\in\lang{\apcLiteral}} \lang{\deriv{\tpProperty}{\apcContract}}
   \end{gather}
   Proof by induction on $\apcContract$.
   \begin{description}
% ___            _        ___      _   
%| __|_ __  _ __| |_ _  _/ __| ___| |_ 
%| _|| '  \| '_ \  _| || \__ \/ -_)  _|
%|___|_|_|_| .__/\__|\_, |___/\___|\__|
%          |_|       |__/              
	  \item[Case] $\apcContract=\apcEmptySet$:~ Claim holds because $\lang{\lderiv{\apcLiteral}{\apcEmptySet}} = \lang{\deriv{\apcLiteral}{\apcEmptySet}} = \apcEmptySet$.

% ___            _      __      __          _ 
%| __|_ __  _ __| |_ _  \ \    / /__ _ _ __| |
%| _|| '  \| '_ \  _| || \ \/\/ / _ \ '_/ _` |
%|___|_|_|_| .__/\__|\_, |\_/\_/\___/_| \__,_|
%          |_|       |__/                     
	  \item[Case] $\apcContract=\apcEmpty$:~ Claim holds because $\lang{\lderiv{\apcLiteral}{\apcEmpty}} = \lang{\deriv{\apcLiteral}{\apcEmpty}} = \apcEmptySet$.

% ___            _         _      
%/ __|_  _ _ __ | |__  ___| |__ _ 
%\__ \ || | '  \| '_ \/ _ \ / _` |
%|___/\_, |_|_|_|_.__/\___/_\__,_|
%     |__/
	  \item[Case] $\apcContract=\apcRegEx$:~
		 \begin{description}
			\item[Subcase] $\regexIsSubSetOf{\apcLiteral}{\apcRegEx}$:~ Claim holds because\\
			   $\lang{\lderiv{\apcLiteral}{\apcRegEx}} = \lang{\deriv{\tpProperty}{\apcRegEx}} = \apcEmpty ~|~ \forall \tpProperty\in\lang{\apcLiteral}$.
			\item[Subcase] $\regexIsNotSubSetOf{\apcLiteral}{\apcRegEx}$:~ Claim holds because $\lang{\lderiv{\tpProperty}{\apcRegEx}} = \apcEmptySet$.
		 \end{description}
% ___ _            
%/ __| |_ __ _ _ _ 
%\__ \  _/ _` | '_|
%|___/\__\__,_|_|  
	  \item[Case] $\apcContract=\regexStar{\apcContract'}$:~ By induction 
		 \begin{align}
			&\lang{\lderiv{\apcLiteral}{\apcContract'}} ~\subseteqIH~ \bigcap_{\tpProperty\in\lang{\apcLiteral}} \lang{\deriv{\tpProperty}{\apcContract'}}
		 \end{align}
		 holds.We obtain that
		 \begin{align}
			&\forall\tpProperty:~ \lang{\deriv{\tpProperty}{\regexStar{\apcContract'}}} ~=~ \lang{\deriv{\tpProperty}{\apcContract'}\apcDot\regexStar{\apcContract'}}\\
			&\forall\apcLiteral:~ \lang{\lderiv{\apcLiteral}{\regexStar{\apcContract'}}} ~=~ \lang{\lderiv{\apcLiteral}{\apcContract'}\apcDot\regexStar{\apcContract'}}
		 \end{align}
		 holds. Claim holds because
		 \begin{align}
			&~ \forall\apcLiteral:~ \lang{\lderiv{\apcLiteral}{\apcContract'\apcStar}}\\
			~=&~ \lang{\lderiv{\apcLiteral}{\apcContract'}\apcDot\regexStar{\apcContract'}}\\
			~\subseteqIH&~ \{ \tpPath\apcDot\tpPath' ~|~ \tpPath\in\bigcap_{\tpProperty\in\lang{\apcLiteral}} \lang{\deriv{\tpProperty}{\apcContract'}}, \tpPath'\in\lang{\regexStar{\apcContract'}}\} \\
			~=&~ \bigcap_{\tpProperty\in\lang{\apcLiteral}} \{ \tpPath\apcDot\tpPath' ~|~ \tpPath\in\lang{\deriv{\tpProperty}{\apcContract'}}, \tpPath'\in\lang{\regexStar{\apcContract'}}\} \\
			~=&~ \bigcap_{\tpProperty\in\lang{\apcLiteral}} \lang{\deriv{\tpProperty}{\regexStar{\apcContract'}}}
		 \end{align}
%  ___      
% / _ \ _ _ 
%| (_) | '_|
% \___/|_|  
	  \item[Case] $\apcContract=\regexOr{\apcContract'}{\apcContract''}$:~ By induction 
		 \begin{align}
			&\lang{\lderiv{\apcLiteral}{\apcContract'}} ~\subseteqIH~ \bigcap_{\tpProperty\in\lang{\apcLiteral}} \lang{\deriv{\tpProperty}{\apcContract'}}\\
			&\lang{\lderiv{\apcLiteral}{\apcContract''}} ~\subseteqIH~ \bigcap_{\tpProperty\in\lang{\apcLiteral}} \lang{\deriv{\tpProperty}{\apcContract''}}
		 \end{align}
		 holds. We obtain that
		 \begin{align}
			&\lang{\deriv{\apcLiteral}{\regexOr{\apcContract'}{\apcContract''}}} ~=~ \lang{\deriv{\apcLiteral}{\apcContract'}}\cup\lang{\deriv{\apcLiteral}{\apcContract''}}\\
			&\lang{\lderiv{\apcLiteral}{\regexOr{\apcContract'}{\apcContract''}}} ~=~ \lang{\lderiv{\apcLiteral}{\apcContract'}}\cup\lang{\lderiv{\apcLiteral}{\apcContract''}}
		 \end{align}
		 holds. Claim holds because
		 \begin{align}
			&~ \lang{\lderiv{\apcLiteral}{\regexOr{\apcContract'}{\apcContract''}}}\\
			~=&~ \lang{\lderiv{\apcLiteral}{\apcContract'}}\cup\lang{\lderiv{\apcLiteral}{\apcContract''}}\\
			~\subseteqIH&~ \bigcap_{\tpProperty\in\lang{\apcLiteral}} \lang{\deriv{\tpProperty}{\apcContract'}} \cup \bigcap_{\tpProperty\in\lang{\apcLiteral}} \lang{\deriv{\tpProperty}{\apcContract''}}\\
			~\subseteq&~ \bigcap_{\tpProperty\in\lang{\apcLiteral}}  \lang{\deriv{\tpProperty}{\apcContract'}} \cup \lang{\deriv{\tpProperty}{\apcContract''}}\\
			~=&~ \bigcap_{\tpProperty\in\lang{\apcLiteral}}  \lang{\deriv{\tpProperty}{\regexOr{\apcContract'}{\apcContract''}}}			
		 \end{align}
%   _           _ 
%  /_\  _ _  __| |
% / _ \| ' \/ _` |
%/_/ \_\_||_\__,_|
	  \item[Case] $\apcContract=\regexAnd{\apcContract'}{\apcContract''}$:~ By induction
		 \begin{align}
			&\lang{\lderiv{\apcLiteral}{\apcContract'}} ~\subseteqIH~ \bigcap_{\tpProperty\in\lang{\apcLiteral}} \lang{\deriv{\tpProperty}{\apcContract'}}\\
			&\lang{\lderiv{\apcLiteral}{\apcContract''}} ~\subseteqIH~ \bigcap_{\tpProperty\in\lang{\apcLiteral}} \lang{\deriv{\tpProperty}{\apcContract''}}
		 \end{align}
		 holds. We obtain that
		 \begin{align}
			&\lang{\deriv{\apcLiteral}{\regexAnd{\apcContract'}{\apcContract''}}} ~=~ \lang{\deriv{\apcLiteral}{\apcContract'}}\cap\lang{\deriv{\apcLiteral}{\apcContract''}}\\
			&\lang{\lderiv{\apcLiteral}{\regexAnd{\apcContract'}{\apcContract''}}} ~=~ \lang{\lderiv{\apcLiteral}{\apcContract'}}\cap\lang{\lderiv{\apcLiteral}{\apcContract''}}
		 \end{align}
		 holds. Claim holds because
		 \begin{align}
			&~ \lang{\lderiv{\apcLiteral}{\regexAnd{\apcContract'}{\apcContract''}}}\\
			~=&~ \lang{\lderiv{\apcLiteral}{\apcContract'}}\cap\lang{\lderiv{\apcLiteral}{\apcContract''}}\\
			~\subseteqIH&~ \bigcap_{\tpProperty\in\lang{\apcLiteral}} \lang{\deriv{\tpProperty}{\apcContract'}} \cap \bigcap_{\tpProperty\in\lang{\apcLiteral}} \lang{\deriv{\tpProperty}{\apcContract''}}\\
			~=&~ \bigcap_{\tpProperty\in\lang{\apcLiteral}}  \lang{\deriv{\tpProperty}{\apcContract'}} \cap \lang{\deriv{\tpProperty}{\apcContract''}}\\
			~=&~ \bigcap_{\tpProperty\in\lang{\apcLiteral}}  \lang{\deriv{\tpProperty}{\regexAnd{\apcContract'}{\apcContract''}}}			
		 \end{align}
%  ___                  _   
% / __|___ _ _  __ __ _| |_ 
%| (__/ _ \ ' \/ _/ _` |  _|
% \___\___/_||_\__\__,_|\__|
	  \item[Case] $\apcContract=\regexConcat{\apcContract'}{\apcContract''}$:~ By induction
		 \begin{align}
			&\lang{\lderiv{\apcLiteral}{\apcContract'}} ~\subseteqIH~ \bigcap_{\tpProperty\in\lang{\apcLiteral}} \lang{\deriv{\tpProperty}{\apcContract'}}\\
			&\lang{\lderiv{\apcLiteral}{\apcContract''}} ~\subseteqIH~ \bigcap_{\tpProperty\in\lang{\apcLiteral}} \lang{\deriv{\tpProperty}{\apcContract''}}
		 \end{align}
		 holds.
		 \begin{description}
			\item[Subcase] $\isNullable{\apcContract}$:~
			   We obtain that 
			   \begin{align}
				  &\forall\tpProperty:~ \lang{\deriv{\apcLiteral}{\regexConcat{\apcContract'}{\apcContract''}}} ~=~
				  \lang{\deriv{\tpProperty}{\apcContract'}\apcDot\apcContract''} \cup \lang{\deriv{\tpProperty}{\apcContract''}}\\
				  &\forall\apcLiteral:~ \lang{\lderiv{\apcLiteral}{\regexConcat{\apcContract'}{\apcContract''}}} ~=~
				  \lang{\lderiv{\apcLiteral}{\apcContract'}\apcDot\apcContract''} \cup \lang{\lderiv{\apcLiteral}{\apcContract''}}
			   \end{align}
			   holds. Claim holds because
			   \begin{align}
				  &~ \lang{\lderiv{\apcLiteral}{\regexConcat{\apcContract'}{\apcContract''}}}\\
				  ~=&~ \lang{\lderiv{\apcLiteral}{\apcContract'}\apcDot\apcContract''} \cup \lang{\lderiv{\apcLiteral}{\apcContract''}}\\
				  ~\subseteqIH&~ \{ \tpPath\apcDot\tpPath' ~|~ \tpPath\in\bigcap_{\tpProperty\in\lang{\apcLiteral}} \lang{\deriv{\tpProperty}{\apcContract'}}, \tpPath'\in\lang{\apcContract''}\}\\
				  &~ \cup \bigcap_{\tpProperty\in\lang{\apcLiteral}} \lang{\deriv{\tpProperty}{\apcContract''}}\\
				  ~\subseteq&~ \bigcap_{\tpProperty\in\lang{\apcLiteral}}\{ \tpPath\apcDot\tpPath' ~|~ \lang{\deriv{\tpProperty}{\apcContract'}}, \tpPath'\in\lang{\apcContract''}\}\\
				  &~ \cup \lang{\deriv{\tpProperty}{\apcContract''}}\\
				  ~=&~ \bigcap_{\tpProperty\in\lang{\apcLiteral}} \lang{\deriv{\tpProperty}{\apcContract'}\apcDot\apcContract''} \cup \lang{\deriv{\tpProperty}{\apcContract''}}\\
				  ~=&~ \bigcap_{\tpProperty\in\lang{\apcLiteral}} \deriv{\tpProperty}{\apcContract'\apcDot\apcContract''}
			   \end{align}

			\item[Subcase] $\neg\isNullable{\apcContract}$:~
			   We obtain that 
			   \begin{align}
				  &\forall\tpProperty:~ \lang{\deriv{\apcLiteral}{\regexConcat{\apcContract'}{\apcContract''}}} ~=~ 
				  \lang{\deriv{\tpProperty}{\apcContract'}\apcDot\apcContract''}\\
				  &\forall\apcLiteral:~ \lang{\lderiv{\apcLiteral}{\regexConcat{\apcContract'}{\apcContract''}}} ~=~
				  \lang{\lderiv{\apcLiteral}{\apcContract'}\apcDot\apcContract''}
			   \end{align}
			   holds. Claim holds because
			   \begin{align}
				  &~ \lang{\lderiv{\apcLiteral}{\regexConcat{\apcContract'}{\apcContract''}}}\\
				  ~=&~ \lang{\lderiv{\apcLiteral}{\apcContract'}\apcDot\apcContract''}\\
				  ~\subseteqIH&~ \{ \tpPath\apcDot\tpPath' ~|~ \tpPath\in\bigcap_{\tpProperty\in\lang{\apcLiteral}} \lang{\deriv{\tpProperty}{\apcContract'}}, \tpPath'\in\lang{\apcContract''}\}\\
				  ~=&~ \bigcap_{\tpProperty\in\lang{\apcLiteral}}\{ \tpPath\apcDot\tpPath' ~|~ \tpPath\in\lang{\deriv{\tpProperty}{\apcContract'}}, \tpPath'\in\lang{\apcContract''}\}\\
				  ~=&~ \bigcap_{\tpProperty\in\lang{\apcLiteral}} \lang{\deriv{\tpProperty}{\apcContract'}\apcDot\apcContract''}\\
				  ~=&~ \bigcap_{\tpProperty\in\lang{\apcLiteral}} \deriv{\tpProperty}{\apcContract'\apcDot\apcContract''}
			   \end{align}
		 \end{description}
   \end{description}
\end{proof}

% ___               __ 
%| _ \_ _ ___  ___ / _|
%|  _/ '_/ _ \/ _ \  _|
%|_| |_| \___/\___/_|  
\begin{proof}[Proof of Lemma \ref{def:literal_derivation2}]
   $\forall \apcLiteral\in\getFirstC{\apcContract}:$
   \begin{gather}
	  \lang{\lderiv{\apcLiteral}{\apcContract}} ~=~ \bigcap_{\tpProperty\in\lang{\apcLiteral}} \lang{\deriv{\tpProperty}{\apcContract}}
   \end{gather}
   Suppose that $\forall\tpProperty,\tpProperty'\in\lang{\apcLiteral}:~ \deriv{\tpProperty}{\apcLiteral\apcDot\apcContract}=\deriv{\tpProperty'}{\apcLiteral\apcDot\apcContract}$.

   Proof by induction on $\apcContract$. The cases for $\apcEmptySet$, $\apcEmpty$, $\apcRegEx$, $\apcContract\apcAnd\apcContract'$ are analogous to the cases in the proof of lemma \ref{def:literal_derivation}. All occurences of $\subseteqIH$ can be replaced by $\eqIH$.

   \begin{description}
%  ___      
% / _ \ _ _ 
%| (_) | '_|
% \___/|_|  
	  \item[Case] $\apcContract=\regexOr{\apcContract'}{\apcContract''}$:~ By induction 
		 \begin{align}
			&\lang{\lderiv{\apcLiteral}{\apcContract'}} ~\eqIH~ \bigcap_{\tpProperty\in\lang{\apcLiteral}} \lang{\deriv{\tpProperty}{\apcContract'}}\\
			&\lang{\lderiv{\apcLiteral}{\apcContract''}} ~\eqIH~ \bigcap_{\tpProperty\in\lang{\apcLiteral}} \lang{\deriv{\tpProperty}{\apcContract''}}
		 \end{align}
		 holds. We obtain that
		 \begin{align}
			&\lang{\deriv{\apcLiteral}{\regexOr{\apcContract'}{\apcContract''}}} ~=~ \lang{\deriv{\apcLiteral}{\apcContract'}}\cup\lang{\deriv{\apcLiteral}{\apcContract''}}\\
			&\lang{\lderiv{\apcLiteral}{\regexOr{\apcContract'}{\apcContract''}}} ~=~ \lang{\lderiv{\apcLiteral}{\apcContract'}}\cup\lang{\lderiv{\apcLiteral}{\apcContract''}}
		 \end{align}
		 holds. Claim holds because
		 \begin{align}
			&~ \lang{\lderiv{\apcLiteral}{\regexOr{\apcContract'}{\apcContract''}}}\\
			~=&~ \lang{\lderiv{\apcLiteral}{\apcContract'}}\cup\lang{\lderiv{\apcLiteral}{\apcContract''}}\\
			~\eqIH&~ \bigcap_{\tpProperty\in\lang{\apcLiteral}} \lang{\deriv{\tpProperty}{\apcContract'}} \cup \bigcap_{\tpProperty\in\lang{\apcLiteral}} \lang{\deriv{\tpProperty}{\apcContract''}}\\
			~=&~ \bigcap_{\tpProperty\in\lang{\apcLiteral}}  \lang{\deriv{\tpProperty}{\apcContract'}} \cup \lang{\deriv{\tpProperty}{\apcContract''}}\\
			~=&~ \bigcap_{\tpProperty\in\lang{\apcLiteral}}  \lang{\deriv{\tpProperty}{\regexOr{\apcContract'}{\apcContract''}}}			
		 \end{align}
%  ___                  _   
% / __|___ _ _  __ __ _| |_ 
%| (__/ _ \ ' \/ _/ _` |  _|
% \___\___/_||_\__\__,_|\__|
	  \item[Case] $\apcContract=\regexConcat{\apcContract'}{\apcContract''}$:~ By induction
		 \begin{align}
			&\lang{\lderiv{\apcLiteral}{\apcContract'}} ~\eqIH~ \bigcap_{\tpProperty\in\lang{\apcLiteral}} \lang{\deriv{\tpProperty}{\apcContract'}}\\
			&\lang{\lderiv{\apcLiteral}{\apcContract''}} ~\eqIH~ \bigcap_{\tpProperty\in\lang{\apcLiteral}} \lang{\deriv{\tpProperty}{\apcContract''}}
		 \end{align}
		 holds.
		 \begin{description}
			\item[Subcase] $\isNullable{\apcContract}$:~
			   We obtain that 
			   \begin{align}
				  &\forall\tpProperty:~ \lang{\deriv{\apcLiteral}{\regexConcat{\apcContract'}{\apcContract''}}} ~=~
				  \lang{\deriv{\tpProperty}{\apcContract'}\apcDot\apcContract''} \cup \lang{\deriv{\tpProperty}{\apcContract''}}\\
				  &\forall\apcLiteral:~ \lang{\lderiv{\apcLiteral}{\regexConcat{\apcContract'}{\apcContract''}}} ~=~
				  \lang{\lderiv{\apcLiteral}{\apcContract'}\apcDot\apcContract''} \cup \lang{\lderiv{\apcLiteral}{\apcContract''}}
			   \end{align}
			   holds. Claim holds because
			   \begin{align}
				  &~ \lang{\lderiv{\apcLiteral}{\regexConcat{\apcContract'}{\apcContract''}}}\\
				  ~=&~ \lang{\lderiv{\apcLiteral}{\apcContract'}\apcDot\apcContract''} \cup \lang{\lderiv{\apcLiteral}{\apcContract''}}\\
				  ~\eqIH&~ \{ \tpPath\apcDot\tpPath ~|~ \tpPath\in\bigcap_{\tpProperty\in\lang{\apcLiteral}} \lang{\deriv{\tpProperty}{\apcContract'}}, \tpPath\in\lang{\apcContract''}\}\\
				  &~ \cup \bigcap_{\tpProperty\in\lang{\apcLiteral}} \lang{\deriv{\tpProperty}{\apcContract''}}\\
				  ~=&~ \bigcap_{\tpProperty\in\lang{\apcLiteral}}\{ \tpPath\apcDot\tpPath ~|~ \lang{\deriv{\tpProperty}{\apcContract'}}, \tpPath\in\lang{\apcContract''}\}\\
				  &~ \cup \lang{\deriv{\tpProperty}{\apcContract''}}\\
				  ~=&~ \bigcap_{\tpProperty\in\lang{\apcLiteral}} \lang{\deriv{\tpProperty}{\apcContract'}\apcDot\apcContract''} \cup \lang{\deriv{\tpProperty}{\apcContract''}}\\
				  ~=&~ \bigcap_{\tpProperty\in\lang{\apcLiteral}} \deriv{\tpProperty}{\apcContract'\apcDot\apcContract''}
			   \end{align}

			\item[Subcase] $\neg\isNullable{\apcContract}$:~
			   Analogus to the case in the proof of lemma \ref{def:literal_derivation}.
		 \end{description}
   \end{description}
\end{proof}

\section{Syntactic containment}
\label{sec:proof_syntactic-containment}

Before proving the syntactic containment we state an auxiliary lemma. For simplification, the literals $\apcAT$, $\apcQMark$, and $\apcNeg\apcRegEx$ are collapsed into a single regular expression literals $\apcRegEx$.

\begin{lemma}[Path-preservation]\label{thm:path-preservation} $\forall \tpProperty, \tpPath$:
   \begin{align}
	  \tpPath\in\lang{\deriv{\tpProperty}{\apcContract}} ~\Rightarrow~ \exists\apcLiteral\in\getFirstC{\apcContract}:~ \tpPath\in\lang{\lderiv{\apcLiteral}{\apcContract}}
   \end{align}
\end{lemma}

\begin{proof}[Proof of Lemma \ref{thm:path-preservation}]
   Suppose $\lang{\deriv{\tpProperty}{\apcContract}}\neq\apcEmptySet$.
   Show $\exists\apcLiteral\in\getFirstC{\apcContract}:~$ $\tpPath\in\lang{\lderiv{\apcLiteral}{\apcContract}}$.
   Proof by induction on $\apcContract$.

   \begin{description}
% ___            _        ___      _   
%| __|_ __  _ __| |_ _  _/ __| ___| |_ 
%| _|| '  \| '_ \  _| || \__ \/ -_)  _|
%|___|_|_|_| .__/\__|\_, |___/\___|\__|
%          |_|       |__/              
	  \item[Case] $\apcContract=\apcEmptySet$, $\getFirstC{\apcContract}=\{\apcEmptySet\}$: Contradicts assumption.
% ___            _         ___         _               _   
%| __|_ __  _ __| |_ _  _ / __|___ _ _| |_ __ _ _ _ __| |_ 
%| _|| '  \| '_ \  _| || | (__/ _ \ ' \  _/ _` | '_/ _|  _|
%|___|_|_|_| .__/\__|\_, |\___\___/_||_\__\__,_|_| \__|\__|
%          |_|       |__/                                  
	  \item[Case] $\apcContract=\apcEmpty$, $\getFirstC{\apcContract}=\{\}$: Contradicts assumption.
% ___          ___     
%| _ \___ __ _| __|_ __
%|   / -_) _` | _|\ \ /
%|_|_\___\__, |___/_\_\
%        |___/         
	  \item[Case] $\apcContract=\apcRegEx$, $\getFirstC{\apcContract}=\{\apcRegEx\}$:~\\
		 We obtain that $\tpProperty\in\lang{\apcRegEx} ~\Rightarrow~ \deriv{\tpProperty}{\apcRegEx}=\apcEmpty$. Claim holds because $\getFirstC{\apcContract}=\{\apcRegEx\}$, $\lderiv{\apcRegEx}{\apcRegEx}=\apcEmpty$, and thus $\tpPath=\tpEmpty$ and $\tpEmpty\in\lang{\apcEmpty}$.

% _  ___                  ___ _            
%| |/ / |___ ___ _ _  ___/ __| |_ __ _ _ _ 
%| ' <| / -_) -_) ' \/ -_)__ \  _/ _` | '_|
%|_|\_\_\___\___|_||_\___|___/\__\__,_|_|  
	  \item[Case] $\apcContract=\apcContract'\apcStar$, $\getFirstC{\apcContract}=\getFirstC{\apcContract'}$:~\\
		 We obtain that $\tpPath\in\lang{\deriv{\tpProperty}{\apcContract'\apcStar}}=\lang{\deriv{\tpProperty}{\apcContract'}\apcDot\apcContract'\apcStar}\neq\apcEmptySet$. By induction $\exists\apcLiteral'\in\getFirstC{\apcContract'}:~ \tpPath'\in\lang{\lderiv{\apcLiteral'}{\apcContract'}}$.
		 The chain holds because $\getFirstC{\apcContract'\apcStar}=\getFirstC{\apcContract'}$ and $\lderiv{\apcLiteral}{\apcContract'\apcStar}=\lderiv{\apcLiteral}{\apcContract'}\apcDot\apcContract'\apcStar$ and $\tpPath'\in\lang{\lderiv{\apcLiteral}{\apcContract'}}, \tpPath''\in\lang{\lderiv{\apcLiteral}{\apcContract'\apcStar}}$ implies $\tpPath=\tpPath'\tpDot\tpPath''\in\lang{\lderiv{\apcLiteral}{\apcContract'\apcStar}}$.

% _              _         _  ___      
%| |   ___  __ _(_)__ __ _| |/ _ \ _ _ 
%| |__/ _ \/ _` | / _/ _` | | (_) | '_|
%|____\___/\__, |_\__\__,_|_|\___/|_|  
%          |___/                       
	  \item[Case] $\apcContract=(\apcContract'\apcOr\apcContract'')$, $\getFirstC{\apcContract}=\getFirstC{\apcContract'}\cup\getFirstC{\apcContract''}$:~\\
		 We obtain that $\tpPath\in\lang{\deriv{\tpProperty}{\apcContract'\apcOr\apcContract''}}=\lang{\deriv{\tpProperty}{\apcContract'}}\cup\lang{\deriv{\tpProperty}{\apcContract''}}\neq\apcEmptySet$. 
		 By induction $\exists\apcLiteral'\in\getFirstC{\apcContract'}:~ \tpPath'\in\lang{\lderiv{\apcLiteral'}{\apcContract'}}$ and 
		 $\exists\apcLiteral''\in\getFirstC{\apcContract''}:~ \tpPath''\in\lang{\lderiv{\apcLiteral''}{\apcContract''}}$.
		 The chain holds because $\getFirstC{\apcContract'\apcOr\apcContract''}=\getFirstC{\apcContract'}\cup\getFirstC{\apcContract''}$ and $\lderiv{\apcLiteral}{\apcContract'\apcOr\apcContract''}=\lderiv{\apcLiteral}{\apcContract'}\apcOr\lderiv{\apcLiteral}{\apcContract''}$ and $\tpPath\in\lang{\lderiv{\apcLiteral}{\apcContract'}}$ or $\tpPath\in\lang{\lderiv{\apcLiteral}{\apcContract''}}$ implies $\tpPath\in\lang{\lderiv{\apcLiteral}{\apcContract'\apcOr\apcContract''}}$.

% _              _         _   _           _ 
%| |   ___  __ _(_)__ __ _| | /_\  _ _  __| |
%| |__/ _ \/ _` | / _/ _` | |/ _ \| ' \/ _` |
%|____\___/\__, |_\__\__,_|_/_/ \_\_||_\__,_|
%          |___/                             
	  \item[Case] $\apcContract=(\apcContract'\apcAnd\apcContract'')$, $\getFirstC{\apcContract}=\{ \regexIsCap{\apcLiteral'}{\apcLiteral''} ~|~ \apcLiteral'\in\getFirstC{\apcContract'}, \apcLiteral''\in\getFirstC{\apcContract''} \} $:~ \\
		 We obtain that $\tpPath\in\lang{\deriv{\tpProperty}{\apcContract'\apcAnd\apcContract''}}=\lang{\deriv{\tpProperty}{\apcContract'}}\cap\lang{\deriv{\tpProperty}{\apcContract''}}$ implies $\tpPath\in\lang{\deriv{\tpProperty}{\apcContract'}}$ and $\tpPath\in\lang{\deriv{\tpProperty}{\apcContract''}}$. 
		 By induction $\exists\apcLiteral'\in\getFirstC{\apcContract'}:~ \tpPath\in\lang{\lderiv{\apcLiteral}{\apcContract'}}$ and 
		 $\exists\apcLiteral''\in\getFirstC{\apcContract''}:~ \tpPath\in\lang{\lderiv{\apcLiteral''}{\apcContract''}}$.
		 Let $\apcLiteral=\regexIsCap{\apcLiteral'}{\apcLiteral''}\in\getFirstC{\apcContract'\apcAnd\apcContract''}$. If $\tpProperty\in\lang{\apcLiteral'}$ and $\tpProperty\in\lang{\apcLiteral'}$ then $\tpProperty\in\lang{\apcLiteral}$.
		 The chain holds because $\getFirstC{\apcContract'\apcAnd\apcContract''}=\{ \regexIsCap{\apcLiteral'}{\apcLiteral''} ~|~ \apcLiteral'\in\getFirstC{\apcContract'}, \apcLiteral''\in\getFirstC{\apcContract''} \}$ and $\lderiv{\apcLiteral}{\apcContract'\apcAnd\apcContract''}=\lderiv{\apcLiteral}{\apcContract'}\apcAnd\lderiv{\apcLiteral}{\apcContract''}$, and $\tpPath\in\lang{\lderiv{\apcLiteral}{\apcContract'}}$ and $\tpPath\in\lang{\lderiv{\apcLiteral}{\apcContract''}}$ implies $\tpPath\in\lang{\lderiv{\apcLiteral}{\apcContract'\apcAnd\apcContract''}}$.

%  ___                  _   
% / __|___ _ _  __ __ _| |_ 
%| (__/ _ \ ' \/ _/ _` |  _|
% \___\___/_||_\__\__,_|\__|
	  \item[Case] $\apcContract=(\apcContract'\apcDot\apcContract'')$:~
		 \begin{description}
			\item[Subcase] $\isNullable{\apcContract'}$, $\getFirstC{\apcContract}=\getFirstC{\apcContract'}\cup\getFirstC{\apcContract''}$:~\\
			   We obtain that $\tpPath\in\lang{\deriv{\tpProperty}{\apcContract'\apcDot\apcContract''}}=\lang{\deriv{\tpProperty}{\apcContract'}\apcDot\apcContract''}\cup\lang{\deriv{\tpProperty}{\apcContract''}}$ implies $\tpPath\in\lang{\deriv{\tpProperty}{\apcContract'}\apcDot\apcContract''}$ or $\tpPath\in\lang{\deriv{\tpProperty}{\apcContract''}}$.
			   By induction $\exists\apcLiteral'\in\getFirstC{\apcContract'}:~ \tpPath'\in\lang{\lderiv{\apcLiteral}{\apcContract'}}$ and 
			   $\exists\apcLiteral''\in\getFirstC{\apcContract''}:~ \tpPath''\in\lang{\lderiv{\apcLiteral''}{\apcContract''}}$.
			   The chain holds because $\getFirstC{\apcContract'\apcDot\apcContract''}= \getFirstC{\apcContract'}\cup\getFirstC{\apcContract''}$ and $\lderiv{\apcLiteral}{\apcContract'\apcDot\apcContract''}=(\lderiv{\apcLiteral}{\apcContract'}\apcDot\apcContract'')\apcOr\lderiv{\apcLiteral}{\apcContract''}$, and $\tpPath'\in\lang{\lderiv{\apcLiteral}{\apcContract'}}$ and $\tpPath''\in\lang{\apcContract''}$ implies $\tpPath=\tpPath'.\tpPath''\in\lang{\lderiv{\apcLiteral}{\apcContract'\apcDot\apcContract''}}$ or $\tpPath=\tpEmpty.\tpPath''\in\lang{\lderiv{\apcLiteral}{\apcContract'\apcDot\apcContract''}}$.

			\item[Subcase] $\neg\isNullable{\apcContract'}$, $\getFirstC{\apcContract}=\getFirstC{\apcContract'}$:~\\
			   We obtain that $\tpPath\in\lang{\deriv{\tpProperty}{\apcContract'\apcDot\apcContract''}}=\lang{\deriv{\tpProperty}{\apcContract'}\apcDot\apcContract''}$ implies $\tpPath\in\lang{\deriv{\tpProperty}{\apcContract'}\apcDot\apcContract''}$.
			   By induction $\exists\apcLiteral'\in\getFirstC{\apcContract'}:~ \tpPath'\in\lang{\lderiv{\apcLiteral}{\apcContract'}}$.
			   The chain holds because $\getFirstC{\apcContract'\apcDot\apcContract''}= \getFirstC{\apcContract'}$ and $\lderiv{\apcLiteral}{\apcContract'\apcDot\apcContract''}=\lderiv{\apcLiteral}{\apcContract'}\apcDot\apcContract''$, and $\tpPath'\in\lang{\lderiv{\apcLiteral}{\apcContract'}}$ and $\tpPath''\in\lang{\lderiv{\apcLiteral}{\apcContract''}}$ implies $\tpPath=\tpPath'.\tpPath''\in\lang{\lderiv{\apcLiteral}{\apcContract'\apcDot\apcContract''}}$.
		 \end{description}
   \end{description}
\end{proof}

\begin{proof}[Proof of Theorem \ref{thm:containment}]
   The proof is by contraposition. If $\isNotSubSetOf{\apcContract}{\apcContract'}$ then $\exists\apcLiteral\in\getFirstC{\apcContract}:~ \isNotSubSetOf{\lderiv{\apcLiteral}{\apcContract}}{\lderiv{\apcLiteral}{\apcContract'}}$ or $\neg(\isNullable{\apcContract}\Rightarrow\isNullable{\apcContract'})$.

   We obtain that:
   \begin{align}
	  \isNotSubSetOf{\apcContract}{\apcContract'} ~&\Leftrightarrow~ \lang{\apcContract}\nsubseteq\lang{\apcContract'}\\
	  ~&\Leftrightarrow~ \exists\tpPath\in\lang{\apcContract}\backslash\lang{\apcContract'}
   \end{align}

   \begin{description}
	  \item[Case] $\tpPath=\tpEmpty$:~ \\
		 Claim holds because $\neg(\isNullable{\apcContract}\Rightarrow\isNullable{\apcContract'})$.
	  \item[Case] $\tpPath\neq\tpEmpty$:~\\
		 It must be that $\tpPath=\tpProperty\tpDot\tpPath'$ with $\tpProperty\in\getFirstP{\apcContract}=\lang{\getFirstC{\apcContract}}$.
		 Therefore $\exists\apcLiteral\in\getFirstC{\apcContract}:~ \tpProperty\in\lang{\apcLiteral}$.
		 \begin{description}
			\item[Subcase] $\tpProperty\notin\getFirstP{\apcContract'}$:~\\
%			   We obtain $\neg\exists\apcLiteral'\in\getFirstC{\apcContract'}:~ \tpProperty\in\lang{\apcLiteral'}$. 
%			   It holds that $\exists\apcLiteral\in\getFirstC{\apcContract}:~ \tpProperty\in\lang{\apcLiteral}$ and $\deriv{\tpProperty}{\apcContract'}=\apcEmptySet$.
			   Claim holds by Lemma \ref{def:literal_derivation} and \ref{def:literal_derivation2} because $\exists\apcLiteral\in\getFirstC{\apcContract}:$ $\lderiv{\apcLiteral}{\apcContract}\neq\apcEmptySet$ and $\lderiv{\apcLiteral}{\apcContract'}=\apcEmptySet$ implies that $\isNotSubSetOf{\lderiv{\apcLiteral}{\apcContract}}{\lderiv{\apcLiteral}{\apcContract'}}$.

			\item[Subcase] $\tpProperty\in\getFirstP{\apcContract'}$:~\\ 
%			   We obtain $\exists\apcLiteral'\in\getFirstC{\apcContract'}:~ \tpProperty\in\lang{\apcLiteral'}$ and $\lang{\apcLiteral}\subseteq\lang{\apcLiteral'}$.
			   By Lemma \ref{def:literal_derivation} and \ref{def:literal_derivation2} claim holds because\\
			   \mbox{$\tpPath'\in\lang{\deriv{\tpProperty}{\apcContract}}\backslash\lang{\deriv{\tpProperty}{\apcContract'}}$} implies that\\
			   $\tpPath'\in\lang{\lderiv{\apcLiteral}{\apcContract}}\backslash\lang{\lderiv{\apcLiteral}{\apcContract'}}$

		 \end{description}
   \end{description}
\end{proof}

\section{Correctness}
\label{sec:proof_correctness}

\begin{proof}[Proof of Theorem \ref{thm:correctness}]
   If $\ccContext ~\entails~ \isSubSetOf{\apcContract}{\apcContract'} ~:~ \top$ than $\isSubSetOf{\apcContract}{\apcContract'}$
   Proof is by induction in the derivation of $\ccContext ~\entails~ \ccExp ~:~ \{\top, \perp\}$

   \begin{description}
	  \item[Case] \Rule{\RuleCCContext}:~\\
		 Obtaining the rule \Rule{\RuleCCUnfoldTrue} and \Rule{\RuleCCUnfoldFalse} the result of $\ccExp=\isSubSetOf{\apcContract}{\apcContract'}$ is the conjunction of its derivative w.r.t. the first literals. If $\inCcContext{\ccExp}$ than $\ccExp$ is already part of the conjunction.
	  \item[Case] \Rule{\RuleCCDisprove}, \Rule{\RuleCCUnfoldTrue}; \Rule{\RuleCCUnfoldFalse}:~\\
		 Claim holds by theorem \ref{thm:containment}
   \end{description}
\end{proof}

\end{document}